\documentclass{amsart}
\usepackage{amssymb,multirow,bbold}
\usepackage{color,comment,hyperref} 
\usepackage{amsmath,amsfonts} 
\usepackage{tikz}
\usetikzlibrary{arrows,decorations,decorations.markings}
\usepackage{fullpage}

\title{Fibre bundle framework\\
	for unitary quantum fault tolerance}
\author{Daniel Gottesman}
\address{Perimeter Institute for Theoretical Physics, Waterloo, Canada}
\email{dgottesman@perimeterinstitute.ca}

\author{Lucy Liuxuan Zhang}
\address{Department of Mathematics, University of Toronto, Toronto, Canada \\
\and Perimeter Institute for Theoretical Physics, Waterloo, Canada}
\email{lzhang@math.utoronto.ca}
    
\date{\today}

\def\flexsquare#1#2#3#4#5#6#7#8{
\tikzset{arrowlabel/.style={font=\scriptsize}}
  \node (top-left)     	at (-1.5, 1.5)	{$#1$};
  \node (top-right)	at ( 1.5, 1.5)	{$#2$};
  \node (bottom-left)  	at (-1.5,-1.5) 	{$#3$};
  \node (bottom-right) at ( 1.5,-1.5)	{$#4$};

  \draw[->] (top-left)		-- node[arrowlabel,above] {$#5$} 	(top-right);
  \draw[->] (bottom-left)	-- node[arrowlabel,above] {$#6$} 	(bottom-right);
  \draw[->] (top-left) 		-- node[arrowlabel,left] {$#7$} 		(bottom-left);
  \draw[->] (top-right) 	-- node[arrowlabel,right] {$#8$} 	(bottom-right);
}

\newtheorem{Theorem}{Theorem}[section]
\newtheorem{Proposition}[Theorem]{Proposition}
\newtheorem{Lemma}[Theorem]{Lemma}
\newtheorem{Corollary}[Theorem]{Corollary}
\newtheorem{Criterion}[Theorem]{Criterion}

\newtheorem{Conjecture}{Conjecture}

\theoremstyle{definition}
\newtheorem{Definition}{Definition}[section]

\newtheorem{Remark}{Remark}

\newcommand{\rmnum}[1]{\romannumeral #1}

\newcommand{\quotient}[2]{{\raisebox{.2em}{$#1$}\left/\raisebox{-.2em}{$#2$}\right.}}
\newcommand{\flatquotient}[2]{{#1}/{#2}}

\DeclareMathOperator{\tr}{Tr}
\newcommand{\Cvect}[1]{\mathbb{C}^{#1}}
\newcommand{\nqudits}{(\Cvect{d})^{\otimes n}}
\newcommand{\nqubits}{(\Cvect{2})^{\otimes n}}

\newcommand{\mqudits}{(\Cvect{d})^{\otimes m}}
\newcommand{\Hilbphy}{\mathcal{H}_{\text{physical}}}
\newcommand{\Hilblog}{\mathcal{H}_{\text{logical}}}
\newcommand{\Hilb}[1]{\mathcal{H}_{#1}}

\newcommand{\ordinal}[1]{
  \ifx#1=1 {#1}^{\text{st}}
  \else\ifx#1=2 {#1}^{\text{nd}}
       \else {#1}^{\text{th}}
       \fi
  \fi
}

\newcommand{\idmatrix}{\mathbb{1}}
\newcommand{\GL}[1]{\mathrm{GL}(#1,\mathbb{C})}
\newcommand{\U}[1]{{\mathcal U}(#1)}
\newcommand{\SU}[1]{{\mathcal{SU}}(#1)}
\newcommand{\Ub}[1]{{\mathcal U}[#1]}

\newcommand{\fE}{f^{*}E}

\newcommand{\TE}[1]{T_{#1}E}
\newcommand{\VE}[1]{V_{#1}E}
\newcommand{\HE}[1]{H_{#1}E}

\newcommand{\lgamma}{\tilde{\gamma}}

\newcommand{\projectiveE}{\tilde{E}}
\newcommand{\Cp}[1]{C_P(#1)} 
\newcommand{\Diff}[1]{\mathrm{Diff}(#1)}

\newcommand{\Hol}{\mathrm{Hol}}
\newcommand{\Grass}{\mathrm{Gr}(K,N)}
\newcommand{\subGrass}[2]{\mathrm{Gr}({#1},{#2})}

\newcommand{\bigvecbund}{\xi(K,N)}
\newcommand{\Mvecbund}{\xi(K,N) |_{\mathcal{M}}}
\newcommand{\bigPbund}{P(K,N)}
\newcommand{\MPbund}{P(K,N) |_{\mathcal{M}}}
\newcommand{\Stiefel}{V_{K}(\mathbb{C}^N)}

\newcommand{\Vperp}{\Ub{V^{\bot}}}
\newcommand{\VVperp}{\Ub{V} \oplus \Ub{V^{\bot}}}

\newcommand{\NmodVVperp}{\quotient{\U{N}}{\Ub{V} \oplus \Ub{V^{\bot}}}}
\newcommand{\flatNmodVVperp}{\flatquotient{\U{N}}{( \Ub{V} \oplus \Ub{V^{\bot}} )}}
\newcommand{\NmodVperp}{\quotient{\U{N}}{\Ub{V^{\bot}}}}
\newcommand{\flatNmodVperp}{\flatquotient{\U{N}}{\Ub{V^{\bot}}}}
\newcommand{\cosetVperp}[1]{{#1} \, \Ub{V^{\bot}}}
\newcommand{\rhoV}{\rho_{\xi}}
\newcommand{\rhoP}{\rho_{P}}
\newcommand{\GV}{\rhoV(\pi_1(\M, C))}
\newcommand{\GP}{\rhoP(\pi_1(\M, C))}

\newcommand{\A}{\mathcal{A}}
\newcommand{\AL}{\mathcal{A}_{\mathcal{L}}}
\newcommand{\Logical}{\mathcal{L}}
\newcommand{\Logicalperp}{\Logical(C^{\bot})}
\newcommand{\effLogical}{\overline{\Logical}}

\newcommand{\C}{\mathcal{C}}

\newcommand{\I}{\mathcal{I}}

\newcommand{\M}{\mathcal{M}}
\newcommand{\N}{\mathcal{N}}
\newcommand{\iotaM}{\iota_{\M}}

\newcommand{\fieldC}{\mathbb{C}}
\newcommand{\Cstar}{\mathbb{C}^{\times}}
\newcommand{\Errors}{\mathcal{E}}
\newcommand{\CorErrors}{\mathcal{E}'}
\newcommand{\Recovery}{\mathcal{R}}
\newcommand{\ket}[1]{|{#1}\rangle}
\newcommand{\bra}[1]{\langle{#1}|}

\newcommand{\supp}[1]{\mathrm{supp}(#1)}
\newcommand\restr[2]{{
  \left.\kern-\nulldelimiterspace 
  #1 
  \vphantom{\big|} 
  \right|_{#2} 
  }}

\newcommand{\timeorder}{\mathcal{T}}

\newcommand{\phases}{\{1, i, -1, -i\}}
\newcommand{\Paulis}{\{\idmatrix, X, Y, Z\}}
\newcommand{\locations}{(S_v, S_f)}
\newcommand{\locationsend}{(S'_v, S'_f)}

\newcommand{\Htoric}{H\locations}
\newcommand{\Ctoricfix}{C_{K}\locations}
\newcommand{\Ctoricfixend}{C_{K}\locationsend}
\newcommand{\Ctorict}[1]{C_{K}(S_v, S_f; S'_v, S'_f)({#1})}

\newcommand{\Ctorichardcore}{C_{K}^{HC, (n_v, n_f)}}
\newcommand{\Ctoricall}{C_{K}}

\newcommand{\Mgraph}{\mathring{\M}}
\newcommand{\toricvecbundgraph}{\bigvecbund |_{\Mgraph}}
\newcommand{\toricvecbundext}{\bigvecbund |_{\M}}
\newcommand{\toricPbundgraph}{\bigPbund |_{\Mgraph}}
\newcommand{\toricPbundext}{\bigPbund |_{\M}}
\newcommand{\conngraphtoric}{\nabla_{\mathrm{graph}}^{\mathrm{K}}}
\newcommand{\connexttoric}{\nabla_{\mathrm{ext}}^{\mathrm{K}}}

\newcommand{\sep}{s}

\newcommand{\vspan}[1]{\mathrm{span}(#1)}
\newcommand{\ord}[1]{{#1}^{\mathrm{th}}}

\newcommand{\F}{\mathcal{F}}
\newcommand{\Fsmall}{\mathring{\F}}
\newcommand{\Fpath}{\Fsmall^{p}}
\newcommand{\FsmallL}{\Fsmall_{\mathcal{L}}}
\newcommand{\Floop}{{\Fsmall}^{l}}

\newcommand{\Fbig}{\F}
\newcommand{\Fbigpaths}{\Fbig^{p}}

\newcommand{\Ftranspath}{\F^{p}}
\newcommand{\FL}{\F_{\mathcal{L}}}
\newcommand{\FLperp}{\FL(C^{\bot})}
\newcommand{\effFL}{\overline{\FL}}

\newcommand{\Fdiscrete}{\F_{\mathrm{discr}}}
\newcommand{\FdiscreteBeginend}{\Fdiscrete^{p}(S_v, S_f; S'_v, S'_f)}
\newcommand{\Fdiscretepaths}{\Fdiscrete^{p}}
\newcommand{\Fgraph}{\F_{\mathrm{graph}}}
\newcommand{\Fgraphpaths}{\Fgraph^{p}}
\newcommand{\Fext}{\F_{\mathrm{ext}}}
\newcommand{\Fextpaths}{\Fext^{p}}

\newcommand{\ver}[1]{\overline{{#1}}}

\begin{document}

\begin{abstract}
We introduce a differential geometric framework for describing families of quantum error-correcting codes and for understanding quantum fault tolerance.  In this paper we show that topological fault tolerance can be discussed in the same language used for fault tolerance in other kinds of quantum error-correcting codes.  In particular, we use fibre bundles with a flat projective connection to study the transformation of codewords under unitary fault-tolerant evolutions.  We explore in detail two examples of fault-tolerant families of operations, the qudit transversal gates and string operators in the toric code and other two-dimensional anyonic models. For these examples, we show that the fault-tolerant logical operations are given by the monodromy group for either of two bundles, both of which have flat projective connections.  We also outline a program which aims to ultimately unify all fault-tolerant protocols into a single framework generalizing that used in the examples.

\end{abstract}

\maketitle
\tableofcontents

\newpage

\section{Introduction}


In order to build a scalable quantum computer, some form of fault tolerance will almost certainly be needed.  Decoherence and other forms of error are likely to be much more serious for a quantum computer than for a classical computer, in part because quantum computers are built out of smaller components, and hence are more fragile, and in part because quantum states are simply inherently more delicate than classical states.

A fault-tolerant protocol~\cite{QECCintro} consists of a quantum error-correcting code (QECC) together with a set of techniques, called \emph{gadgets}, which allow us to perform a universal set of gates without compromising the protection of the code.  In particular, there are two main families of fault-tolerant protocols known, which we study here using our framework.  One uses a stabilizer code\footnote{However, the results in Section \ref{sec:transversal} of this work apply not only to stabilizer codes, but to arbitrary codes.} and performs logical gates largely using \emph{transversal operations} (gates which can be decomposed into a tensor product over the qudits in the code), supplementing with a few gadgets based on a modified form of teleportation (``magic state'' constructions) in order to produce a universal set of fault-tolerant gates.  The other family uses a topological code and performs gates using topologically robust gates, for instance by braiding defects in the lattice of qubits making up the code.  In some models, braiding is sufficient to produce a universal gate set, while other models can be made universal by augmenting the set of operations with magic state constructions.  There are a few other fault-tolerant protocols which do not fit neatly into either of these categories.  We do not study these miscellaneous protocols in this paper, but we expect similar principles to apply to them.

On the face of it, the techniques used in the two main families of fault-tolerant protocols seem very different.  We argue to the contrary that there is a natural picture which unifies topological fault tolerance with fault tolerance based on transversal gates.  In particular, both can be understood as fundamentally a topological effect.  By choosing the correct geometric space to study, we can view transversal gates as the result of moving along a topologically non-trivial path in that space, much as gates for a topological code result from moving along a path in configuration space.

We view this paper as the first step in a program which could ultimately unify all notions of fault-tolerant quantum computation.  Fault-tolerant constructions have a deserved reputation as being rather complicated and utilizing an ad hoc assortment of different tricks, and a conceptual simplification and unification of the principles of fault tolerance would be very welcome.  The program could lead to improved fault-tolerant protocols or general results about fault tolerance because of deeper insights into the structure of fault tolerance.  The work also helps put work on fault-tolerant quantum computation into a broader mathematical context which could suggest interesting future avenues of mathematical development.  

In addition, some of the techniques we develop along the way may be helpful in describing other quantum systems.  In particular, our construction for the toric code involves using superposition to simulate a continuous geometric space, allowing us to interpret a family of finite-dimensional subspaces as a continuous configuration space.  That is, the continuous nature of quantum amplitudes substitutes for continuous variation in position.

For much of this paper, we shall focus on building the mathematical framework and the geometrical pictures needed for this program.  This picture gives us an alternative, geometric and quite global way of looking at quantum error-correcting codes and quantum fault tolerance.


Roughly speaking, the main point of the paper is:

\begin{Conjecture}
Fault-tolerant logical gates can always be expressed as arising from monodromies of an appropriate fibre bundle with a flat projective connection.
\label{conj:generalFTtopo}
\end{Conjecture}

For the bulk of this paper, we will focus on \emph{unitary} fault-tolerant gates, for which the above conjecture specializes to and is stated more precisely as Conjecture~\ref{conj:flatconnection}.  We have so far only proven Conjecture \ref{conj:generalFTtopo} or Conjecture \ref{conj:flatconnection} in the cases of transversal gates and a family of topological codes which includes surface codes.  The result for transversal gates is a consequence of Theorem~\ref{thm:flatconnection.transversal}.  For toric codes with a fixed number of quasiparticles or defects, it follows from Theorem~\ref{thm:toricflat}.  For the standard toric code, few interesting operations can be done by braiding defects.  We show how to generalize the construction for the toric code to a variety of more general topological codes and anyon models, for which braiding of quasiparticles is more powerful.

\subsection{The basic idea}

For those already possessing the requisite mathematical background, we give a quick sketch of the conceptual aspects of our picture.  The Grassmannian $\Grass$ is the set of subspaces of a given size of a Hilbert space, where $K$ is the dimension of the subspace and $N$ the dimension of the Hilbert space.  $\Grass$ can be thought of as the manifold of quantum error-correcting codes.  Performing a continuous-time unitary evolution corresponds to a path in $\Grass$.  There are two natural fibre bundles over $\Grass$, with the fibre over a point $C$ equal to either the set of vectors in the subspace $C$ or the set of orthonormal $K$-frames for $C$.  Unitary evolution naturally transforms elements of the fibre, but unfortunately, does not in general give us a well-defined notion of parallel transport in the bundle over $\Grass$, because different unitary evolutions may give rise to the same path in $\Grass$, potentially leading to different paths or ``lifts'' in the bundle.  However, by restricting to a subset $\M$ of $\Grass$ (consisting of certain codes with a given quantum error correction property) and to a set of \emph{fault-tolerant} unitary evolutions, a notion of parallel transport \emph{is} well-defined.  Furthermore, in such cases, we get a natural flat projective connection on the fibre bundle being considered.  Thus, due to the projective flatness, fault-tolerant gates performing a non-trivial logical operation correspond to homotopically non-trivial loops on $\M$.

For fault-tolerant protocols based on transversal gates, we let $\M$ be the set of subspaces which can be reached from a given code $C$ by performing a tensor product of single-qudit unitary gates.  For the toric code with a fixed number of defects, we define an $\M$ that is isomorphic to the configuration space of the defects with a hard-core condition that prevents the defects from getting too close to each other.  The standard description of a toric code only allows defects at a discrete set of locations (vertices of a lattice or a dual lattice), but we demonstrate how to interpolate between such codes, and eventually define a continuous topological space $\M$ of defect configurations.  The construction for the toric code can be generalized to other models with localized anyonic quasiparticles, including, we believe, the Kitaev quantum double models~\cite{Kitaev1} and Levin-Wen string net models~\cite{LW}.

\subsection{Comparison with prior work}

While there has been quite a lot of work on fault-tolerance with transversal gates (see, e.g., references in \cite{QECCintro}) and on topological codes (for instance, \cite{Dennis,Kitaev1} and many subsequent papers), we are not aware of any previous work that attempts to unify the two approaches into a single picture.  

The approach we use for transversal gates breaks up the Grassmannian into submanifolds which are orbits of local unitaries (which, in this context, means tensor products of unitaries on the individual qudits).  The idea of classifying quantum error-correcting codes up to local unitary equivalence has appeared before, in the notion of \emph{invariants} of quantum codes.  (See, for instance, \cite{Rains_invariants}.)  The use of local unitary invariants has also become a standard approach for analyzing multi-particle entanglement~\cite{GRB}.  Both of these lines of work certainly have some connection to our approach, but the focus in prior work has been on identifying invariants, frequently polynomial invariants, that distinguish different classes of states or codes, whereas we are more interested in specific topological properties of a single class.  Polynomial invariants and the associated techniques are thus not useful for our purposes.  Our main theorem for transversal gates makes use of the results of \cite{EastinKnill}.

For topological codes, we work by considering a space which is isomorphic to the space of configurations of particles on a torus.  When the particles are located on lattice sites (or dual lattice sites for the dual defects), the codes we get are completely standard~\cite{Dennis,Kitaev1}.  The novelty in our construction is to give a description of the configuration space as sitting inside $\Grass$, by assigning a topological quantum error-correcting code to every configuration of defects on a torus, including those configurations which have defects at locations other than lattice sites (or dual lattice sites for dual defects), subject only to a hard-core condition that pairs of defects are not too close together.  This trick may have other applications in the study of topological codes and topologically ordered systems.

\subsection{Plan of the paper}

We have included a good deal of introductory material in order to make the paper somewhat self-contained for people from different backgrounds.  Sections \ref{sec:introFT}--\ref{sec:geoUnitary} provide background material and definitions we will use later, as well as providing partial motivation for some of the choices we make later in the paper.  We begin with a brief introduction to quantum fault tolerance in Section \ref{sec:introFT}, portraying the interplay between its essential ingredients.  
In Section \ref{sec:geometry}, we introduce the abstract geometric notions we shall use in this paper. We establish a geometric picture for unitary evolutions of QECCs (of any fixed dimension) in Section \ref{sec:geoUnitary}.  Section~\ref{sec:geoUnitary} does not contain any original mathematical content, but does start to illustrate the connections between our picture and existing concepts; in particular, it explains how we can understand the Grassmannian as the manifold of quantum error-correcting codes of a given size.

The new material begins in Section~\ref{sec:geoFT}, where we summarize our program to unify the principles of unitary fault tolerance.  Sections \ref{sec:transversal} and \ref{sec:topologicalFT} present two detailed worked-out examples of the framework (transversal gates and topological codes), proving special cases of Conjecture \ref{conj:flatconnection}.  In Section \ref{sec:fullFT}, we sketch how to push the present framework further to include non-unitary fault-tolerant operations, such as those which would change the dimension of the physical Hilbert space.  We conclude in Section~\ref{sec:conclusion} with a discussion of some possible future directions and open problems.

\section{Quantum information preliminaries: quantum fault tolerance}
\label{sec:introFT}

This section contains some of the concepts and lemmas from quantum information which we shall need to build or understand our geometric picture.

In short, quantum fault tolerance refers to a framework for designing quantum computation procedures which are robust against physical errors.  There are three main ingredients that go into building a fault-tolerant protocol: error models, quantum error-correcting codes (QECCs), and fault-tolerant gadgets.  The three components, as pictured in Figure \ref{fig:FTaspects}, must be chosen to work together properly for the protocol to be fault tolerant.  These three ingredients are addressed in Sections~\ref{sec:errormodels} -- \ref{sec:FToperations}, and then we give two examples in Sections~\ref{sec:squdit_errors} and \ref{sec:local_errors} to illustrate the desired compatibilities.

\begin{figure}
\begin{tikzpicture}
\path (-2,2) coordinate (EM);
\draw (EM) circle (1);
\node [anchor=south] at (EM) {Error};
\node [anchor=north] at (EM) {Model};

\path (2,2) coordinate (QECC);
\draw (QECC) circle (1);
\node at (QECC) {QECC};

\path (-2,2) ++(300:4) coordinate (FT);
\draw (FT) circle (1);
\node [anchor=south] at (FT) {FT};
\node [anchor=north] at (FT) {Operations};

\draw [thick,red] (EM) ++(300:1) -- +(300:2);
\draw [thick,blue] (EM) ++(0:1) -- +(0:2);
\draw [thick,green] (QECC) ++(240:1) -- +(240:2);

\end{tikzpicture}
\caption{The three aspects of a fault-tolerant protocol.}
\label{fig:FTaspects}
\end{figure}
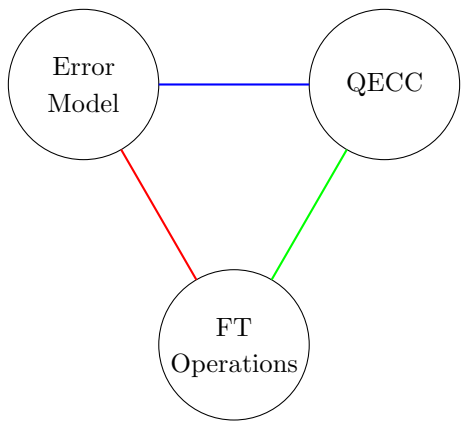

\subsection{Error models}
\label{sec:errormodels}

The error model tells us what kinds of errors occur during the computation.  Quantum states are vectors in a Hilbert space.  Since global phase has no physical significance, in fact a projective Hilbert space is more appropriate, but frequently we work with the full Hilbert space and only remove the phase at a later stage.  For the purpose of this paper, \emph{errors} are simply linear operators on the physical Hilbert space.

Given a particular physical implementation of quantum computation, certain kinds of errors will be more likely than others.  For instance, errors acting on a few qubits at a time are typically much more common than errors acting on many qubits.  The aim of the error model is to encapsulate the essential properties of the physical errors, while still being simple enough to make analysis tractable.  We will deal with a set of possible errors representing the range of likely errors; typically, the Kraus operators in a CP map decomposition\footnote{The CP map decomposition is not unique; however, the linear span of the Kraus operators is unique.} will be in the linear span of the set of possible errors.

Conceptually, it is often helpful to make a distinction between storage errors and gate errors:
\begin{Definition}
A \emph{storage error} is an error which occurs in the computer while no gates are being performed.  A \emph{gate error} is an error which occurs in one or more qudits while they are involved in the implementation of a physical gate.
\label{dfn:errortypes}
\end{Definition}
Storage errors are relatively simple to imagine, but the physical reasons for gate errors are a little more complicated, so let us elaborate a little.  Gate errors may be caused by faulty control, such as turning on an incorrect interaction or turning the correct interaction on for the wrong amount of time.  However, gate errors can also occur due to an interaction between storage errors on the qudits being manipulated and attempted control of the system.  This is a possibility because gates are not performed instantaneously in the lab, so storage errors continue to affect the system during that time and furthermore become altered by the gate, diversifying the types of storage-induced errors, which we classify as gate errors in this context.  For instance, if a phase ($Z$) error occurs during the implementation of a CNOT gate, the state at the end of the gate could have $X$ or $Y$ errors on it as well as $Z$ errors.  The gate error due to a storage $Z$ error occurring during a CNOT implementation can be written as $E = aI \otimes F_I + bX \otimes F_X + cY \otimes F_Y+ dZ \otimes F_Z$ (focusing on the error on the first qubit).  In fact, for any CNOT implementation, there exists a time for the storage error to occur so that $b \neq 0$.  Similarly, we can replace $b \neq 0$ with $c \neq 0$ or $d \neq 0$ here, and the statement remains true.  By the same token, a single-qubit storage error that occurs while a two-qubit gate is being performed is likely to produce errors in both qubits involved in the gate.  Due to these reasons, for the purpose of mathematical analysis, it is helpful to consider a storage error that occurs while a gate is being performed as part of the gate error.  In fact, as per Definition~\ref{dfn:errortypes}, we take the liberty to group together all errors which happen during the gate, regardless of their physical source, and call them gate errors.

More precisely, we can consider an \emph{error model} to consist of:
\begin{itemize}
\item a set of possible storage errors $\Errors$,
\item plus a set of gate errors $\Errors_U$ for each unitary physical gate $U$ (at least for those in the fault-tolerant protocol),\footnote{Some physical gates will never occur in the fault-tolerant protocol, so we don't have to specify $\Errors_U$ for those specific $U$'s.}
\item as well as other error sets for other (non-unitary) physical actions like preparing a qubit or measuring.
\end{itemize}

The sets $\Errors_U$ are highly dependent on $\Errors$, as we shall see shortly.  For most of this paper, except for Section \ref{sec:fullFT}, we focus on the first two classes of errors, i.e. the errors in $\Errors$ and in the sets $\Errors_U$.

For applications like proving the threshold theorem, the error model should also specify the probabilities of the different kinds of errors, but for the purposes of this paper, we won't need to think about probabilities; it is sufficient for us to just focus on the sets of errors.  This typically means that we want the error model to contain only the likely errors and exclude errors which are possible but very unlikely.  

For our fault tolerance analysis, the role of $\Errors$ is simply the set of possible errors that we assign to a pure wait time step.  A wait time step is shown as an ``empty'' extended rectangle in Figure \ref{fig:exRec}, i.e. where we only have two EC procedures at the two ends of the extended rectangle, sandwiching no physical or logical gate.

Mathematically and without loss of generality, let us think of a gate error as acting \emph{after} a perfect implementation of the desired physical gate.  That is, suppose the actual implementation of the gate did not realize the desired unitary $U$ but instead some other operator 
\begin{equation}
U' = E \circ U
\end{equation}
(not necessarily unitary) on the physical Hilbert space.  We then say that the gate error is $E$.  The choice to have the errors act after the gates is arbitrary; we could have chosen to have them act before, and that would have worked just as well, but we need to fix a convention.

We usually require that
\begin{equation}
\Errors \subseteq \Errors_U
\label{eqn:aftererror}
\end{equation}
for all $U$ since a physical storage error could occur right at the end of a gate implementation, in which case we still consider it as part of the gate error.  Similarly, a storage error could occur right at the beginning of a gate $U$, so we also require that
\begin{equation}
E \in \Errors
\Longrightarrow U E U^{-1} \in \Errors_U,
\label{eqn:beforeerror}
\end{equation}
since what actually gets implemented is the gate $U \circ E = (U E U^{-1}) \circ U$.  Again, we consider the storage-induced error $U E U^{-1}$ as part of the gate error $\Errors_U$, since it occurs during the implementation of a gate.
As a useful convention, let us throw the identity error into each of the error sets.


We frequently work with Pauli errors.  When expressed in the qubit computational basis, the single-qubit Pauli operators are given by
\begin{equation}
\sigma^{X} := X := \begin{pmatrix} 0&1\\ 1&0 \end{pmatrix},
\;\;
\sigma^{Z} := Z :=  \begin{pmatrix} 1&0\\ 0&-1 \end{pmatrix},
\;\;
\sigma^{Y} := Y := iXZ =  \begin{pmatrix} 0&-i\\ i&0 \end{pmatrix}.
\label{eqn:Paulis}
\end{equation}
In addition, we shall need the definition
\begin{equation}
\sigma^{0} := \idmatrix
\end{equation}
where $\idmatrix$ is the identity operator on the qubit Hilbert space.  There are standard generalizations of the Pauli errors to qudits~\cite{Knill}.  We will also deal with multiple-qudit Pauli errors, which are tensor products of the single-qudit Paulis.

\begin{Definition}
The \emph{weight} of an $n$-qudit Pauli error is the number of non-identity tensor factors in the tensor product.
\label{def:weight}
\end{Definition}

\subsection{Quantum error-correcting codes}
\label{sec:QECCs}

The QECC gives us the ability to correct errors.  In a QECC, a logical Hilbert space is encoded in a larger physical Hilbert space in such a way that errors can be identified and corrected.  Naturally, it is not possible to correct all possible errors\footnote{For example, those errors that are ``tangential" to the code subspace will be undetectable and hence uncorrectable.}, so to be useful, a QECC should focus on the most likely kinds of errors in an implementation.

Here we briefly review a few key concepts of quantum error correction.  See~\cite{QECCintro} for more details.  Let us define an $((n,K))$ {\it qudit quantum error-correcting code} as a linear subspace $C$ of dimension $K$, called the \emph{codespace}, sitting inside a $d^{n}$-dimensional \emph{physical} Hilbert space $\Hilbphy = \nqudits$ of $n$ qudits.  Typically the physical system in use has a physically natural tensor product decomposition into qudits.  For instance, each qudit may be one ion in an ion trap.

It is often the case that we want to run a computation using more than $\log_d K$ logical qudits\footnote{Typically $K$ is a power of $d$, the size of the physical qudits, namely, we usually like to encode an integer number of qudits.}, so a single $((n,K))$ qudit QECC will not be enough.  Let's say we need $m$ logical qudits.  We call the Hilbert space $\mqudits$ of the logical qudits the \emph{logical space}, denoted by $\Hilblog$.  When $m > \log_d K$, we can divide up the $d^m$-dimensional logical space into a tensor product of logical subspaces of dimension $\leq K$ (in whatever way is most convenient) and then encode each of these tensor factors as codespaces in a separate \emph{block} of the QECC, using a total of $n \lceil m/\log_d K \rceil$ physical qudits.

QECCs are made to correct errors.  The following definition makes the notion of error correction precise.
\begin{Definition}
Given a set $\CorErrors = \{E_a\}$ of errors, we say that a QECC $C$ \emph{corrects} this set of errors if there exists a quantum (recovery) operation $\Recovery$ such that
\begin{equation}
(\Recovery \circ E_a) \ket{\psi} \propto \ket{\psi} \text{  for all $E_a \in \CorErrors$ and $\ket{\psi} \in C$}.
\end{equation}
\end{Definition}

This operational definition of correctability translates into the following algebraic condition \cite{BDSW,KL}:
\begin{Theorem}
\label{thm:correction}
A QECC $C$ corrects the set of errors $\CorErrors$ if and only if
\begin{equation}
\bra{\psi_i} E_a^{\dagger} E_b \ket{\psi_j} = f_{ab} \delta_{ij}
\label{eqn:correction}
\end{equation}
where $E_a, E_b \in \CorErrors$ and $\{\ket{\psi_i}\}$ is an orthonormal basis for the codespace $C$.
\end{Theorem}

The independence of the numbers $f_{ab}$ from $i$ and $j$ is the crucial element here.
For any pair of errors $E_a$ and $E_b$, since $f_{ab}$ does not depend on $i, j$, one can verify that if equation (\ref{eqn:correction}) is satisfied for one orthonormal basis of $C$, then it is satisfied for any other orthonormal basis as well, with the same constants $f_{ab}$.

\begin{Definition}
The \emph{distance} of a QECC $C$ is the minimum weight\footnote{See Definition \ref{def:weight}.} of a (qudit) Pauli operator $P$ for which there does not exist $f(P)$ independent of $i$ and $j$ such that
\begin{equation}
\bra{\psi_i} P \ket{\psi_j} = f(P) \delta_{ij}
\end{equation}
for all orthonormal pairs of vectors $\ket{\psi_i}$, $\ket{\psi_j}$ in $C$. If $C$ is an $((n,K))$ qudit QECC with distance $\delta$, we say that $C$ is a \emph{minimum distance (qudit) code} of type $((n,K,\delta))$.
\label{def:distance}
\end{Definition}

As a corollary to Theorem \ref{thm:correction}, we have the following.

\begin{Corollary}
An $((n,K,\delta))$ code can correct all errors\footnote{Only Pauli errors of weight $t$ is defined in this paper.  However, an $((n,K,\delta))$ code indeed corrects general errors of weight $t$ once that is properly defined.} of weight $t$ where $2t + 1 \le \delta$.
\label{cor:distance}
\end{Corollary}

It is also useful to have the notion of an encoding for a QECC.  Let $C$ be an $((n,K))$ qudit quantum error-correcting code as above, that is, $C$ is a $K$-dimensional vector subspace of the $d^n$-dimensional $\Hilbphy$.  An {\it encoding}\footnote{Here, we do not consider subsystem encodings.} for $C$ refers to a choice of a unitary isomorphism $\iota$ from the logical space $\Cvect{K}$ to $C$.  A {\it decoding} for $C$, on the other hand, is a choice of a unitary isomorphism $\omega$ from $C$ to the logical space $\Cvect{K}$.  In a fault-tolerant protocol, the encoding and decoding mostly happen only abstractly in the mind of the quantum information theorist, but are nevertheless very useful and important concepts.

In a fault-tolerant protocol, the QECC $C$ must be matched to the given error model (the blue line in Figure~\ref{fig:FTaspects}).  In particular, the code $C$ should be able to correct at least the set $\Errors$ of storage errors, and for $U$ ranging over the unitary physical gates used in the fault-tolerant protocol, $U(C')$ should be able to correct $\Errors_U$, where $C'$ is the code obtained by applying all the physical gates prior to $U$ in the fault-tolerant gadget to the reference code $C$.

\subsection{Fault-tolerant operations}
\label{sec:FToperations}

Fault-tolerant gadgets let us perform various kinds of actions on a quantum state encoded in a QECC.  

First, we quickly review what is involved in a fault-tolerant quantum computation.  A fault-tolerant protocol consists of a set of \emph{gadgets} which perform encoded versions of the various components of a quantum circuit.  At a minimum, a fault-tolerant protocol needs gadgets for state preparation, measurement, and a universal set of gates.  The protocol also needs a gadget for fault-tolerant error correction, inserted into the fault-tolerant circuit periodically, as in Figure \ref{fig:exRec}, to prevent errors from building up during the computation.  All gadgets are performed on states encoded in the QECC used in the protocol.  Thus, a fault-tolerant circuit might begin with state preparation gadgets to create a number of encoded $\ket{0}$ states, continue with a sequence of gate gadgets implementing a quantum algorithm, interleaved by fault-tolerant error correction procedures, and end with measurement gadgets which output the logical value of one or more encoded qubits.

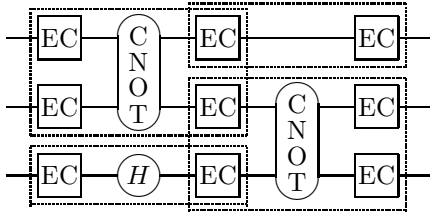
\begin{figure}
\begin{picture}(160, 80)

\put(0,20){\line(1,0){12}}
\put(0,46){\line(1,0){12}}
\put(0,72){\line(1,0){12}}

\put(12,12){\framebox(16,16){EC}}
\put(12,38){\framebox(16,16){EC}}
\put(12,64){\framebox(16,16){EC}}

\put(28,20){\line(1,0){14}}
\put(28,46){\line(1,0){14}}
\put(28,72){\line(1,0){14}}

\put(50,20){\circle{16}}
\put(50,20){\makebox(0,0){$H$}}

\put(50,59){\oval(16,44)}
\put(46,40){\shortstack{C\\N\\O\\T}}

\put(58,20){\line(1,0){14}}
\put(58,46){\line(1,0){14}}
\put(58,72){\line(1,0){14}}

\put(72,12){\framebox(16,16){EC}}
\put(72,38){\framebox(16,16){EC}}
\put(72,64){\framebox(16,16){EC}}

\put(88,20){\line(1,0){14}}
\put(88,46){\line(1,0){14}}
\put(88,72){\line(1,0){44}}

\put(110,33){\oval(16,44)}
\put(106,14){\shortstack{C\\N\\O\\T}}


\put(118,20){\line(1,0){14}}
\put(118,46){\line(1,0){14}}

\put(132,12){\framebox(16,16){EC}}
\put(132,38){\framebox(16,16){EC}}
\put(132,64){\framebox(16,16){EC}}

\put(148,20){\line(1,0){12}}
\put(148,46){\line(1,0){12}}
\put(148,72){\line(1,0){12}}

\put(9,35){\dashbox{0.7}(82,48){ }}
\put(9,9){\dashbox{0.7}(82,22){ }}

\put(69,7){\dashbox{0.7}(82,50){ }}
\put(69,61){\dashbox{0.7}(82,24){ }}

\end{picture}
\caption{A schematic example of a portion of a circuit built of fault-tolerant gadgets, such as EC, CNOT and $H$.  We can consider a fault-tolerant circuit to be broken up into fundamental units called ``extended rectangles'', marked by dotted lines in the figure.  Each (solid) line represents a block of the QECC.  EC is an error correction step, and CNOT and $H$ are two kinds of logical gate.  In general, a fault-tolerant gadget, such as the CNOT or $H$ logical gates here, might also contain error correction.  This is because a fault-tolerant gate might itself be made up of multiple physical gates, and sometimes we need to perform error correction between the physical gates in order to prevent the errors from becoming too severe.}
\label{fig:exRec}
\end{figure}

Each gadget is composed of a number of physical gates, and each physical gate $U$ can potentially be followed by an error drawn from $\Errors_U$.  Sometimes, error correction is performed between the physical gates, to prevent the errors from becoming too severe.  Figure~\ref{fig:exRec} gives an example of part of a fault-tolerant circuit.  In the case of the $7$-qubit code, the logical CNOT and $H$ gate gadgets are each made up of a tensor product of $n$ physical gates (i.e. are transversal), where $n$ is the number of physical qudits used in the QECC.  Since only a tensor product is involved in this case, the logical CNOT and $H$ gates can actually be accomplished in depth 1, and can therefore be considered as made up of a single physical gate.  Better examples of logical gates being composed of multiple physical gates include magic state constructions, error correction gadgets, and moving defects around in a toric code as described in Sections \ref{sec:topologicalFT} and \ref{sec:local_errors}.

Naturally, in a fault-tolerant protocol, fault-tolerant gadgets must be matched to the QECC (the green line in Figure~\ref{fig:FTaspects}) and to the error model (the red line in Figure~\ref{fig:FTaspects}).  In the absence of error, a fault-tolerant gate\footnote{A fault-tolerant gate is always used to mean a fault-tolerant logical gate, not the physical gates that make up the fault-tolerant gadget/gate.} $F$ should map a valid encoded state into another encoded state, namely
\begin{equation}
F(C)=C
\label{eqn:code-gate-relation}
\end{equation}
for the QECC $C$ chosen for our fault-tolerant protocol.  In general, there may be some physical errors already in the state at the beginning of the gadget.  The FT operation will combine these with additional errors (after possible intermediate error correction) that occur during the gadget.  As a property of an FT operation, if neither group of errors is too severe, we ask that the error at the end of the gadget is still correctable by the code.

In order to make sure that the errors during a FT gadget don't become too severe, we might have to perform error correction within the gadget, and the error-correction property the FT gadget must satisfy has implications for the intermediate physical gates, their error sets and the intermediate codes.  Here it is important to keep track of not just the fault-tolerant gates, but also how each fault-tolerant gate $F$ is discretely implemented, i.e. the decomposition $F = U_m \cdots U_1$ into a series of physical gates $U_i$.  In some sense, we need to make sure the error-correcting properties of the reference code and of intermediate codes are not spoiled by the implementation of the gadgets, even in the presence of imperfection in their implementation, and that the error-correcting properties of the intermediate codes are strong enough to correct the accumulated errors.  Since we don't correct errors after every physical gate, when we do correct errors, say after some physical gate $U$, we might need to correct accumulated errors beyond what is in the set $\Errors_U$.  Therefore, the compatibility we ask for here between the fault-tolerant operations and the error model is more stringent than that outlined at the end of the Section \ref{sec:QECCs}.

Similarly, the error-correction property of the FT gadget has implications not only on the discrete implementation of the gadget, but also on the continuous-time implementation and all the intermediate codes.  By continuous-time implementation of a fault-tolerant gate $F$, we mean we keep track of not only the physical gates $U_i$ that make up $F$, but also their continuous-time implementation, e.g. evolving the system under some Hamiltonian.  However, we do not discuss the compatibility between the continuous-time implementation and the error model so much here because for now we model the errors to happen at discrete time, after each physical gate.

In this paper, we will focus mainly on \emph{unitary} fault-tolerant gate gadgets.  To qualify as a unitary fault-tolerant gate, not only does the resulting fault-tolerant gate need be unitary, it also needs admit a continuous-time implementation which is unitary for all time $0 \le t \le 1$.

\subsection{Example: The $s$-qudit error model and transversal gates}
\label{sec:squdit_errors}

Let us consider an example that we will return to in Section~\ref{sec:transversal} and see how the choice of error model, QECC, and fault-tolerant gadgets (and physical gates) must fit together in this case.

We define an \emph{$s$-qudit error model} for when the underlying system has errors naturally occurring independently on $\leq s$ separate qudits.  We shall make this brief statement more precise below.  The underlying physical Hilbert space in this case is $\Hilbphy := \bigotimes_{i=1}^{n} \Hilb{i}$ where $n$ is the number of qudits and $\Hilb{i} = \Cvect{d}$ denotes the qudit Hilbert space for the $\ord{i}$ qudit.  The set of storage errors $\Errors$ in the $s$-qudit error model is the set of all tensor product errors of weight $s$ or less, namely $E = \bigotimes_i E_i$ where $E_i$ acts on $\Hilb{i}$ and $E_i \ne \idmatrix$ for at most $s$ places.

Without loss of generality, a unitary gate $U$ takes the form $\bigotimes_j U_j$, expressed as a tensor decomposition maximizing the number of tensor factors.  Each $U_j$ acts on a subspace $\overline{\Hilb{}}_{j}$ which contains the tensor product of an integer number of the qudit Hilbert spaces $\Hilb{i}$.  $U$ acts on $\Hilbphy$, so $\bigotimes_{j} \overline{\Hilb{}}_{j} = \Hilbphy$.  Furthermore, since the tensor decomposition $\bigotimes_j U_j$ maximizes the number of tensor factors, such a decomposition is unique.  With respect to the Hilbert space decomposition $\bigotimes_{j} \overline{\Hilb{}}_{j} = \Hilbphy$, we define the set of gate errors $\Errors_U$ associated with $U$\footnote{For the error model to be realistic, $U$ must be implemented by a time sequence $U(t)$ of unitaries each of which respecting the tensor product decomposition $\bigotimes_{j} \overline{\Hilb{}}_{j}$.  The $U(t)$ here is a unitary evolution, as defined in Definition \ref{def:uni_evolution}.  In particular, $U(0) = \idmatrix$ and $U(1) = U$.} to be the set of all tensor product errors $F = \bigotimes_j F_j$ where $F_j$ acts on $\overline{\Hilb{}}_{j}$ and at most $s$ of the $F_j$'s are nontrivial.  In particular, if $U$ is a tensor product of unitaries on the individual qudit Hilbert spaces $\Hilb{i}$, then $\Errors = \Errors_U$; otherwise $\Errors \subsetneq \Errors_U$.

By Corollary \ref{cor:distance}, a code with distance $\delta$ corrects $s$ qudit errors iff $\delta > 2s$.  Thus, any code of, for example, distance $2s+1$ is well-suited to correct a $s$-qudit error model if we consider only storage errors.

However, in a quantum fault-tolerant protocol, we need a code that can correct gate errors from $\Errors_U$.  Suppose the gate $U$ is of the form $\bigotimes_i U_i$ where each $U_i$ acts on one of the qudits.  Then $\Errors = \Errors_U$ and the gate errors $E_1$ and $E_2$ (as in Theorem \ref{thm:correction}) would also have a weight of at most $s$, hence in this case $E_1^\dagger E_2$ is an error of weight at most $2s$.  Therefore, the same code that corrects the errors in a wait step can also correct the errors incurred in a time step where the gate $U$ is performed.  Hence, a distance $2s+1$ code may again be a suitable QECC for fault-tolerant computation under the $s$-qudit error model, if all gates involved are of this form.

On the other hand, if each $U_i$ acts non-trivially on $2$ qudits in a single block of the code, then $\Errors_U$ contains errors of weight $2s$.  In this case, we would need a code with distance at least $4s+1$.  To avoid needing larger distance codes than $2s+1$, we shall design fault-tolerant protocols for the $s$-qudit error model to avoid multi-qudit gates acting within a single block of the QECC.  Therefore, to achieve fault tolerance in this case, we restrict our physical gates to be transversal gates, defined below.

\begin{Definition}
A (possibly multi-block) \emph{transversal} gate $U$ is one that can be written as $U = \bigotimes_i U_i$, where the unitary $U_i$ acts only on the $\ord{i}$ qudit in each block of a QECC.
\label{def:transversal}
\end{Definition}
For instance, if we are dealing with a system of qubits, any tensor product of single-qubit gates is a transversal operation, as is a tensor product of CNOT gates, where the $\ord{i}$ CNOT uses qubit $i$ from block $1$ as the control qubit and qubit $i$ from block $2$ as the target qubit.  For CSS codes, the transversal CNOT between two code blocks implements the logical CNOT between the two encoded qubits \cite{Shor}.

A transversal gate that maps the code space back into itself is automatically fault-tolerant with respect to the $s$-qudit error model and a distance $2s+1$ code, in part by the same argument which shows that a code with distance $2s+1$ corrects weight $s$ errors.  Thus, for example, the $s$-qubit error model, a distance $2s + 1$ QECC, and (faulty) transversal gates fit together to produce a sensible fault-tolerant protocol.  Unfortunately, such a protocol is incomplete, since the set of transversal gates on a QECC can never execute a universal set of logical gates~\cite{EastinKnill}.  However, in some important cases, a large number of gates can be performed transversally.  For instance, the $7$-qubit code has a transversal implementation of the full Clifford group \cite{Shor}.

\subsection{Example: The geometrically local error model and the toric code}
\label{sec:local_errors}

A second example is when errors are geometrically local and we use topological codes and gates.  More specifically, in this paper we simply consider the toric code and string operators, which will serve as examples of a topological code and topological gates.  We present a fibre bundle picture of this model in Section~\ref{sec:topologicalFT}.

The \emph{2D $(s,t)$-geometrically local error model} is specified by two parameters $s$ and $t$.  Let us consider the $n$-qubit Hilbert space $\Hilbphy := \nqubits$, with the qubits arranged in $2$ spatial dimensions, say on the vertices (as in Figure~\ref{fig:geoerrors}) or the edges (as in toric code) of a square lattice with unit lattice spacing.  The set of possible storage errors $\Errors$ consists of \emph{all} tensor product Pauli errors which have support on a union $\bigcup_{i=1}^{s} S_i$ of at most $s$ clusters of qubits, where each $S_i$ is contained in a disk $D_i$ of diameter at most $t$, as depicted in Figure~\ref{fig:geoerrors}.

\begin{figure}
\begin{tikzpicture}
\draw [gray,very thin,step=.5cm] (-2,-2) grid (2,2);
\foreach \x in {-2,..., 1}
	\foreach \y in {-2,-1.5,...,2}
		\draw [blue, very thick] (\x,\y) -- +(0.5,0);

\node at (-1,-1) {X};
\node at (-1,-0.5) {X};
\draw (-1,-0.75) circle (0.4);
\draw [dotted] (-0.75,-0.75) circle (0.65);

\node at (1,0) {X};
\draw (1,0) circle (0.3);
\draw [dotted] (1.25,0) circle (0.45);

\node at (0,1.5) {X};
\node at (0.5,1.5){X};
\draw (0.25,1.5) circle (0.4);
\draw [dotted] (0.25,1.5) circle (0.45);

\end{tikzpicture}
\caption{An example of a $(3,2)$-geometrically local error.  X represents an error, and a solid circle indicates a cluster $D_i$.  If we do two-qubit gates between adjacent qubits (represented by thick blue lines), the clusters can be larger, represented by dotted circles.  In this figure, the qubits live on the vertices of the lattice.  This is different from the usual convention for the toric code, which has qubits on the edges of a lattice.}
\label{fig:geoerrors}
\end{figure}
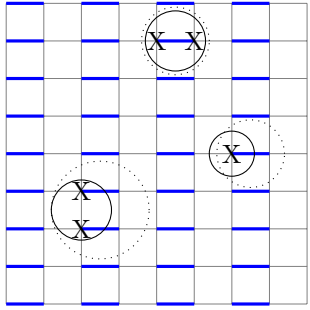

Without loss of generality, a unitary gate $U$ can be written as $U = \bigotimes_{j=1}^{r} U_j$, where $U_j$ acts on $\overline{\Hilb{}}_{j}$, consisting of a tensor product of an integer number of geometrically neighbouring $\Hilb{i}$'s.  The decomposition is chosen to maximize the number of tensor factors.  Then the set $\Errors_U$ of gate errors consists of tensor product Pauli errors acting on a union $\bigcup_{i=1}^{s} S'_i$ of at most $s$ clusters, where each cluster $S'_i$ is contained within $D_i \cup \supp{U_j}$, for some disk $D_i$ of radius $t$ and for all $j$ such that $\supp{U_j} \cap D_i \ne \emptyset$.  We see that in general $\Errors \subsetneq \Errors_U$ because $S_i \subset D_i$ for errors in $\Errors$ whereas the analogous set inclusion $S'_i \subset D_i$ does not generally hold true for errors in $\Errors_U$.  In fact, it can be that $D_i \subset S'_i$, as shown in Figure~\ref{fig:geoerrors}.  We may assume $t \geq 1$ always; when the error sets are all empty, we say that we have a $(0,1)$-geometrically local error model.

As a QECC resilient against the $(s,t)$-geometrically local error model, we use Kitaev's toric code \cite{Kitaev1}, which we define explicitly in Section \ref{sec:toric}.  The standard toric code is defined for a qubit system on a torus, geometrically arranged so that the qubits sit on the edges of a square lattice with unit lattice spacing.  If the underlying lattice on the torus has period $L$, the code then has distance $L$, that is, it is robust against the $s$-qubit storage errors,  as defined in Section \ref{sec:squdit_errors}, with weight $s < \flatquotient{L}{2}$.  On the other hand, since the toric code is sensitive to the local and toroidal geometry, it is also a natural candidate for providing protection against geometrically local errors, and the code distance by itself is not a good measure of the error correction power of the toric code.  We expect the toric code to ``do better'' when we consider a more restricted class of errors.  Indeed, this is what happens.  If we consider only the storage errors, i.e. in a quantum memory setting, a toric code of size $L$ is robust against the $(s,t)$-geometrically local error model for $st < \flatquotient{L}{2}$, effectively increasing the maximum allowed number of error locations from $\lfloor \flatquotient{(L-1)}{2} \rfloor$ to about $\flatquotient{s \pi t^2}{4} = (\flatquotient{\pi t}{4}) \lfloor \flatquotient{(L-1)}{2} \rfloor$ by taking advantage of the geometric clustering of the errors, especially effective when $t$ is large.  Also note that the $s$-qubit error model is the same as the $(s,1)$-geometrically local error model.  Hence, robustness against $(s,t)$-geometrically local error models for $st < \flatquotient{L}{2}$ is a strictly stronger statement than robustness against $s$-qubit error models for $s < \flatquotient{L}{2}$.

In this case, a natural set of fault-tolerant operations compatible with the $(s,t)$-geometrically local error model and the toric code are those consisting of a tensor product of single-qubit gates or gates interacting only clusters of qubits of diameter $d$, since these gates respect the geometric locality of the code.  Then if the error before a gadget is an $(s_1, t)$ error and the gate errors are $(s_2,t)$-geometrically local, then the error after the gadget becomes an $(s_1 + s_2, t+2(d-1))$-geometrically local error, since each cluster could have been increased in size by error propagation, but only by at most $(d-1)$ qubits in each direction.  Here, when the $s_1 + s_2$ disks of error support overlap after being expanded by $(d-1)$ in radius, for ease and certainty of analysis, we still consider the new support of errors to be decomposed into the same $s_1 + s_2$ disks, say with the original centers.

For instance, if the gates involve tensor products of two-qubit Clifford operators (where gates interact only nearest-neighbour qubits), then each cluster can have diameter up to $t+2$.  In this case, the toric code of size/distance $L$ is able to correct the accumulated error after the gate, provided that $(s_1+s_2)(t+2) < \flatquotient{L}{2}$.

In section \ref{sec:topologicalFT}, we work with FT operations which are certain tensor products of single-qubit Paulis, usually known as the \emph{string operators}, which will be defined in Section \ref{sec:toric}.  String operators interact clusters of qubits of diameter $d=1$, so if the initial error is $(s_1, t)$-geometrically local, the error after applying a string operator with $(s_2, t)$-geometrically local gate errors is $(s_1 + s_2, t)$-geometrically local.  The string operators are far from being universal --- they generate the logical Paulis.  For certain topological codes, sequences of local operations acting on slightly larger clusters \emph{can} produce universal sets of logical gates~\cite{FLW}.

\section{Geometric preliminaries}
\label{sec:geometry}

In this section, we will introduce the main mathematical structures used in this paper.  In Section~\ref{sec:fibre}, we introduce fibre bundles, including vector bundles and principal bundles, and in Section~\ref{sec:connection}, we discuss connections and holonomy.  For more information about these topics, see, for instance, \cite{Nakahara}.

\subsection{Fibre bundles}
\label{sec:fibre}

To establish the main claim of this paper, we shall make use of two bundles over the set of all $K$-dimensional subspaces (all the possible $K$-dimensional {\it codespaces} in the language of quantum error correction) of the physical Hilbert space $\Hilbphy = \Cvect{N}$.  Here, we introduce the general notion of a fibre bundle and two prominent subclasses thereof, principal $G$-bundles and vector bundles.

\subsubsection{Definition of a fibre bundle}
The notion of a fibre bundle mirrors that of a manifold --- recall the definition of a manifold involving a maximal atlas of coordinate charts.  Intuitively, a fibre bundle is a manifold with some topological space (the \emph{fibre}) attached to each point.  To make this definition rigorous, we employ the notions of bundle projection and local trivialization.

\begin{Definition}[Bundle projection \& Trivialization]
Let $E$, $B$ and $F$ be Hausdorff spaces\footnote{Roughly speaking, a Hausdorff space is a topological space that satisfies some basic properties to make it reasonably well-behaved.  For a formal definition, refer to \cite{Munkres}.} and $\pi:E\rightarrow B$ a (continuous) map.  Then $\pi$ is called a \emph{bundle projection} with fibre $F$ if each point of $B$ has a neighbourhood $U$ such that there is a homeomorphism $\phi: U\times F \rightarrow \pi^{-1}(U)$ such that $\pi(\phi(p,f))=p$ for all $p\in U$ and $f\in F$.
Such a map $\phi$ is called a \emph{trivialization} of the bundle over $U$.  More specifically, $\phi$ may be called a trivialization of $\pi$, or of $E$, over $U$.  Such a pair $(U, \phi)$ is called a \emph{(local) trivialization chart}.
\label{def:bundproj}
\end{Definition}

We now explain a couple of terms which will appear in Definition \ref{def:bund} of a fibre bundle.  A \emph{topological group} is a group which is also a topological space and where the group operations are continuous.  An action of $G$ on $F$ is said to be \emph{faithful} or \emph{effective} if for any $g\ne h$ in $G$, there exists an element $f\in F$ such that $g\cdot f \ne h\cdot f$.

The following is basically Definition 13.2 of \cite{Bredon}.

\begin{Definition}[Fibre bundle]
Let $G$ be a topological group acting faithfully on the Hausdorff space $F$ as a group of homeomorphisms from the topological space $F$ to itself.  Let $E$ and $B$ be Hausdorff spaces. Then a \emph{fibre bundle} $(E,B,\pi,F,G)$ [over the base space $B$ with total space $E$, fibre $F$ and structure group $G$] is a bundle projection $\pi:E \rightarrow B$ (with fibre $F$) together with a collection $\Phi$ of trivialization charts,
such that:
\begin{enumerate}
\item Each point of $B$ has a neighbourhood $U$ over which there is a trivialization chart $(U, \phi)$ in $\Phi$; \label{item:neighbourhood}
\item If $\phi_i,\phi_j \in \Phi$ are charts over $U$, then there is a continuous map $\theta_{ij}: U\rightarrow G$ such that $\phi_j(p,f)=\phi_i (p,\theta_{ij}(p)\cdot f)$ for all $p\in U$, $f\in F$, called a {\it transition function} or {\it coordinate transformation};
\item If $\phi: U\times F \rightarrow \pi^{-1}(U)$ is in $\Phi$ and $V\subset U$, then the restriction of $\phi$ to $V\times F$ is in $\Phi$;
\item The set $\Phi$ is maximal among collections satisfying (\rmnum{1}), (\rmnum{2}), and (\rmnum{3}).
\end{enumerate}
\label{def:bund}
\end{Definition}

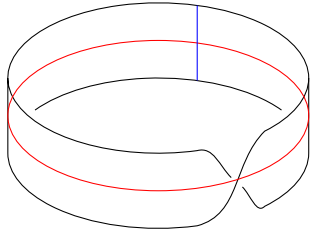
\begin{figure}
\begin{tikzpicture}

\path  (285:2cm and 1cm) ++(0,1) coordinate (A);
\path (299:2cm and 1cm) ++(0,0.55) coordinate (midup) ;
\path (304:2cm and 1cm)++(0,0.38) coordinate (middown);

\path (285:2cm and 1cm) coordinate (end);
\draw (midup) to [out=135,in=15] (A) arc (285:0:2cm and 1cm) arc (360:315:2cm and 1cm) to [out=225,in=15] (end) arc (285:180:2cm and 1cm) -- +(0,1);

\draw (middown) to [out=315,in=225] (315:2cm and 1cm) arc (315:360:2cm and 1cm) -- +(0,1);
\draw (35:2cm and 1cm) arc (35:145:2cm and 1cm);

\draw [color=red]  (0,0.5) ellipse (2cm and 1cm);
\draw [color=blue] (75:2cm and 1cm) -- ++(0,1);

\end{tikzpicture}
\caption{The M\"{o}bius strip is a fibre bundle over the base space $S^1$ (in red) with fibre the interval $[-1, 1] \subset \mathbb{R}$ (fibre at one point shown in blue).}
\label{fig:moebius}
\end{figure}

A collection $\Phi$ of trivialization charts satisfying properties (\rmnum{1}) -- (\rmnum{2}) is known as a \emph{$G$-atlas}, in analogy with an atlas for a manifold.  Comparing Definitions~\ref{def:bundproj} and \ref{def:bund}, we see that property~(\ref{item:neighbourhood}) of a fibre bundle is similar to the intrinsic requirement of a bundle projection.  However, in the definition of a bundle projection, we only ask that a trivialization over some neighbourhood $U$ exits, whereas in the definition of a fibre bundle, we ask the $G$-atlas $\Phi$ to contain such a trivialization chart.

Figure~\ref{fig:moebius} shows a simple example of a fibre bundle.  Here, the base space $B$ is $S^1$, and the fibre $F$ is the interval $[-1, 1]$ in the vector space $\mathbb{R}$.  To get a nontrivial fibre bundle, we let the total space $E$ be the M\"{o}bius strip.  The structure group of this bundle is $G = O(1) = \mathbb{Z}_2$, where $O(1)$ stands for the group of orthogonal linear transformations of the vector space $\mathbb{R}$.


Let us mention a fact about the structure group $G$: Suppose $p\in B$, then the set
\begin{equation}
\{\theta_{ij}(p) \in G \;|\; \phi_i,\phi_j\in \Phi \text{ are charts over $U$ for some $U \ni p$}\} = G,
\label{eqn:strucgrp}
\end{equation}
where $\Phi$ and $G$ are as in Definition \ref{def:bund} and $\theta_{ij}$ is the transition function between $\phi_i$ and $\phi_j$ given by condition (\rmnum{2}).
This is because the $G$-atlas $\Phi$ in Definition \ref{def:bund} is maximal, and therefore contains charts ``shifted by'' any group element $g\in G$.  So, we could alternatively think of the structure group as defined by the set on the left hand side of equation (\ref{eqn:strucgrp}).  The structure group is also known as the \emph{gauge group} in physics.

\begin{Remark}
In this paper, all the spaces that appear are not only Hausdorff topological spaces but also subsets of manifolds.  Furthermore, all the spaces and functions are continuous and piecewise smooth.  In all cases of interest to us, $G$ is a smooth Lie group, the action of a group element $g\in G$ is a diffeomorphism, sometimes even a linear map, on $F$, and all other (originally continuous) maps appearing in the definitions are now piecewise smooth.
\end{Remark}

\subsubsection{Vector bundles and principal bundles}
\begin{Definition}[Vector bundle]
A \emph{vector bundle} $(E,B,\pi,F,G)$ is a fibre bundle whose fibre $F$ is not only a Hausdorff topological space but also comes equipped with a vector space structure.
\end{Definition}

For a complex vector bundle $E$, if its fibre $F=\Cvect{d}$, we call $d$ the \emph{fibre dimension} and denote it by $\dim{E}$.  Then the structure group $G$ for such a vector bundle may be $\GL{d}$ or some subgroup thereof, such as the unitary group $\U{d}$.

\begin{Definition}[Principal bundle]
A \emph{principal bundle} or \emph{$G$-bundle} $(E,G,\pi,F,G)$ is a fibre bundle with fibre $F=G$ and the structure group $G$ acting on the fibre $F=G$ by left multiplication.
\end{Definition}

Principal bundles play an important role in the setting of gauge theories and in this work.

\subsubsection{Transition functions and associated bundles}

The transition functions are what allow us to have a solid handle on fibre bundles, and are useful for understanding associated bundles.  Let us examine them more closely.

Let $\theta_{ij}$ be the transition function from the fibre coordinates under chart $j$ to those under chart $i$, 
as per Definition~\ref{def:bund} (\rmnum{2}).  Since the $\phi_{*}$'s are charts, the $G$-valued transition functions at any point $p\in B$ satisfy:
\begin{enumerate}
\item $\theta_{ii}(p)=e$,
\item $\theta_{ji}(p)=\theta_{ij}(p)^{-1}$,
\item Suppose $\phi_i, \phi_j, \phi_k$ are three charts over $U_i, U_j, U_k$ respectively, with nonempty triple intersection $U_i \cap U_j \cap U_k \ne \emptyset$.  Then $\theta_{ik}(p)=\theta_{ij}(p) \theta_{jk}(p)$.
\end{enumerate}

\begin{Definition}[Associated bundle]
Given a bundle $E$, short for $(E,B,\pi,F,G)$, with fibre $F$, and another Hausdorff space $F'$ with an action of the topological group $G$, we can produce a new bundle $(E',B,\pi',F',G)$ with the same base space and same transition functions (and hence necessarily the same structure group), but with different fibre $F'$ upon which $G$ also acts.  The new bundle $E'$ is called the \emph{associated bundle} of $E$ with fibre $F'$.
\end{Definition}

The notion of associated bundles allows us to go between certain different bundles with the same structure group.
The concept of associated bundles relates the tautological vector bundle $\bigvecbund$ and the tautological principal bundle $\bigPbund$, which we will introduce later in the paper.

\subsubsection{Pullback of a fibre bundle}
\label{sec:pullback}

Let $\pi: E \rightarrow B$ be a fibre bundle with fibre $F$.  If a continuous map $f: B' \rightarrow B$ is given, then the fibre bundle $E$ over $B$ can be ``pulled back'' along $f$ to obtain a new fibre bundle over $B'$ with the same fibre $F$, denoted by $\fE$ and defined below.

\begin{Definition}[Pullback bundle]
Let $E$ and $f$ be as above, and let $\pi$ be the bundle projection for $E$.  Then we define $\fE$ to be the subspace of $B'\times E$ which consists of points $(p,u)$ such that $f(p)=\pi(u)$, namely
\begin{equation}
\fE : = \{ (p,u) \in B' \times E \;|\; f(p)=\pi(u) \},
\label{eqn:pullback}
\end{equation}
with bundle projection $\pi_1: \fE \rightarrow B'$ given by $\pi_1(p,u)=p$.  We thus obtain the {\it pullback bundle} of $E$ by $f$.
\end{Definition}
The idea of the pullback is that, roughly speaking, the fibre of $\fE$ over $p \in B'$ corresponds to the fibre of $E$ over $f(p) \in B$.  Implicitly, the topology on $\fE$ is taken to be the subspace topology of the product space $B' \times E$.  In fact, we can formulate this relation between $E$ and $\fE$ more clearly by observing that there is a bundle map from $\fE$ to $E$.  Below, we define the general notion of a bundle map, and specify this canonical bundle map from $\fE$ to $E$.

\begin{Definition}[Bundle map]
Let $\pi: E \rightarrow B$ and $\pi': E' \rightarrow B'$ be fibre bundles.  A continuous map $\bar{f}: E' \rightarrow E$ is called a bundle map if it maps each fibre $\pi'^{-1}(p)$ of $E'$ into fibre $\pi^{-1}(q)$ of $E$ for some $q$ in $B$.
\end{Definition}

In other words, a bundle map is a continuous fibre-preserving map.  We see that a bundle map $\bar{f}: E' \rightarrow E$ induces a map $f: B' \rightarrow B$, given by $f(p):= q$ where $q$ is as above, or alternatively by $f(p) = \pi(\bar{f}(v))$ for any $v \in F'_{p}$, the fibre of $E'$ over $p$.

Now, if we define $\pi_2: \fE \rightarrow E$ by $\pi_2(p,u) = u$, that gives us a bundle map $\pi_2$ from $\fE$ over $B'$ to $E$ over $B$.  The fact that $\pi_2$ is a bundle map means exactly that the following diagram commutes.
\begin{center}
\begin{tikzpicture}
  \flexsquare{\fE}{E}{B'}{B}{\pi_2}{f}{\pi_1}{\pi}
\end{tikzpicture}
\end{center}
This commutative diagram clearly follows from the definition of $\fE$ in Equation \ref{eqn:pullback}.

We will use pullback bundles in our construction of the important bundles $\Mvecbund$ and $\MPbund$ from $\bigvecbund$ and $\bigPbund$, respectively.  Our application of the pullback bundle is particularly straightforward because the map $f$ will simply be the inclusion map of $\M$ into the Grassmannian $\Grass$, to be defined in Section \ref{sec:Grass}.
%
%

\subsection{Connections and holonomy}
\label{sec:connection}

\subsubsection{Ehresmann connection and horizontal subspaces}

The transition functions tell us how to line up different bundle coordinate charts over a point.  We are also interested in relating the fibres at different points, and the notion of a connection tells us how to do that.

In this section, we assume that all our topological spaces are also smooth manifolds, so the notions of tangent spaces, etc.\ make sense.  If $M$ is a manifold, then $T_p M$ is the tangent space of the manifold at the point $p$, and $TM = \sqcup_{p} T_p M$ is the \emph{tangent bundle}, a vector bundle over $M$ with fibre isomorphic to $T_p M$ over each point $p$.

Let $E$ be a (smooth) fibre bundle with base space $B$, fibre $F$ and structure group $G$.  Take an element $u$ of $E$, and let $F_p$ denote the fibre over $p=\pi(u)$.  We now consider the tangent space $\TE{u}$ of the \emph{total} manifold $E$ over the point $u$.  The idea is to break the vector space $\TE{u}$ into a direct sum of a vertical subspace $\VE{u}$ and a horizontal subspace $\HE{u}$.  The subspace $\VE{u}$ will be canonically defined at each $u\in E$ as the tangent space to the submanifold $F_p$ of $E$, and we obtain the \emph{vertical bundle} $\VE{} = \sqcup_u \VE{u}$ over $E$.  However, there are many ways of choosing the subspace $\HE{u}$ while ensuring that $\TE{u} = \VE{u} \oplus \HE{u}$.  The Ehresmann connection is basically a ``good'' choice or a specification of such a subspace $\HE{u}$ at each $u\in E$.  We give the precise statements below.  The following definition is based on and generalizes the discussions in Sections 10.1.1 and 10.1.2 of \cite{Nakahara}.

\begin{Definition}[Ehresmann connection]
\label{def:connection}
An \emph{Ehresmann connection} on $E$ is a choice of piecewise smooth vector subbundle $\HE{} = \sqcup_u \HE{u}$ of $TE$ over $E$, called the \emph{horizontal bundle} (of the connection), whose fibres satisfy the following properties:
\begin{enumerate}
\item $\TE{u} = \HE{u} \oplus \VE{u}$ for every $u\in E$; \label{item:connection_basic}
\item $\HE{u}$ depends smoothly on $u$;
\item $\HE{}$ is compatible with the action of the structure group $G$ of the bundle. \label{item:struct_compat}
\end{enumerate}
$\HE{u}$ is the fibre of $\HE{}$ over $u \in E$ and $\VE{u}$ is the fibre of $\VE{}$ over $u$.
\end{Definition}

The precise formulation of property (\ref{item:struct_compat}) is most straightforwardly phrased in terms of parallel transport, so we delay this part of the definition until later in statement (\ref{sta:connection_compatibility}) of Section \ref{sec:horizontal_lift}, when we have defined parallel transport.

The projection $\pi$ provides us with a map from $\HE{u}$ to $T_p B$ when $\pi(u) = p$, if we keep in mind equivalence classes of curves in both cases.  $\HE{u}$ is a kind of lift of $T_p B$.  Property (\ref{item:connection_basic}) ensures that the map is an isomorphism.  Thus, we can think of the connection as a way of assigning a direction in the bundle for each direction in the base space.  That is, it tells us how the position in the fibre changes when we move in the base space and therefore how the fibres at different points connect up.

\subsubsection{Horizontal lift and parallel transport}
\label{sec:horizontal_lift}

A curve\footnote{Curves are implicitly piecewise-smooth, and in particular continuous, throughout this paper.} $\lgamma(t)$ in a fibre bundle $E$ is called a {\it lift} of the base space curve $\gamma(t)$ if $\pi \circ \lgamma(t) = \gamma(t)$ for all $t \in [0,1]$.  We shall define below lifts which are considered as ``horizontal''.  The following is a more general version of Definition 10.9 of \cite{Nakahara}.

\begin{Definition}[Horizontal lift]
Let $E$ be a piecewise smooth fibre bundle with an (Ehresmann) connection over the base space $B$, and let $\gamma: [0,1] \rightarrow B$ be a curve in $B$.  A curve $\lgamma: [0,1] \rightarrow E$ in $E$ is said to be a {\it horizontal lift} of $\gamma$ if $\pi \circ \lgamma = \gamma$ and the tangent vector to $\lgamma(t)$ always belongs to $\HE{\lgamma(t)}$.
\end{Definition}

\begin{Theorem}[Generalization of Theorem 10.10 of \cite{Nakahara}]
Let $\pi: E \rightarrow B$ be a piecewise smooth fibre bundle with an (Ehresmann) connection.  Let $\gamma: [0,1] \rightarrow B$ be a piecewise smooth curve in $B$ and let $u_0$ be an element of the fibre at $\gamma(0)$.  Then there exists a unique piecewise smooth horizontal lift $\lgamma(t)$ in $E$ such that $\lgamma(0) = u_0$.
\label{thm:horizontal_lift}
\end{Theorem}

Theorem~\ref{thm:horizontal_lift} can be proven using the existence of a unique solution to an appropriate differential equation.

From the unique horizontal lift $\lgamma(t)$ which satisfies $\lgamma(0) = u_0$, we obtain the element $u_1=\lgamma(1)$ in the fibre over $\gamma(1)$, called the {\it parallel transport} of $u_0$ along $\gamma$.  We define the following \emph{parallel transport map} along the curve $\gamma$:
\begin{eqnarray}
\Gamma(\gamma): F_{\gamma(0)} & \rightarrow & F_{\gamma(1)}\\
		  	  u_0 & \mapsto	    & u_1.
\label{eqn:parallel_transport}
\end{eqnarray}

Now we can precisely formulate property (\ref{item:struct_compat}) from Definition~\ref{def:connection}.  Take two trivializations $\phi_0$ and $\phi_1$ of $E$ over $U_0$ and $U_1$ containing $\gamma(0)$ and $\gamma(1)$ respectively.  Then with respect to the trivializations $\phi_0$ and $\phi_1$, we obtain isomorphisms $F_{\gamma(0)} \simeq F$ and $F_{\gamma(1)} \simeq F$.  That way, we can compare $F_{\gamma(0)}$ and $F_{\gamma(1)}$ directly via these trivializations.  Property~(\ref{item:struct_compat}) says that, with respect to this pair of trivializations:
\begin{equation}
\text{For any curve $\gamma(t)$ in the base space, $\Gamma(\gamma)$ should be a transformation in the structure group $G$.}
\label{sta:connection_compatibility}
\end{equation}
Furthermore, statement~(\ref{sta:connection_compatibility}) does not depend on the choice of trivializations.  That is, suppose $\Gamma(\gamma)$ does not transform fibre $F_{\gamma(0)}$ to fibre $F_{\gamma(1)}$ according to the action of some group element $g$ in $G$ with respect to $\phi_0$ and $\phi_1$, then with respect to another such pair $\psi_0'$ and $\psi_1'$ of trivializations, it also won't transform according to any group element of $G$.  On the other hand, if $\Gamma(\gamma)$ transforms the fibre as $g$ for one such pair of trivializations, then a different choice of trivializations will result in the action by a possibly different group element $g'$.  This follows from the property of trivializations given by property~(\rmnum{2}) of Definition~\ref{def:bund}.

Property (\ref{item:struct_compat}) of Definition \ref{def:connection} guarantees the following: If the structure group $G$ acts on the fibre $F$ as diffeomorphisms\footnote{In the main bundles $\bigvecbund$ and $\bigPbund$ considered in this paper, the structure group action on the fibre will be smooth.} (as opposed to only as homeomorphisms), then the map $\Gamma(\gamma)$ is a diffeomorphism with respect to the standard topology on the fibre, that is, smooth and with smooth inverse given by $\Gamma(\gamma^{-1})$, where $\gamma^{-1}: [0,1] \rightarrow B$ is a curve defined by $\gamma^{-1}(t) := \gamma(1-t)$ for $0 \le t \le 1$.  For a vector bundle, parallel transport should be a linear map; and for a principal bundle, parallel transport should perform multiplication by a group element of $G$.

We don't want to restrict attention to just smooth curves $\gamma(t)$, but we do want to use curves which are sufficiently well-behaved that parallel transport can be sensibly defined on them.  In particular, we shall work with piecewise smooth curves in this paper, which arise naturally from our constructions.

\subsubsection{Holonomy}

Let $(E,B,\pi,F,G)$ be a piecewise smooth fibre bundle, $\HE{}$ a choice of connection in $E$, and $\gamma: [0,1] \rightarrow B$ a \emph{loop} based at $p$, that is $p=\gamma(0)=\gamma(1)$.  Suppose $\lgamma$ is the (unique) horizontal lift of $\gamma$ with $\lgamma(0)=u_0 \in F_p$ with respect to this connection.  In general, $\lgamma(0)\ne \lgamma(1)$ in $F_p$.  This difference is the \emph{holonomy along the loop $\gamma$}, and the holonomy group, to be defined below, captures the range of such differences for all possible loops based at $p$.

For a loop $\gamma(t)$ with $\gamma(0)=\gamma(1)=p$, the parallel transport map gives a mapping of the fibre $F_p$ to itself, $\Gamma(\gamma): F_p \rightarrow F_p$.  Let $\Cp{B}$ denote the set of loops based at $p$ in $B$.  The {\it holonomy group} $\Hol_{p} (\HE{})$ is defined as
\begin{equation}
\Hol_{p}(\HE{}) := \{\Gamma(\gamma) \in \Diff{F_p} \;|\; \gamma \in \Cp{B}\},
\end{equation}
where $\Diff{F_p}$ stands for the group of diffeomorphisms\footnote{The map $\Gamma(\gamma)$ on $F_p$ is diffeomorphic as long as the structure group $G$ acts smoothly on the fibre $F$, which is indeed the case for piecewise smooth bundles.} of the fibre $F_p$, which is a smooth manifold.
While $\Cp{B}$ is not quite a group, the set $\Hol_{p}(\HE{})$ forms a group, specifically a subgroup of $\Diff{F_p}$.  We can see this by noting that if $\alpha$ and $\beta$ are loops at $p$, so is $\alpha*\beta$ ($\alpha$ followed by $\beta$, with reparameterization); therefore $\Gamma(\beta) \Gamma(\alpha) = \Gamma(\alpha*\beta)$ is also in $\Hol_{p}(\HE{})$.  Also $\mathrm{id} \in \Hol_{0}(\HE{})$, because $\Gamma(\gamma_0)$, where $\gamma_0$ is the constant loop $\gamma_0 (t) = p$, yields the identity element in $\Diff{F_p}$.  The inverse of an arbitrary diffeomorphism $\Gamma(\gamma)$ is given by $\Gamma(\gamma^{-1})$.

If the base space $B$ is path-connected, that is between any two points $p_0$ and $p_1$ there is a curve $\gamma$ in $B$ with $\gamma(0)=p_0$ and $\gamma(1)=p_1$, then the holonomy groups $\Hol_{p_0}(\HE{})$ and $\Hol_{p_1}(\HE{})$ at different base points are isomorphic via conjugation as follows:
\begin{equation}
\Hol_{p_1}(\HE{}) = \Gamma(\gamma) \Hol_{p_0}(\HE{}) \Gamma(\gamma)^{-1},
\label{eqn:conj}
\end{equation}
where $\gamma$ can be any such curve.  Different choices of $\gamma$ simply yield different isomorphisms between $\Hol_{p_0}(\HE{})$ and $\Hol_{p_1}(\HE{})$. Note that $\Gamma(\gamma)^{-1} = \Gamma(\gamma^{-1})$.

The discussion above does not involve or depend on a choice of trivialization.  However, if we choose a trivialization $\phi$ over $U \ni p$, we obtain an isomorphism $F_p \simeq F$.  Via the trivialization $\phi$ and the fibre identification, we can view the holonomy group $\Hol_{p}(\HE{})$ as a subgroup\footnote{Recall that the action of the structure group $G$ on the fibre $F$ is faithful.} of the structure group $G$, and simultaneously we can view $G$ as a subgroup of $\Diff{F_p}$.  So, we have
\begin{equation}
\Hol_{p_1}(\HE{}) < G < \Diff{F_p}
\end{equation}
upon choosing a trivialization over $U \ni p$.  Note that the setup here is a special case to the discussion in Section \ref{sec:horizontal_lift}.  Here, we only need to choose a single trivialization to interpret the parallel transport around a loop as a group element of $G$, whereas the parallel transport along an arbitrary curve in Section \ref{sec:horizontal_lift} requires the choice of a pair of trivializations in general.

In the case of a vector bundle, $\Hol_{p}(\HE{}) < \GL{F_p} < \Diff{F_p}$ upon choosing a trivialization over $U \ni p$, since the structure group in this case is $\GL{F_p}$.

\subsubsection{Flatness of a connection and monodromy}
\label{sec:monodromy}

In this paper, we focus on (projective) connections on fibre bundles which are (projectively) flat.  The projectiveness of a connection or of the flatness of a connection will be discussed in some details in Section \ref{sec:projective}, so here we simply focus on defining and analyzing the flatness of a connection.  When we have a flat connection, the holonomy is known as the \emph{monodromy}.  Monodromies are particularly interesting because they are topologically robust, in that homotopic loops produce the same fibre transformations over the base point.  We will return to this point and its physical implications in Section~\ref{sec:FTmonodromy}.

Here we choose to define flatness of a connection indirectly via the triviality of the restricted holonomy group, which we define below.  The definition sets up the stage for the ``robustness'' we have just described.

\begin{Definition}[Restricted holonomy group]
Given $(E,B,\pi,F,G)$ a piecewise smooth fibre bundle, a connection $\HE{}$ on it, and $\Gamma(\gamma)$ the parallel transport along $\gamma$ as defined by equation (\ref{eqn:parallel_transport}), the restricted holonomy group at a point $p\in B$ is defined as
\begin{equation*}
\Hol_{p}^{0}(\HE{}) := \{\Gamma(\gamma)\in \Diff{F_p} \;|\; \gamma \in \Cp{B} \text{ contractible in }B\}.
\end{equation*}
\end{Definition}

Clearly, $\Hol_{p}^0(\HE{})$ is a subgroup of $\Hol_p(\HE{})$.  We are now ready to define the flatness of a connection.

\begin{Definition}[Flat connection]
A connection $\TE{u} = \HE{u} \oplus \VE{u}$ on a fibre bundle $E$ is called \emph{flat} if the restricted holonomy group $\Hol_{p}^{0}(\HE{})$ at point $p$ is trivial for \emph{all} $p\in B$. (Or equivalently, for a fixed $p$ in each path-connected component, by equation (\ref{eqn:conj}).)
\label{def:flat}
\end{Definition}

If the connection $\HE{}$ is flat, then for any $\gamma$ in $B$ based at $p$, the transformation $\Gamma(\gamma)$ of the fibre $F_p$ is only dependent on the homotopy class of $\gamma$. This is because given two homotopic loops $\gamma$ and $\gamma'$, $\gamma' * \gamma^{-1}$ is a contractible loop, so
\begin{equation}
\Gamma(\gamma)^{-1} \Gamma (\gamma') = \Gamma(\gamma' * \gamma^{-1}) = e.
\end{equation}

We explore the implications of this analysis.  For a general (not necessarily flat) connection, we only have a homomorphism $\Cp{B} \rightarrow \Hol_p(\HE{})$.  However if $\HE{}$ is a flat connection, there is a homomorphism\footnote{This homomorphism is surjective, as is the other homomorphism mentioned in this paragraph.} $\pi_{1}(B,p) \rightarrow \Hol_{p}(\HE{})$, where $\pi_{1}(B,p)$ is the fundamental group of manifold $B$ based at the point $p$.  This action of $\pi_{1}(B,p)$ on $F$ is called the \emph{monodromy action} (or \emph{monodromy representation} in the case of a vector bundle), and in this case, the group $\Hol_{p}(\HE{})$ is also known as the \emph{monodromy group}.  In other words, the flat connection lets us complete the following commutative diagram:
\begin{center}
\begin{tikzpicture}
\node (start) at (0,1) {$\Cp{B}$};
\node (down) at (0,0) {$\pi_{1}(B,p)$};
\node (end) at (2.5,0) {$\Hol_{p}(\HE{})$};

\draw [->>,anchor=west] (start)  -- (end);
\draw [->>,anchor=south] (start) -- (down);
\draw [->>,anchor=west,dotted] (down) -- (end);

\end{tikzpicture}
\end{center}

\begin{Remark}
For physicists, it may be worth noting that a non-flat connection is what gives rise to the geometric or Berry phase.  The base space in this case is the manifold of gapped Hamiltonians, with fibre at $H$ equal to the ground state space of $H$.
\end{Remark}

\subsubsection{Projective connections, flat projective connections and projectively flat connections}
\label{sec:projective}

A subgroup $Z$ of a group $G$ is called \emph{central} if it lies in the centre of $G$.  For $G = \GL{K}$, an obvious choice of a central subgroup is $Z = \{\lambda \idmatrix \;|\; \lambda \in \fieldC, |\lambda| \ne 0\} = \Cstar = \GL{1}$.  In the case of interest to us, $G = \U{K}$ and $Z = \{\lambda \idmatrix \;|\; \lambda \in \fieldC, |\lambda| = 1\} = \U{1}$.

\begin{Definition}[Projective fibre bundle]
Now, suppose $E$ is a fibre bundle with structure group $G$.  Upon choosing $Z$, we can create a \emph{projective fibre bundle} $\projectiveE$ from the bundle $E$ by choosing a trivialization and identifying points within each fibre that differ only by the action of $Z$.
\end{Definition}

Since $Z$ sits in the center of $G$, the identification we made to obtain $\projectiveE$ from $E$ is actually independent of our choice of trivializations.  We can see this by letting $(U_i, \phi_i)$ and $(U_j, \phi_j)$ be two trivializations such that $U_i \cap U_j \ne \emptyset$, and $p \in U_i \cap U_j$.  We wish to show that for a pair of arbitrary points $u_1$ and $u_2$ in the fibre $\pi^{-1}(p)$, $u_1 \sim u_2$ with respect to $\phi_j$ implies that $u_1 \sim u_2$ with respect to $\phi_i$.  From the definition of the projective bundle $\projectiveE$, we know that $u_1 = \phi_j (p,f)$, where $f \in F$, and $u_2$ are identified with respect to $\phi_j$ if and only if $u_2 = \phi_j (p, z \cdot f)$ for some $z \in Z$.  On one hand, $u_1 = \phi_j (p, f) = \phi_i (p, \theta_{ij}(p) \cdot f)$ using the transition function $\theta_{ij}$ as in Definition \ref{def:bund}; on the other hand, $\phi_j (p, z \cdot f) = \phi_i (p, \theta_{ij}(p) \cdot z \cdot f) = \phi_i (p, z \cdot \theta_{ij}(p) \cdot f)$ due to the fact that $z \in Z$ is a central element in $G$.  Therefore, we see quite clearly that if $u_1 \sim u_2$ with respect to $\phi_j$, then $u_2 = \phi_i (p, z \cdot \theta_{ij}(p) \cdot f)$, and $u_1 \sim u_2$ with respect to $\phi_i$.  Similarly, using the transition element $\theta_{ji}(p) = \theta_{ij}(p)^{-1}$, we can show that for $u_1$ and $u_2$ in the fibre $\pi^{-1}(p)$, $u_1 \sim u_2$ with respect to $\phi_i$ implies that $u_1 \sim u_2$ with respect to $\phi_j$.  Thus, we have proven that the identification to obtain the projective version $\projectiveE$ from $E$ is independent of the choice of trivializations.

In this paper, we shall use the term \emph{projective vector bundle} to mean the projectivized version $\projectiveE$ of a vector bundle $E$.  Note that even if $E$ is a vector bundle, the fibres of $\projectiveE$ might not be projective spaces, as in the example where the fibre is $F = \Cvect{K}$ and the central subgroup is $Z = \U{1}$.

Next, we define a projective connection.

\begin{Definition}[Projective connection]
A \emph{projective connection} on $E$ is a connection on $\projectiveE$.
\end{Definition}

Let us examine the differences between a connection on $E$ and a projective connection on $E$.

A connection on $E$ always yields a connection on $\projectiveE$.  This is essentially due to property (\ref{item:struct_compat}) of Definition \ref{def:connection} and the fact that $Z$ is a central subgroup of $G$.  Here, we do not give a full proof, but shall show only that a parallel transport in $E$ projects down to a parallel transport in $\projectiveE$, that is if $u_0 \sim u_0'$ in $F_{\gamma(0)}$, then $\Gamma(\gamma) u_0 \sim \Gamma(\gamma) u_0'$ in $F_{\gamma(1)}$ for any curve $\gamma$ in $B$.  Let us choose two trivializations $(U_0, \phi_0)$ and $(U_1, \phi_1)$ where $U_0$ contains $\gamma(0)$ and $U_1$ contains $\gamma(1)$ respectively.  Suppose $u_0 = \phi_0 (\gamma(0), f) \sim u_0'$ in $F_{\gamma(0)}$, then $u_0' = \phi_0 (\gamma(0), z \cdot f)$ for some $z \in Z$.  By property (\ref{item:struct_compat}) of Definition \ref{def:connection}, $\Gamma(\gamma) u_0 = \phi_1 (\gamma(1), g \cdot f)$ and $\Gamma(\gamma) u_0' = \phi_1 (\gamma(1), g \cdot z \cdot f)$ for some $g \in G$.  Now, since $z$ is a central element, we can rewrite $\Gamma(\gamma) u_0'$ as $\Gamma(\gamma) u_0' = \phi_1 (\gamma(1), z \cdot g \cdot f)$, and it becomes clear that $\Gamma(\gamma) u_0 \sim \Gamma(\gamma) u_0'$ in $F_{\gamma(1)}$.

On the other hand, a connection on $\projectiveE$ does not by itself yield a connection on $E$, but it is very close to being one.  In fact, parallel transport in $\projectiveE$ almost translates to parallel transport in $E$, but with an additional ambiguity up to $Z$.  That is, given a curve $\gamma$ in the base space $B$, the parallel transport map $\tilde{\Gamma}(\gamma)$ in $\projectiveE$ yields a one-to-many map $\Gamma(\gamma)$ from the fibre $F_{\gamma(0)}$ to the fibre $F_{\gamma(1)}$ in $E$, taking $u_0 \in F_{\gamma(0)}$ to $\{ \tilde{\Gamma}(\gamma) [u_0] \} = \{ z \cdot u_1 \;|\; z \in Z \}$ where $u_1 \in F_{\gamma(1)}$ is an element of $E$ in the equivalence class $\tilde{\Gamma}(\gamma) [u_0]$.


In this paper, we often prefer to work with the original fibre bundle $E$ rather than $\projectiveE$, and talk about projective connections on $E$ rather than connections on $\projectiveE$, in order to match the common convention in the field of quantum information.

For an honest connection, we define the notion of projective flatness for a connection on $E$, to be used in Section \ref{sec:topologicalFT}, as follows:

\begin{Definition}[Projectively flat connection]
Let $E, G, Z$ be as above, and let $\HE{}$ be a connection on the fibre bundle $E$.  We say that the connection $\HE{}$ is \emph{projectively flat} if, for all points $p \in B$ (or equivalently, for just one point), the restricted holonomy group $\Hol_{p}^{0}(\HE{}) < G$ satisfies
\begin{equation}
\Hol_{p}^{0}(\HE{}) < Z \cdot {e} < G,
\end{equation}
where $e$ is the identity element of $G$.
\label{def:projectively_flat}
\end{Definition}

Contrasting with Definition \ref{def:flat}, we see that this notion is a relaxation of the flatness condition, by the central subgroup $Z$.  In particular, if $G$ is a matrix group and $Z = \U{1}$, projective flatness would mean $\Hol_{p}^{0}(\HE{}) < \U{1} \cdot \idmatrix$, where $\idmatrix$ is the identity operator.  This is the case that concerns us in the present work.

We shall also need the following notion, which invokes Definition \ref{def:flat} of a flat connection, in Section \ref{sec:transversal}.

\begin{Definition}[Flat projective connection]
A \emph{flat projective connection} on $E$ is a flat connection on $\projectiveE$.
\label{def:flat_projective}
\end{Definition}

Comparing Definition \ref{def:flat_projective} with Definition \ref{def:projectively_flat}, we see that they are actually quite similar in terms of the conditions imposed upon the respective restricted holonomy groups.  This is because $\Hol_{p}^{0}(H\projectiveE{}) = [e]$ is essentially equivalent to saying $\Hol_{p}^{0}(\HE{}) < Z \cdot e$.  The main difference is then that a flat projective connection is only a projective connection, whereas a projective flat connection is a true connection.


\section{A geometric picture for unitary evolutions of QECCs}
\label{sec:geoUnitary}

In this section, we begin to paint the geometric picture central to this paper.  We explain how a few key objects in the theory of quantum error correction and fault tolerance relate to natural geometric constructions.  We first describe in Section \ref{sec:codes_Grass} a natural manifold structure on the set of all codes with a fixed logical dimension $K$ and physical dimension $N$.   We then give the set of all codewords in these codes a vector bundle structure in Section \ref{sec:codewords_bigvecbund}, and the set of all encodings for these codes the structure of a principal $\U{K}$-bundle in Section \ref{sec:encodings_bigPbund}.  Eventually, in Section \ref{sec:geoFT}, we want to construct a smaller vector bundle and a smaller principal $\U{K}$-bundle for every family of fault-tolerant unitary operations, along with a QECC.  

\subsection{Codes and the Grassmannian manifold $\Grass$}
\label{sec:codes_Grass}

Let us define the notion of \emph{code dimension} as a pair of integers $(K, N)$ where $K$ gives the dimension of the codespace and $N$ is the dimension of the physical Hilbert space.  Below, we give a manifold structure on the set of all QECCs of a fixed code dimension $\vec{d} = (K, N)$.  In particular, this set can be endowed with the manifold structure of $\Grass$, the Grassmannian manifold of $K$-dimensional subspaces in $\Cvect{N}$.  We introduce the Grassmannian manifold and make the aforementioned identification in Section \ref{sec:Grass}, establishing a ``static'' geometric picture; in Section \ref{sec:path.Grass}, we capture some dynamics on the manifold by considering unitary evolutions acting on the manifold of codes.

\subsubsection{Space of $(K,N)$-dimensional codes as the Grassmannian manifold $\Grass$}
\label{sec:Grass}

In this section, we introduce the Grassmannian manifold, which can be identified with the set of all $(K,N)$-dimensional codes.  This means we can think of the set of codes as the \emph{manifold of codes}.  The specification of a code dimension $\vec{d} = (K, N)$ is always implicit.

As a set, the Grassmannian is defined as
\begin{equation*}
\Grass := \{\text{All $K$-dimensional linear subspaces of $\Cvect{N}$}\}.
\end{equation*}
We will put a manifold structure on $\Grass$ by first describing it as a ``homogeneous space'' and then as a coset space of a closed (Lie) subgroup in the Lie group $\U{N}$, giving $\Grass$ the structure of a smooth manifold.

\begin{Definition}
\label{def:homogeneous}
A \emph{homogeneous space} for a group $G$, or a \emph{$G$-space}, is a non-empty set  with a transitive $G$ action.  
\end{Definition}

The Grassmannian is a homogeneous space for the group $G = \U{N}$, where the action of $A \in \U{N}$ on a point $V \in \Grass$ is simply given by the application of linear transformation $A$ to $\Cvect{N}$, transforming the $K$-dimensional subspace $V \subseteq \Cvect{N}$ to the subspace $A(V)$.  It is not difficult to convince yourself that this action is transitive on $\Grass$, namely between any two $K$-planes $V$ and $W$, there always exists some unitary operator $A \in \U{N}$ such that $A(V) = W$.  For any $G$-space $S$, we can always express it as the left cosets of the subgroup $G_x$ in $G$, where $G_x$ is the stabilizer in $G$ of any point $x \in S$.  We identify the subgroup $G_x$ for the $\U{N}$-space $\Grass$ below, culminating in Theorem \ref{thm:Grassmap}.

First we clarify a couple of notations used below.  Let $\U{N}$, as usual, denote the Lie group of all unitary operators on $\Cvect{N}$.  Given a $K$-dimensional vector subspace $V \subset \Cvect{N}$, we can decompose $\Cvect{N}$ into orthogonal subspaces
\begin{equation}
\Cvect{N} = V \oplus V^{\bot},
\end{equation}
where $V^{\bot}$ is the orthogonal complement of $V$ in $\Cvect{N}$ with respect to the standard inner product.
We define $\Ub{V}$ as the subgroup of operators on $\Cvect{N} = V \oplus V^{\bot}$ of the form $S' \oplus \idmatrix_{V^{\bot}}$, where $S'$ is a unitary operator on the subspace $V$ and $\idmatrix_{V^{\bot}}$ is the identity transformation on its orthogonal complement $V^{\bot}$.  It is pretty clear that $\Ub{V} \subset \U{N}$ and $\Ub{V} \simeq \U{K}$.  Similarly, $\Vperp$ is defined as the subgroup of operators of the form $\idmatrix_{V} \oplus S''$ where $S''$ acts unitarily on the subspace $V^{\bot}$.  Analogously, we have $\Vperp \subset \U{N}$ and $\Vperp \simeq \U{N-K}$.

The following theorem and corollary are standard properties of the Grassmannian (see, e.g., Sec. 3.65(f) of \cite{Warner}).  We present them without proof.

\begin{Theorem}
Fix a $K$-plane $V$ in $\Cvect{N}$.  There is a bijection of sets
\begin{equation}
\Grass \simeq \NmodVVperp,
\label{eqn:Grass}
\end{equation}
identifying $\Grass$ with the coset space on the right.  For any $W \in \Grass$, let $A^{W}$ be any unitary operator satisfying $W=A^{W}(V)$.  The explicit bijective map is given by
\begin{equation}
W \mapsto [A^{W}] := \{ A^{W} (S' \oplus S'') \;|\; S' \in \Ub{V} \text{ and } S'' \in \Vperp \} = A^W \left(\VVperp \right),
\label{eqn:Grassmap}
\end{equation}
mapping $W$ to the left coset of $\VVperp$ in $\U{N}$ containing the operator $A^W$.  
\label{thm:Grassmap}
\end{Theorem}

\begin{Remark}
In fact, the left coset $[A^W]$ contains precisely all those unitary operators on $\Cvect{N}$ which take the subspace $V$ into the subspace $W$, as summarized by the following equation
\begin{equation}
[A^{W}] = \{A \in \U{N} \;|\; A(V) = W \}.
\end{equation}
\end{Remark}

The natural action of $\U{N}$ on $\Grass$ for $g \in \U{N}$ and $W \in \Grass$ is given by
\begin{equation}
g \cdot W := g(W).
\label{eqn:actionGrass}
\end{equation}
It follows from a general fact about homogeneous spaces that this $\U{N}$ action on $\Grass$ agrees precisely with the natural $\U{N}$ action on $\flatNmodVVperp$ under the homogeneous space identification of equations (\ref{eqn:Grass}) and (\ref{eqn:Grassmap}).  Via this same identification, we arrive at the following corollary.

\begin{Corollary}
$\Grass$ is a smooth manifold.  The $\U{N}$ action on the Grassmannian $\Grass$ is smooth.
\label{cor:smoothGrass}
\end{Corollary}

Recall from Section \ref{sec:QECCs} that an $((n,K))$ qudit quantum error-correcting code is simply a $K$-dimensional subspace in $\Hilbphy = \nqudits$.  In other words, the code dimension of such a QECC is then $(K,N)$ where $N = d^n$.  Therefore, the Grassmannian $\Grass$, where $N=d^n$, is also equal to the set of $((n,K))$ qudit QECCs, so we may equally well think of $\Grass$ as the \emph{manifold of codes}.

\subsubsection{Unitary evolution of codes yielding paths in $\Grass$}
\label{sec:path.Grass}

In this section, our goal is to study unitary evolutions of codes, with special attention to their geometric meanings.  We now concern ourselves with the action of not only a single unitary, but a piecewise smooth one-parameter family of unitaries.  Physically, this difference corresponds to the difference between a one-shot implementation of a unitary and a continuous implementation thereof, which is typically more realistic.

Let us first define the notion of unitary evolution.

\begin{Definition}
A \emph{unitary evolution} is a one-parameter family $U(t)$ of unitary operators such that, at time 0, $U(0) = \idmatrix$, and as time passes, $U(t)$ evolves piecewise smoothly with time, until at time 1, it accomplishes some target unitary $U(1)=U$.
\label{def:uni_evolution}
\end{Definition}

Since we are interested in understanding the action of such a $U(t)$ on the Grassmannian or the manifold of codes, we will presumably be studying the trajectory $U(t) (V)$ of $K$-planes in $\Cvect{N}$ for all points $V$ in $\Grass$.  However, since the action of $\U{N}$ is transitive on $\Grass$, it's sufficient, both mathematically and for fault tolerance purposes, to understand the action of unitary evolutions on a single point $C \in \Grass$, which we call the \emph{base point}.  In the context of fault tolerance, the base point of our choice corresponds to the QECC employed in a fault-tolerant protocol, which we shall call the \emph{reference code}.  Mathematically, suppose we choose a different base point $C'$.  Then there is always a unitary $A \in \U{N}$ for which $A(C) = C'$, and the homotopy classes of loop trajectories based at $C'$ are in one-to-one correspondence with the homotopy classes of loops based at $C$.  More precisely, any loop $\gamma(t) := U(t) (C)$ at $C$ can be turned into a loop $\gamma'(t) = U'(t)(C)$ based at $C'$ via $U'(t) = A U(t) A^{-1}$ for all $t \in [0,1]$.  Mapping the homotopy class $[\gamma(t)]$ to $[\gamma'(t)]$ gives a \emph{group isomorphism} between $\pi_1(\Grass, C)$, the fundamental group of $\Grass$ with base point $C$, and $\pi_1(\Grass, C')$, the fundamental group with base point $C'$.  Therefore, it is sufficient to work with just base point $C$.

The fact, stated in Corollary \ref{cor:smoothGrass}, that the action of $\U{N}$ on $\Grass$ is smooth means that the map $\U{N} \times \Grass \rightarrow \Grass$ is a smooth map from the smooth product manifold on the left to the smooth manifold on the right.  In particular, this map restricted to $\U{N} \times \{C\}$ still gives a smooth map into $\Grass$.  Therefore we have that a piecewise smooth family of unitaries $\{U(t) \;|\; 0 \le t \le 1 \}$ gives rise to a piecewise smooth curve $\gamma(t) = U(t) (C)$ in $\Grass$.%
\footnote{A parametrized unitary path $U(t)$ or curve $\gamma(t)$ is smooth if all $t$-derivatives are defined.  This includes the possibility of derivatives being $0$, so it is possible for a smooth curve to ``stop'' and then resume in a different direction, provided it stops and re-starts smoothly.  It is piecewise smooth if it can be decomposed into a finite sequence of smooth curves.}
Since $U(0) = \idmatrix$ and $U(1) = U$ for some target unitary $U$, $\gamma(t)$ is a piecewise smooth curve from the point $C$ to $U(C)$ in the Grassmannian manifold.  Note that $\gamma(t)$ may be self-intersecting even when $U(t)$ is not, as depicted in Figure~\ref{fig:unitarypath} of Section \ref{sec:restriction_bundle}.

\subsection{Codewords and the tautological vector bundle $\bigvecbund$}
\label{sec:codewords_bigvecbund}

In this section, we describe the first of the two fibre bundles over the Grassmannian manifold which will be important to us.  We start with the simpler one, where the fibres have the structure of a vector space.  This fibre bundle can hold more information about our quantum system than just its base space, the Grassmannian; in particular, it allows us not only to store information about where the subspace gets mapped to by a given unitary evolution, but also to keep track of the evolution of individual vectors in the subspaces.

\subsubsection{Space of codewords as the tautological vector bundle $\bigvecbund$}
\label{sec:codewords_bigvecbund2}

Here we describe a vector bundle $\bigvecbund$ over the Grassmannian.  All the vector bundles considered for the concrete quantum information applications in this paper are pullbacks of the vector bundle $\bigvecbund$ to some subset of the Grassmannian.  Hence we sometimes refer to the bundle $\bigvecbund$ as the ``big vector bundle''.  Mathematically, $\bigvecbund$ is known as the tautological vector bundle, defined below.

\begin{Definition}[Tautological vector bundle $\bigvecbund$ over $\Grass$]
The \emph{tautological vector bundle} $\bigvecbund$ over $\Grass$ is the vector bundle whose fibre at a point $W \in \Grass$ is simply the $K$-plane $W$, and the total space $E$ is constructed by taking the disjoint union of all such fibres.  All the fibres $F_W = \pi^{-1}(W) = W$ are isomorphic to $F = \Cvect{K}$ (upon choosing a trivialization), and the structure group $G = \U{K}$ acts on the fibre $\Cvect{K}$ from the left in the obvious manner.
\end{Definition}

%
%
By a \emph{codeword}, we mean a pair $(C, w)$ where $C$ is a code and $w$ a vector in $C \subset \Cvect{N}$.  Clearly, we have a one-to-one correspondence between the codewords and the elements of the tautological vector bundle: $C$ tells us which fibre of $\bigvecbund$ the codeword belongs to, and the vector $w$ tells us which point in that fibre $(C, w)$ represents.

Note that for simplicity, our definition of a codeword did not take into account the projective equivalence relation
\begin{equation}
w \sim w' \text{\;\; if $w = \xi w'$ for some $\xi \in \Cstar$},
\end{equation}
which is appropriate for quantum states.
%
%

\subsubsection{Unitary evolution of codewords as paths in $\bigvecbund$}
\label{sec:Upath.V}

We wish to consider the action of a unitary evolution of $\Cvect{N}$ on the space of codewords $\bigvecbund$.  Let us first briefly consider the action of one single unitary operator $U \in \U{N}$ on the space.  Let $u = (C,w)$ be a codeword in $\bigvecbund$.  The action of $U$ results in the new codeword $(U(C), U(w))$, which stands for ``the vector $U(w)$ in the code $U(C)$''.  We see that a unitary action by $U$ maps any fibre $F_C$ of $\bigvecbund$ to the fibre $F_{U(C)}$, so $U$ induces a smooth\footnote{The smooth action of $G = \U{N}$ on the topological space $F = \Cvect{K}$ guarantees the smoothness within fibres, and the smoothness of the action across fibres follows from the fact that a unitary $U$ maps the fibre over $C$ to the fibre over $U(C)$ and Corollary \ref{cor:smoothGrass}.} bundle map from $\bigvecbund$ to itself, which we call $f_U$.  Since $U$ is invertible, the map $f_U$ is actually a diffeomorphism as it is also invertible, with inverse given by $f_{U^{-1}}$.

Now, let $U(t)$ be a unitary evolution and let $(C, w)$ be a codeword in the reference code $C$.  If we define 
\begin{equation}
\eta^{\{U\}}(t) := f_{U(t)}(C, w) = (U(t)(C), U(t)(w)), 
\label{eq:vectorbundlelift}
\end{equation}
we see that $\eta^{\{U\}}(t)$ is a path in the bundle $\bigvecbund$ which projects via the bundle projection $\pi$ to the path $\gamma(t) = U(t) (C)$ in the base Grassmannian manifold.

As discussed in Section~\ref{sec:path.Grass}, if the one-parameter family of unitaries $U(t)$ is piecewise smooth in $\U{N}$, which is the case for unitary evolutions, then the path $\gamma(t)$ is piecewise smooth in $\Grass$.  This is a consequence of Corollary~\ref{cor:smoothGrass}, which states that the action of $\U{N}$ on $\Grass$ is smooth.  Similarly, the path $\eta^{\{U\}}(t)$ is also piecewise smooth in $\bigvecbund$, because the $\U{N}$ action on $\bigvecbund$ is also smooth.

\subsection{Encodings and the tautological principal bundle $\bigPbund$}
\label{sec:encodings_bigPbund}

In Section \ref{sec:codewords_bigvecbund}, we developed a geometric picture for the space of codewords $\bigvecbund$ and the unitary action on it.  In this section, we shall build a similar picture for the space of encodings, also giving it the structure of a fibre bundle.  However, this one will not be a vector bundle but one with fibres diffeomorphic to the structure Lie group $\U{K}$ itself.  That is, it is a principal $\U{K}$-bundle.  In a way, this second fibre bundle allows us to capture the unitary evolution of QECCs more efficiently, in that it captures the evolution of ordered orthonormal bases or orthonormal $K$-frames in QECCs instead of the evolution of one codeword at a time.

\subsubsection{The tautological principal bundle $\bigPbund$ over $\Grass$}
\label{sec:bigPbund}

The mathematical hero this time is the tautological principal $\U{K}$-bundle, which is the associated bundle of $\bigvecbund$ with fibre $\U{K}$.  Just like for $\bigvecbund$, the principal $\U{K}$-bundle $\bigPbund$ can be thought of as the ``big principal bundle'' for our framework, because all the principal bundles we will consider in unitary fault tolerance applications are pullbacks or restrictions of the bundle $\bigPbund$ to some subset of the base manifold $\Grass$.

Here we shall give two descriptions of the bundle $\bigPbund$.  In the first approach, we define the bundle using the language of Lie groups.  In the second approach, which we will take in Section~\ref{sec:encodings}, we describe the same bundle in terms of $K$-frames.  To juxtapose the two descriptions, we can think of the first one as the ``operator'' perspective and the second one as the ``state'' perspective.

\begin{Definition}[Tautological principal bundle $\bigPbund$ over $\Grass$]
\label{def:bigPbund}
Fix a particular $K$-dimensional subspace $V \subset \Cvect{N}$.  The \emph{tautological principal bundle} $\bigPbund$ over $\Grass$ is the $\U{K}$-principal bundle by choosing an identification\footnote{That is, we want to choose an encoding for the code $V$.} $\Cvect{K} \simeq V$ once and for all, and first describing $\bigPbund$ as a $\Ub{V}$-principal bundle with
\begin{enumerate}
\item Total space $E$: the smooth manifold $\flatNmodVperp$;\footnote{$\flatNmodVperp$ is a \emph{smooth manifold} because $\Vperp$ is a closed subgroup of $\U{N}$.  However, $\flatNmodVperp$ is not a Lie group because $\Vperp$ is not a normal subgroup of $\U{N}$.} \label{item:totalP}
\item Fibre $\pi^{-1}(W)$ over $W$: Consists of those left cosets of $\Vperp$ in $\U{N}$ which contain unitary operators $A$ for which $A(V) = W$.  Such cosets are denoted as $\cosetVperp{A}$; \label{item:fibreP}
\item A trivialization-dependent left $\U{K}$-action on the fibre $\pi^{-1}(W)$: Choose a trivialization at $W$, namely an isomorphism $\pi^{-1}(W) \xrightarrow{\sim} \Ub{V}$, by choosing a particular coset $\cosetVperp{A}$ and using the following map
\begin{equation}
\cosetVperp{B} \mapsto A^{-1} B |_{V} \oplus \idmatrix_{V^{\perp}}.
\label{eqn:trivializationP}
\end{equation}
Then $\Ub{V}$ acts on $\pi^{-1}(W)$ via this trivialization, explicitly as follows:
\begin{equation}
Y \cdot \cosetVperp{B} = \cosetVperp{(A Y A^{-1} B)}
\label{eqn:actionP}
\end{equation}
for $Y$ in $\Ub{V}$.  It's not hard to see that both the right hand side of Equation (\ref{eqn:trivializationP}) and (\ref{eqn:actionP}) depend only on $A|_{V}$ and not on $A|_{V^{\perp}}$.
\label{item:actionP}
\end{enumerate}
\end{Definition}

After our first definition of $\bigPbund$ above, let us try to understand and expand on it.  Property (\ref{item:totalP}) of Definition \ref{def:bigPbund} reads
\begin{equation}
\bigPbund := \NmodVperp = \{\text{Left cosets $\cosetVperp{A} \;|\; A \in \U{N}$}\}
\label{eqn:bigPbund}
\end{equation}
which bears resemblance to Equation (\ref{eqn:Grass}) of Theorem \ref{thm:Grassmap}, formulating $\bigPbund$ as a homogeneous space with a transitive left action of $\U{N}$.  Hence, parallel to Corollary \ref{cor:smoothGrass}, we have that $\bigPbund$ is a smooth manifold, with a smooth $\U{N}$ action.

As a set, and as discussed in property (\ref{item:fibreP}) of Definition \ref{def:bigPbund}, the total space $\bigPbund$ can be decomposed as
\begin{eqnarray}
\bigPbund &=& \bigsqcup_{W \in\Grass} \{\text{Left cosets $\cosetVperp{A} \;|\; A \in \U{N}$ and $A(V) = W$}\} \label{eqn:bigPbund.fibre} \\
		& =: & \bigsqcup_{W \in\Grass} F(W),
\label{eqn:decomposeP}
\end{eqnarray}
where the definition of fibres $F(W)$ is self-evident.  This is the fibre decomposition for $\bigPbund$.  We saw in Definition \ref{def:bigPbund}, especially in property (\ref{item:actionP}), that $F(W)$ is diffeomorphic to $\U{K}$.  This is a necessary condition for $\bigPbund$ to be a principal $\U{K}$-bundle.

\begin{Remark}
Instead of viewing elements of the fibre $F(W)$ as left cosets $\cosetVperp{A}$ for unitary operators $A \in \U{N}$ which in particular takes $V$ to $W$, we can instead think of them as the different unitary transformations $\overline{A} := A|_{V}: V \rightarrow W$ obtained through restricting the operator $A$ to the subspace $V$.  This way, we have exactly one $\overline{A}$ for each coset $\cosetVperp{A}$.
\end{Remark}

%
%
%

\subsubsection{Space of encodings as an alternative description $\Stiefel$ of $\bigPbund$}
\label{sec:encodings}

Our goal in this section is to relate the physical concept of QECC encodings to a mathematical object, the smooth fibre bundle $\bigPbund$.

Let us recall and elaborate on the concept of encoding for a QECC $C$, introduced in Section \ref{sec:QECCs}.  We gave a very brief characterization of an encoding for a $(K,N)$-dimensional code $C$ as a unitary isomorphism $\iota: \Cvect{K} \rightarrow C \subset \Cvect{N}$.  
Let us call any encoding of a $K$-dimensional logical Hilbert space into an $N$-dimensional physical Hilbert space a \emph{$(K,N)$-dimensional encoding}.  Then any $(K,N)$-dimensional encoding is given by a partial isometry\footnote{For us, a \emph{partial isometry} is a linear map from a vector space $A$ to a vector space $B$ which preserves the inner product.  If the vector spaces $A$ and $B$ are complex vector spaces, then a partial isometry differs from a unitary transformation in that it does not have to be an isomorphism.  It must be injective, but might not be surjective.} $\iota: \Cvect{K} \hookrightarrow \Cvect{N}$.  Given such an injective map $\iota$, we say that $\iota$ is an encoding \emph{for the QECC $\iota(\Cvect{K})$}.

In order to come up with an efficient way to specify an encoding, let us now introduce the notion of a $K$-frame.

\begin{Definition}
Assume $K \leq N$.
A \emph{$K$-frame} is an ordered set of $K$ linearly independent vectors in the vector space $\Cvect{N}$.  An \emph{orthonormal $K$-frame} is a $K$-frame where all vectors in the set have norm $1$ and are pairwise orthogonal.
\end{Definition}

For us, $K \le N$ in general.  If $K = N$, a $K$-frame is precisely an ordered basis for $\Cvect{N}$.  In the general case of $K \le N$, a $K$-frame $\beta$ would be an ordered basis for the subspace $W = \mathrm{span}(\beta)$ of $\Cvect{N}$.

Now, let us describe encodings in terms of orthonormal $K$-frames.  Let $\alpha = \{e_1, e_2, ..., e_K\}$ be the standard (orthonormal, ordered) basis for $\Cvect{K}$.  To specify a QECC encoding $\iota: \Cvect{K} \hookrightarrow \Cvect{N}$, it is necessary and sufficient to specify $\iota(\alpha)$,\footnote{In other words, we should specify the codewords representing logical 0 through logical $K-1$.} which is an orthonormal $K$-frame in $\Cvect{N}$.  On the other hand, if we have an orthonormal $K$-frame $\beta$ in $\Cvect{N}$, we can obtain a partial isometry $\iota$ by letting $\iota(\alpha) = \beta$.  Hence, the set of possible encodings $\iota$ is in bijection with the set of possible orthonormal $K$-frames $\beta$, where their relation is given by $\beta = \iota(\alpha)$.  To summarize, we have
\begin{eqnarray}
\mathcal{S} &:=& \{\text{All encodings for $(K,N)$-dimensional codes}\} \\
	&=& \{\text{All orthonormal $K$-frames in $\Cvect{N}$}\}.
\end{eqnarray}

It turns out that the set of all orthonormal $K$-frames in $\Cvect{N}$ is a well-known mathematical object called the \emph{Stiefel manifold}, denoted as $\Stiefel$.  Therefore, we see that the set of all $(K,N)$-dimensional encodings is exactly the Stiefel manifold, namely
\begin{equation}
	\mathcal{S} = \Stiefel.
\end{equation}
%

Next, we give $\Stiefel$ a principal $\U{K}$-bundle structure, then in Theorem \ref{thm:PSt} below we identify it with $\bigPbund$ as principal bundles, and thereby letting $\Stiefel$ inherit the smooth manifold structure of $\bigPbund$ in Corollary \ref{cor:manifoldSt}.

Let $G(W)$ be the set of all orthonormal $K$-frames $\beta$ in $\Cvect{N}$ such that $\mathrm{span}(\beta) = W$.  Then
\begin{equation}
\Stiefel = \bigsqcup_{W \in \Grass} G(W).
\label{eqn:decomposeStiefel}
\end{equation}
$G(W)$ will be the fibre over $W$ in the bundle $\Stiefel$.

\begin{Theorem}
Let us fix a $K$-plane $V$ in $\Cvect{N}$ and an encoding $\iota: \Cvect{K} \hookrightarrow \Cvect{N}$ for $V$, namely $\mathrm{span}(\iota(\alpha))=V$ where $\alpha$ is the standard ordered orthonormal basis for $\Cvect{K}$.  Then there is a natural bijection of sets
\begin{equation}
	\bigPbund \simeq \Stiefel.
\label{eqn:bundleP}
\end{equation}
Furthermore, the set bijection preserves the fibres
\begin{equation}
	F(W) \simeq G(W)
\label{eqn:fibreP}
\end{equation}
for each $W \in \Grass$, where $F(W)$ is defined by Equation (\ref{eqn:decomposeP}) and $G(W)$ by the discussion preceding Equation (\ref{eqn:decomposeStiefel}).  That is, the above isomorphism is a bundle isomorphism.
\label{thm:PSt}
\end{Theorem}

\begin{proof}
Let us describe the bijections.
Recall that $\bigPbund$ is the coset space $\flatNmodVperp$.  Having fixed an encoding $\iota$ for $V$, define $\phi: \bigPbund \rightarrow \Stiefel$ by $\phi (\cosetVperp{A}) = A(\iota(\alpha))$.  Recall that $A$ is unitary, so $A(\iota(\alpha))$ is another orthonormal $K$-frame, this time spanning $A(V)$. Also define $\psi: \Stiefel \rightarrow \bigPbund$ by letting $\psi(\beta)$, for each orthonormal $K$-frame $\beta \in \Stiefel$, be the unique $\Vperp$-coset of unitary transformations which takes the $K$-frame $\iota(\alpha) \subset V$ to $\beta \subset \mathrm{span}(\beta)$.  The maps are well-defined, and it is easy to see that $\phi$ and $\psi$ are inverses of each other.

Now, we will show that $\phi$ maps the set $F(W)$ to $G(W)$ for every $W \in \Grass$.  Recall that $F(W)$ consists of all the left cosets $\cosetVperp{A}$ where $A(V) = W$ and $A$ unitary.  Hence, for $\cosetVperp{A}$ in $F(W)$, we have $\phi(\cosetVperp{A}) = A(\iota(\alpha))$, which is an orthonormal $K$-frame spanning $W$ because $\iota(\alpha)$ spans $V$ and $A(V) = W$.  Therefore, $\phi(\cosetVperp{A})$ is in $G(W)$.
\end{proof}


Because of Theorem \ref{thm:PSt}, we could view the fibre bundle structure on $\Stiefel$ as being inherited from $\bigPbund$ rather than being independently defined.  Thus we formulate the following corollary.

\begin{Corollary}
\label{cor:manifoldSt}
The Stiefel manifold has the structures of a smooth manifold and of a principal $\U{K}$-bundle, both inherited from $\bigPbund$, as described above.
\end{Corollary}

Since the bundle $\Stiefel$ has been identified with the bundle $\bigPbund$, we shall view the description of $\Stiefel$ as an alternative description of $\bigPbund$, and often refer to either description as $\bigPbund$ in the remainder of the paper.  We use the notation $\Stiefel$ when we would like to emphasize the $K$-frame interpretation of the bundle.



\subsubsection{Unitary evolution of encodings as paths in $\Stiefel$}
\label{sec:Upath.P}


In this section, we are concerned with the action of the full $\U{N}$ on $\Stiefel$.  The $\U{N}$ action, unlike the $\U{K}$ action of Equation (\ref{eqn:actionP}), does not preserve the fibres of $\Stiefel$.  If $U$ is a unitary operator on $\Cvect{N}$, it naturally induces a smooth map $h_U: \Stiefel \rightarrow \Stiefel$, namely 
\begin{equation}
h_U (\beta) = U(\beta).  
\label{eq:h_U}
\end{equation}
In particular, $h_U$ maps the fibre $G(W)$ to the fibre $G(U(W))$.

We now consider the action of a unitary evolution, namely, a piecewise smooth one-parameter family of unitary operators $U(t)$ on $\Cvect{N}$ such that $U(0) = \idmatrix$.  Suppose $U(t)$ is a unitary evolution and $\beta = \iota (\alpha)$ an orthonormal $K$-frame, or equivalently an encoding, for the reference code $C$.  We define 
\begin{equation}
\lambda^{\{U\}}(t) := h_{U(t)}(\beta) = U(t)(\beta), 
\label{eq:principalbundlelift}
\end{equation}
a one-parameter family of orthonormal $K$-frames, or equivalently, encodings.  Then at time $t$,
\begin{equation}
\pi(\lambda^{\{U\}}(t)) = \pi(U(t)(\beta)) = \mathrm{span}(U(t)(\beta)) = U(t)\, (\mathrm{span} (\beta)) = U(t)(C).
\end{equation}
Therefore, the encoding path $\lambda^{\{U\}}(t)$ projects to the subspace path $\gamma(t) = U(t)(C)$ in the base Grassmannian manifold.  In other words, the path $\lambda^{\{U\}}(t)$ is a lift of the path $\gamma(t)$ from the base manifold $\Grass$ to the bundle $\bigPbund$.  If $U(t)$ is piecewise smooth, so are $\lambda^{\{U\}}(t)$ and $\gamma(t)$.

\subsection{Candidates of natural parallel transport structures on $\bigvecbund$ and on $\bigPbund$}
\label{sec:proposal}

\begin{Remark}
Unitary evolutions corresponding to fault-tolerant operations, or any unitary evolutions that induce a closed loop in the base Grassmannian manifold, when applied to an initial codeword $(C,w)$, do not necessarily yield closed loops in $\bigvecbund$.  In general, $f_{U(1)}(C,w) = (C,w')$ for some element $w'$ in $\pi^{-1}(C)$.  These non-closed loops in the bundle $\bigvecbund$ are precisely the interesting points of this paper and the source of computational power of fault-tolerant protocols.
\label{remark:nonclosed}
\end{Remark}

\begin{Remark}
Similar to Remark \ref{remark:nonclosed}, a lift of a closed loop in $\Grass$ may not result in a closed loop in $\bigPbund$.  If $F(t)$ is a unitary fault-tolerant loop and if $h_{F(1)}(\beta) \ne \xi \beta$ for any $\xi \in \Cstar$ (or $\xi \in \U{1}$), this $F(t)$ shall correspond to a fault-tolerant gate implementation which yields a nontrivial logical operation.  
\end{Remark}

Here, we give primitive proposals for a connection on $\bigvecbund$ and on $\bigPbund$.  We then show that our proposals will not in general yield a flat projective connection on either of the two bundles over the entire Grassmannian manifold $\Grass$ using all unitary evolutions.  However, our ultimate goal is to show that they do yield well-defined flat projective connections when restricted to suitable $\M$ (i.e., when considered in $\Mvecbund$ and $\MPbund$) and sensible $\F$-evolutions.

The main idea is that a unitary path $U(t)$ induces not only a path $\gamma(t)$ in the Grassmannian, but also paths $\lgamma(t)$ in the fibre bundles $\bigvecbund$ and $\bigPbund$.  $\lgamma(t)$ is supposed to represent the lift of $\gamma(t)$ to the fibre bundle provided by parallel transport, and with a certain initial value $u_0$ in $\pi^{-1}(C_1)$ where $C_1$ is the starting $K$-dimensional subspace.  This is depicted in Figure~\ref{fig:unitarypath}.

We will use the term \emph{pre-connection} or \emph{pre-parallel transport} to indicate a connection-like object that gives a possibly non-unique notion of parallel transport.  In other words, a pre-connection is a set of possible parallel transport structures.  Two examples are given by equations (\ref{eqn:paratranspV}) or (\ref{eqn:paratranspP}).

\subsubsection{A proposal for a pre-connection on $\bigvecbund$}
\label{sec:connectionV}

Given a piecewise smooth path $\gamma(t)$ in $\Grass$ with $\gamma(0) = C_1$, we can choose a piecewise smooth unitary evolution $U(t)$ that induces $\gamma(t)$ when acting on $C_1$.  Then there is a path in $\bigvecbund$ defined by equation (\ref{eq:vectorbundlelift}),
\begin{equation}
\eta^{\{U\}}(t) := f_{U(t)} (C_1, w_1) = (U(t)(C_1), U(t)(w_1)),
\label{eqn:paratranspV}
\end{equation}
which is a lift of $\gamma(t)$, and the map $f_{U(t)}$ gives us a way to relate the fibre over $C_1$ to the fibre over $U(t)(C_1) = \gamma(t)$ for any $t \in [0,1]$.  In other words, $f_{U(t)}$ resembles the notion of a parallel transport, which is key to the concept of a connection.  

However, this primitive proposal of a parallel transport, $f_{U(t)}$, on $\bigvecbund$ depends crucially on the unitary evolution $\{U\} = U(t)$ and might not yield a consistent parallel transport in $\bigvecbund$ for different choices of $\{U\}$ from an allowed set $\Fpath$.  Depending on the unitary evolution $\{U\}$ we choose, there are various possibilities for a parallel transport between fibres over any two fixed points, $C_1$ and $C_2$, in the Grassmannian.  The proposal will not lead to a well-defined connection unless $f_{U(t)}$ is independent of the choice of unitary evolutions $\{U\}$ consistent with the same path $\gamma$ in the Grassmannian.

\subsubsection{A proposal for a pre-connection on $\bigPbund$}
\label{sec:connectionP}

Similarly, we can define a pre-connection on $\bigPbund$: Given any curve $\gamma(t)$ in the base space $\Grass$, choose a unitary evolution $\{U\} = U(t)$, which when acting on $\Grass$ results in the path $\gamma(t)$.  A natural $\{U\}$-dependent pre-parallel transport of a $(K,N)$-dimensional encoding $\beta = \iota(\alpha)$, where $\alpha$ is the standard basis for $\Cvect{K}$, over a curve $\gamma(t)$ in $\Grass$ is given by equation (\ref{eq:principalbundlelift}),
\begin{equation}
\lambda^{\{U\}}(t) := h_{U(t)}(\beta) = U(t)(\beta).
\label{eqn:paratranspP}
\end{equation}
Again, equation (\ref{eqn:paratranspP}) only defines a true parallel transport on $\bigPbund$ if along any path of subspace unitary evolution $\gamma(t)$, the proposed parallel transport $\lambda^{\{U\}}(t)$ is independent of the choice of unitary evolution $\{U\}$.

\subsubsection{Failure to have a connection for $\Fpath = \{\text{All unitary paths}\}$ or $\M = \Mgraph = \Grass$}
\label{sec:connection_failure}

If we try to apply these definitions to the full bundles $\bigvecbund$ or $\bigPbund$ while allowing all possible unitary paths, we find that we do not obtain connections.  For example, consider $\subGrass{2}{4}$, the manifold of $2$-dimensional subspaces of a $2$-qubit Hilbert space, and let $C$ be the subspace $\{ \ket{0} \otimes \ket{\psi} \, | \, \ket{\psi} \in \Cvect{2}\}$.  Let $\gamma(t) = C$ be the trivial path that stays at $C$ for all times.  $\gamma(t)$ can be derived from many different unitary paths, for instance a path of the form $\idmatrix \otimes U(t)$ for any $U(t)$.  However, any two such paths with different values of $U(1)$ will result in a different transformation of the fibre at $C$, giving many different pre-parallel transports over the trivial path at $C$ in $\Grass$.  So the pre-parallel transport we get this way is inconsistent, and fails to give us a well-defined connection.


\section{A richer geometric picture for \emph{fault-tolerant} unitary evolutions}
\label{sec:geoFT}

Given a unitary fault-tolerant protocol, our goal is to construct a connection out of the $\U{N}$ action on $\bigvecbund$ and $\bigPbund$.  We do not have a connection over the full bundles $\bigvecbund$ and $\bigPbund$.  However, when considering fault-tolerant operations, we do not consider arbitrary unitary paths but only certain unitary paths in a subset of the Grassmannian.  This leads us to a bundle with a smaller base space, on which we can hope to have a well-defined flat projective connection.

In this section, we outline a program to demonstrate the \emph{topological} nature of fault tolerance.  We define the main mathematical objects that we expect to play a role in the construction and state Conjecture~\ref{conj:flatconnection}, which codifies the main goal of the program.  We also prove Criterion~\ref{prop:flat_condition}, which will be useful in Sections \ref{sec:transversal} and \ref{sec:topologicalFT}, where we work out in detail our picture for two examples of fault-tolerant protocols.  In the next few sections, we analyze in detail only the \emph{unitary} fault-tolerant gates; we return to non-unitary gates briefly in Section~\ref{sec:fullFT}.

Ultimately, we would like to have a way to automatically construct the smaller bundles given a fault-tolerant protocol.  However, at this stage in the development of our program, the topological picture requires us to make some additional choices which might not be uniquely specified by the fault-tolerant protocol we begin with.  These choices need to be made carefully in order to both be consistent with the structure of the original fault-tolerant protocol and have the right properties to give us a flat projective connection.  Currently we must rely on intuition about the nature of fault tolerance in the model being studied in order to determine the correct choices.  In some cases, the choice is clear, while in others there may be a degree of arbitrariness in the choice.

In Section \ref{sec:specsFT}, we bridge the gap between Section \ref{sec:introFT} and the input required for our geometric framework, and we specify how (the continuous-time implementations of) the fault-tolerant operations should be described.  
In Section~\ref{sec:frameworksummary}, we introduce the main mathematical objects for the framework which we want to create, including the restriction bundles $\Mvecbund$ and $\MPbund$.  In Section \ref{sec:connectionM}, we give the main conjecture: on these restriction bundles, the general proposal for parallel transport in Section \ref{sec:proposal} is a flat projective connection.  The implication on the fault-tolerant logical gates is highlighted in Section \ref{sec:FTmonodromy}.

\subsection{Specifications of a fault-tolerant protocol for our framework}
\label{sec:specsFT}

In the existing literature, a rigorous definition of quantum fault tolerance in the general setting (for general error models and protocols)\footnote{In some restrictive settings, however, one can find more precise definitions for quantum fault tolerance.  For example, see~\cite{threshold} for a more precise definition of fault tolerance in one important context.} is still lacking.  Naturally, this poses certain challenges for us as we seek to give a geometric formulation for unitary quantum fault tolerance.  We do not attempt to give a rigorous definition for fault tolerance here; we merely try to organize it enough for our use.  

\subsubsection{Bridging the gap between Section \ref{sec:introFT} and our framework}

In Section~\ref{sec:introFT}, we described the three main ingredients of a fault-tolerant protocol.  There we illustrated how the error model, the QECC, and the FT operations must all be chosen harmoniously, a concept illustrated pictorially in Figure \ref{fig:FTaspects}.

For Sections \ref{sec:geoFT}, \ref{sec:transversal} and \ref{sec:topologicalFT}, we assume that we are given a fault-tolerant protocol, upon which we carry out the construction of a mathematical framework as sketched out in the present section.  We will not need all the information in a specification of an FT protocol for building our current framework, instead we extract the part of the specification we need.

For one thing, we do not explicitly incorporate error models into our framework at this stage, although we hope to address and incorporate errors into our followup work.  However, we do assume that an appropriate error model has been used, and compatibility considered, in coming up with the set of fault-tolerant operations, the physical gates and their continuous-time implementations, all of which play key roles in our geometric construction.

For our framework, we do however need a description of the reference code and the FT operations, along with some idea of how the FT operations are implemented.  Now, let us take a closer look at how these ingredients of a fault-tolerant protocol are described.  One simple way to partially prescribe a fault-tolerant protocol is to give a reference code $C$ to which the encoded Hilbert space always returns in between fault-tolerant gates and a set of fault-tolerant gates acting on $C$, along with a continuous-time implementation for each fault-tolerant gate.%
\footnote{Sometimes the continuous-time implementation is implicitly given or not completely specified, leaving some ambiguity in the detailed implementation.}
This description is somewhat minimal and often used when \emph{specifying} an FT protocol.  However, when \emph{designing} FT protocols, that is when we are mainly concerned about which physical gates or implementations to allow in order to keep errors under control, another approach to describing the FT operations is often used: as in the examples in Sections \ref{sec:squdit_errors} and \ref{sec:local_errors}, FT protocols are often designed by first figuring out and specifying the allowed physical gates (the ordering could matter), and then defining the FT operations to be a sequence of physical gates which takes code $C$ back to itself.

Even if $F$ is a fault-tolerant gate, not all unitary paths that start at $\idmatrix$ and end at $F$ are valid fault-tolerant implementations of that gate.  Some such paths will pass through codes that are not able to correct errors well, and others will cause errors to propagate in an uncontrolled way.  A fault-tolerant implementation must avoid these dangers so that errors can be corrected no matter when they occur.  It may be that in order for the implementation to be fault tolerant, error correction must be performed multiple times in the course of the procedure.  Since we are only considering paths without errors and since error correction does not change a state with no errors in it, any required error correction steps do not show up in the formulation we use.  In the remainder of Section~\ref{sec:geoFT}, we will assume the fault-tolerant protocol is given to us and will not worry about how it is designed.


\subsubsection{Fault-tolerant unitary operations $\FsmallL$ and loops $\Floop$}
\label{sec:FTlogicals}

As discussed above, a fault-tolerant protocol consists of a group $\FsmallL$ of fault-tolerant unitary operations and a set $\Floop$ of fault-tolerant unitary loops, the specified fault-tolerant implementations of the gates in $\FsmallL$.

Recall that one of the properties of an FT operation was captured by Equation (\ref{eqn:code-gate-relation}) of Section \ref{sec:FToperations}:
\begin{equation}
F(C) = C
\end{equation}
where $C$ is the reference code for the FT protocol (and will always mean the reference code by default).  Note that not all unitaries satisfying this equation are fault tolerant.  

The fault-tolerant implementations of the elements of $\FsmallL$ are given by paths in the set $\Floop$.  The paths in $\Floop$ are maps from $[0,1]$ to $\U{N}$ which start at the identity and end at elements of $\FsmallL$:
\begin{itemize}
\item If $F(t) \in \Floop$, then $F(0) = \idmatrix$ and $F(1) \in \FsmallL$.
\end{itemize}

\begin{Definition}
When working within the context of an FT protocol, we use the term \emph{FT unitary operations} to mean only those operators in $\FsmallL$.  An \emph{FT unitary loop} is a unitary path in $\Floop$.
\end{Definition}

It is easy to see that we may assume the identity operator is in $\FsmallL$ because we can take its unitary implementation to be $F(t) = \idmatrix$, which is clearly fault tolerant for sensible pairs of code and error model.  It is also clear that the composition of the two elements of $\FsmallL$ is again in $\FsmallL$ by concatenating the two corresponding unitary fault-tolerant implementations.  However, to fully establish the group structure on $\FsmallL$ (which ensures that $\Floop$ has some convenient properties for our construction), we shall consider only fault-tolerant protocols where the following holds:
\begin{itemize}
\item If $F$ is in $\FsmallL$, then $F^{-1} \in \FsmallL$.
\end{itemize}
In standard examples this property holds, but it is not clear it is true for all fault-tolerant protocols.

Fault-tolerant unitary loops are basically those unitary evolutions which fault-tolerantly and unitarily implement a fault-tolerant unitary operation.  
Conversely, a unitary is only considered a fault-tolerant unitary operation if there exists a fault-tolerant unitary loop implementing it.  Putting the two together, we have the following relation between the sets $\FsmallL$ and $\Floop$:
\begin{equation}
\FsmallL = \{ F(1) \;|\; F(t) \in \Floop \}.
\end{equation}

Furthermore, if $F(t)$ with $t \in [0,1]$ is a unitary fault-tolerant implementation, so is $F(t)$ with $t \in [0, t_0]$, where $0 \le t_0 \le 1$, as long as $F(t_0)(C) = C$.  Therefore, we also have
\begin{equation}
\FsmallL = \{ F(t_0) \;|\; 0 \le t_0 \le 1,\, F(t_0)(C) = C,\, F(t) \in \Floop \},
\label{eqn:FsmallL_max_Floop}
\end{equation}
which expresses a kind of maximality of the group $\FsmallL$ with respect to the set $\Floop$.  

Since $\FsmallL$ is closed under inverses, we may also assume $\Floop$ is closed under reversal of paths:
\begin{itemize}
\item If $F(t) \in \Floop$, then $F(1-t) F(1)^{-1} \in \Floop$.
\end{itemize}

In terms of the geometric picture we have been building since Section \ref{sec:geoUnitary}, the trajectory $F(t)(C)$ of codes, where $F(t)$ is a fault-tolerant unitary loop, actually forms a closed loop (based at $C$) in the Grassmannian, hence the term ``loop'' even though in general $F(0) \ne F(1)$ and $F(0)|_{C} \ne F(1)|_{C}$.  This means that $F(t)$ only induces a closed loop in the base space when acting on the subspace $C$, but $F(t)$ does not form a closed loop in $\U{N}$, nor does it yield closed loops in general in the two fibre bundles when acting on an element $u$ in the fibre $\pi^{-1}(C)$.  To gain some intuition, the reader may consult Figure~\ref{fig:unitarypath}, especially focusing on the segment of $U(t)$ that gives rise to the little loop (to the right of the intersection point) in the base of the bundle (depicted in the right hand side of the figure) and its lift.

Note that the set of loops we get by applying fault-tolerant unitary loops to the reference code $C$ will in general not be the set of all loops in $\Grass$ based at $C$, but instead just those loops passing through sufficiently robust intermediate codes at all times.  

\subsection{The main elements of the framework}
\label{sec:frameworksummary}

\subsubsection{The sets $\Fsmall$, $\Fpath$, and $\Mgraph$}

We have been given the following two sets specifying a fault-tolerant protocol:

\begin{itemize}
\item $\FsmallL$, the group of gates in the protocol implemented via fault-tolerant unitary operations
\item $\Floop$, the set of implementations (i.e., FT unitary loops) of the gates in $\FsmallL$
\end{itemize}

Based on these, we introduce the following objects:
\begin{itemize}
\item $\Fsmall$, the set of unitary operations that can appear in $\Floop$,
\begin{equation}
\Fsmall := \{ F(t_0) \;|\; 0 \le t_0 \le 1,\, F(t) \in \Floop \} \subset \U{N}.
\label{eqn:Fsmall_def}
\end{equation}

\item $\Fpath$, the set of paths in $\Fsmall$:
\begin{equation}
\Fpath := \{ F(t): [0,1] \rightarrow \Fsmall, F(t) \text{ piecewise smooth} \}
\end{equation}

\item $\Mgraph$, the subset of $\Grass$ indicating those subspaces which appear in the course of a fault-tolerant unitary gate,
\begin{equation}
\Mgraph := \Fsmall(C)
\end{equation}
\end{itemize}

We don't restrict $\Fsmall$ to contain only unitaries with $U(C) = C$ for the reference code $C$ because we want $\Fsmall$ to include both fault-tolerant gates and intermediate unitaries that we pass through while implementing fault-tolerant gates.
Since $F(0) = \idmatrix$ for all $F(t) \in \Floop$, it's not hard to see that in general $\Fsmall$ will be a subset of $\U{N}$ that is path-connected to the identity.  However, unlike $\FsmallL$, $\Fsmall$ does not always have a group structure.

\subsubsection{The sets $\Fbig$, $\Fbigpaths$, and $\M$}

The set $\Mgraph$ gives a first approximation of the base space of the bundles with projectively flat connections.  In the case of transversal gates (Section~\ref{sec:transversal}), it actually is the base space.  However, in other cases, such as for the toric code (Section~\ref{sec:topologicalFT}), $\Mgraph$ is too sparse to meaningfully use as a base space, so we wish to extend it.  This leads to some larger sets:
\begin{itemize}
\item $\M$, a subset of $\Grass$ with a projectively flat connection
\item $\Fbig$, a set of unitary operations including $\Fsmall$ but acting on $\M$.  $\Fbig$ must have the property that
\begin{equation}
\M = \Fbig(C)
\end{equation}

\item $\Fbigpaths$, the set of paths in $\Fbig$:
\begin{equation}
\Fbigpaths := \{ F(t): [0,1] \rightarrow \Fbig, F(t) \text{ piecewise smooth} \}
\end{equation}

\end{itemize}

We currently have no precise formula for choosing $\M$ and $\Fbig$ given the fault-tolerant protocol.  This is the biggest missing piece of our program.

In general, we wish to choose $\M$ to be a topological space which is ``not too much larger'' than $\Mgraph$.  In particular, we want $\M$ to have the property that ``small'' loops in $\Mgraph$ should be homotopically trivial in $\M$ even though they might not be in $\Mgraph$.  ``Large'' loops may remain homotopically non-trivial. The precise meaning of the terms ``small loops'' and ``large loops'' has to be determined on a case-by-case basis.


The choice of $\Fbig$ is closely tied to the choice of $\M$.  Unitaries in $\Fbig$ should be implementable by paths that have similar but perhaps somewhat weaker fault-tolerance properties to those in $\Floop$.  One possible route for the program is to consider paths in $\Floop$ but with errors added, and then let $\Fbig$ be the set of unitaries that appear on the noisy paths.

Both $\M$ and $\Fbig$ play critical roles in our program.  We use $\M$ as the base space for $\Mvecbund$ and $\MPbund$, which are supposed to have projectively flat connections.  However, in order to move from a pre-connection to a well-defined projective connection, we also need to restrict the unitaries available to the set $\Fbig$.  Thus, the claim is that when $\M$ and $\Fbig$ are chosen properly, parallel transport is well-defined and flat on $\M$ for unitary paths involving operators from $\Fbig$.

\subsubsection{The restriction bundles $\Mvecbund$ and $\MPbund$}
\label{sec:restriction_bundle}

Consider the (smooth) inclusion map
\begin{equation}
\iota_{\M}: \M \hookrightarrow \Grass,
\end{equation}
where the map $\iota_{\M}$ is the standard inclusion map obtained by viewing $\M$ as a subset of $\Grass$.  As described in Section \ref{sec:pullback}, we can pull back the ``big bundles'' over $\Grass$ to bundles over $\M$ along the base map $\iota_{\M}$.  We call the pullback of $\bigvecbund$ along $\iotaM$ the \emph{restriction vector bundle} over $\M$, and denote it by $\Mvecbund$.  Similarly, the pullback of $\bigPbund$ is called the \emph{restriction principal bundle} over $\M$ and is denoted $\MPbund$.

The pullback is generally defined for any continuous (or smooth, depending on what kind of bundles we are dealing with) map, and not just for injective ones.  But since $\iota_{\M}$ is injective, the bundles resulting from the pullback along $\iota_{\M}$ will be embeddable in the respective ``big bundles'', fibre by fibre.  In fact, the restriction bundles over $\M$ obtained this way, as a topological space, consist exactly of those fibres over the points in the subspace $\M$.

\begin{figure}
\begin{tikzpicture}

\path (-1,0) coordinate (Ucenter);
\path (Ucenter) ++(0,1) coordinate (Uabove);
\draw (Ucenter) ++(-1,0) arc (180:360:1cm and 0.5cm);
\draw (Ucenter) circle (1cm);
\node [anchor=south] at (Uabove) {$\U{N}$};

\path (2,-1) coordinate (Mleft);
\path (Mleft) ++(3,0) coordinate (Mright);
\path (Mleft) ++(0,2) coordinate (topleft);
\path (Mright) ++(0,2) coordinate (topright);
\path (topleft) ++(1.5,0.7) coordinate (topmid);

\draw (Mleft) .. controls +(0,-0.2) and +(-0.2,0) .. ++(0.5,-0.5) .. controls +(0.2,0) and +(-0.5,0) .. ++(1,0.3) .. controls +(0.2,0) and +(-0.5,0) .. ++(1,-0.4) .. controls +(0.5,0) and +(0,-0.3) .. (Mright) .. controls +(0,0.2) and +(330:0.3) .. ++(-1.3,0.5) .. controls +(150:0.3) and +(0.2,0) .. ++(-0.7,0.2) .. controls +(-0.2,0) and +(0,0.2) .. (Mleft);

\draw (topleft) .. controls +(0,-0.2) and +(-0.2,0) .. ++(0.5,-0.5) .. controls +(0.2,0) and +(-0.5,0) .. ++(1,0.3) .. controls +(0.2,0) and +(-0.5,0) .. ++(1,-0.4) .. controls +(0.5,0) and +(0,-0.3) .. (topright) .. controls +(0,0.2) and +(330:0.3) .. ++(-1.3,0.5) .. controls +(150:0.3) and +(0.2,0) .. ++(-0.7,0.2) .. controls +(-0.2,0) and +(0,0.2) .. (topleft);

\draw (Mleft) -- (topleft);
\draw (Mright) -- (topright);

\node [anchor=north west] at (Mright) {$\M$};
\node [anchor=south] at (topmid) {$\Mvecbund$ or $\MPbund$};

\path (Ucenter) ++(-0.5,0.1) coordinate (Ustart);
\path (Ustart) ++(1,0.2) coordinate (Uend);
\path (Ustart) ++(0.5,-0.1) coordinate (Ulabel);

\draw (Ustart) .. controls +(45:0.8) and +(160:0.3) .. (Uend);
\draw (Ustart) [fill=black] circle (0.05);
\draw (Ustart) ++(1,0.2) [fill=black] circle (0.05);
\node at (Ulabel) {$U(t)$};

\path (Mleft) ++(0.6,0) coordinate (Mstart);
\path (Mstart) ++(1.5,0) coordinate (Mmid);
\path (Mmid) ++(0,-0.1) coordinate (Mmiddown);
\path (Mstart) ++(0.5,0.4) coordinate (Mend);
\path (Mstart) ++(0,1) coordinate (liftstart);
\path (liftstart) ++(1.5,0) coordinate (liftmid);
\path (liftstart) ++(0.5,0.505) coordinate (liftend);

\draw (Mstart) .. controls +(30:0.5) and +(90:0.3) .. (Mmid) .. controls +(270:0.3) and +(330:0.6) .. (Mend);
\draw (Mstart) [fill=black] circle (0.05);
\draw (Mend) [fill=black] circle (0.05);
\node [anchor=west] at (Mmiddown) {$\gamma(t)$};

\draw (liftstart) .. controls +(30:0.5) and +(190:0.3) .. ++(0.85,0.2);
\draw (liftstart) ++(1.05,0.21) .. controls +(10:0.3) and +(90:0.3) .. (liftmid) .. controls +(270:0.2) and +(310:0.2) .. ++(-0.55,0.205) .. controls +(130:0.2) and +(330:0.3) .. (liftend);
\draw (liftstart) [fill=black] circle (0.05);
\draw (liftend) [fill=black] circle (0.05);
\node [anchor=west] at (liftmid) {$\lgamma(t)$};

\draw [dotted] (Mstart) -- (liftstart);
\draw [dotted] (Mend) -- (liftend);

\path (Uend) ++(0.2,-0.2) coordinate (Uendnw);
\path (Mstart) ++(-0.2,0.05) coordinate (Mstarte);
\draw [dashed, ->] (Uendnw) -- (Mstarte);

\path (Uend) ++(0.2,0) coordinate (Uendw);
\path (liftstart) ++(-0.2,0) coordinate (liftstarte);
\draw [dashed,->] (Uendw) -- (liftstarte);

\end{tikzpicture}
\caption{A unitary path $U(t)$ induces a path $\gamma(t)$ in the base space $\M$ and a path $\lgamma(t)$ in the fibre bundle $\Mvecbund$ or $\MPbund$.  By definition, $\lgamma(t)$ is a lift of $\gamma(t)$.  The induced path $\gamma(t)$ may be self-intersecting even when the original unitary path $U(t)$ is not.}
\label{fig:unitarypath}
\end{figure}
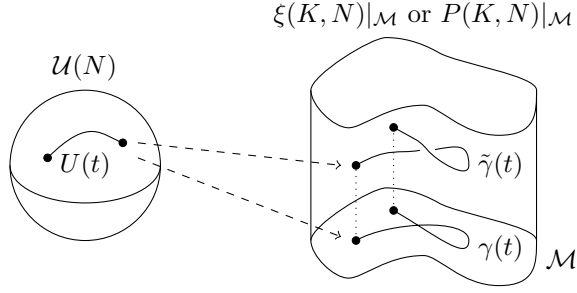

\subsection{Parallel transport structures and flat projective connections on $\Mvecbund$ and $\MPbund$}
\label{sec:connectionM}

\subsubsection{Partial connections and their classification}

The set $\Mgraph$ is the object given to us by the fault-tolerant protocol, but parallel transport on $\Mgraph$, even if consistent, may only give a \emph{partial connection}, telling us how to perform parallel transport along a limited set of paths in $\M$ in the case that $\M \neq \Mgraph$.  Parallel transport using this ``connection'' makes sense when we restrict attention only to $\Mgraph$, but this may not be a sufficiently well-behaved space, so to talk about (projective) flatness, we wish to think about all of $\M$ and use an extended version $\Fbigpaths$ of $\Fpath$.

We can distinguish between two senses in which a partial connection may be ``projectively flat'':
\begin{Definition}
\label{def:FPP}
A partial connection on $\Mgraph$ is a \emph{strong flat projective partial, or strong FPP, connection} over $\Mgraph \subseteq \M$ if it can be extended to a flat projective connection on $\Mvecbund$ and $\MPbund$.
If it can be extended to a flat connection (not just a projective connection), then it is a \emph{strong flat partial, or strong FP, connection}.  A partial connection on $\Mgraph$ is a \emph{weak FPP connection} if the modified restricted holonomy group $\Hol_{C}^{'0}(\HE{})$ (for $C \in \Mgraph$) defined below:
\begin{equation*}
\Hol_{C}^{'0}(\HE{}) := \{\Gamma(\gamma) \;|\; U(t) \in \Fpath,\ \gamma(t) = U(t)(C) \text{ contractible in } \M \}
\end{equation*}
is projectively trivial (i.e., is contained in $\Cstar$ or $\U{1}$).
If $\Hol_{C}^{'0}(\HE{})$ is trivial (not just projectively trivial), then it is a \emph{weak flat partial, or weak FP, connection}.
\end{Definition}

A weak FPP connection is only defined along paths induced by $\Fpath$ --- it may not be extendable to all of $\M$ --- and we only use $\M$ to determine if a loop is topologically trivial or not.  A strong FPP connection is originally defined only on paths $\gamma(t)$ lying in $\Mgraph$, but we can extend it to all of $\M$, using $\Fbigpaths$ instead of $\Fpath$.

To show that a partial connection is a strong FPP, we must define the larger set of allowed unitary paths $\Fbigpaths \supset \Fpath$ so that any path in $\M$ starting at $C$ can be derived from a unitary path in $\Fbigpaths$.  Then we use the natural pre-connection for $\Fbigpaths$ to define a pre-connection on $\Mvecbund$ and $\MPbund$.  If the pre-connection defines a flat projective connection on all of $\M$, then we had a strong FPP connection over $\Mgraph$.

For transversal gates (Section~\ref{sec:transversal}), the projective connection is naturally defined on a manifold.  It makes sense for this example to choose $\Mgraph = \M$, so we do not have to deal with partial connections.  For toric codes (Section~\ref{sec:topologicalFT}), $\Mgraph$ is not a manifold and we naturally define only a partial connection.  However, we will show that it is actually a strong FP connection, namely it can be extended to a flat connection on some suitable $\M$.  For more general cases, it might be necessary to resort to weak FPP connections.

\subsubsection{Main conjecture}
\label{sec:flatconnection_conj}

We now state the main conjecture of the paper:

\begin{Conjecture}
For any set $\Fpath$ of fault-tolerant unitary paths that comes from a fault-tolerant protocol, involving a reference code $C$ and a set $\Floop$ of fault-tolerant unitary loops, the natural proposals from Section~\ref{sec:proposal} for pre-connections on $\Mgraph$ may be extended to a well-defined projective connections in $\Mvecbund$ and $\MPbund$ when
restricted to an appropriately chosen subset $\M \supset \Mgraph$ of $\Grass$ and an appropriate subset of unitary paths $\Fbigpaths \supset \Fpath$; furthermore, these extended connections are projectively flat.
\label{conj:flatconnection}
\end{Conjecture}

Our examples both produce strong FPPs, so we have phrased Conjecture \ref{conj:flatconnection} to say that we will always get strong FPP connections, but it is possible that in more general cases, it has to be weakened a bit to allow weak FPPs.

By restricting ``parallel transport'' to be defined via a properly chosen set $\Fpath$ of fault-tolerant unitary paths in $\U{N}$, we believe that the problems noted in Section~\ref{sec:connection_failure} vanish, and the pre-connection becomes an actual (projective) connection.  We will show that the conjecture is true for the examples given in Sections \ref{sec:transversal} and \ref{sec:topologicalFT}.  It is quite natural for a fault-tolerant protocol to yield a flat projective connection: A fault-tolerant protocol should be robust against small deviations from the ideal implementation of the gates, which are created via evolutions from $\Floop$, and a flat projective connection indicates that small distortions of the path of unitaries $F(t)$ should not change the logical gate.  This physical intuition leads us as well to the converse conjecture, Conjecture \ref{conj:FTprotocol}, which we state in the conclusions section.

\subsubsection{Criteria for projective flatness of the proposed (extended) connection on $\M$}

The following provides a useful criterion to determine if the pre-connection defined in Section~\ref{sec:proposal} is a flat projective connection.

\begin{Criterion}
\label{prop:flat_condition}
Let $W_1, W_2 \in \M$ be contained in a simply-connected subset $\N \subseteq \M$.  Suppose that for any such $W_1$, $W_2$, and $\N$, and any pair of unitary paths $U_1(t), U_2(t) \in \Fbigpaths$ such that $U_1(0) = U_2 (0) = \idmatrix$, $U_1(1)(W_1) = W_2 = U_2(1)(W_1)$, and $U_1(t)(C), U_2(t)(C) \in \N$ for all $t \in [0,1]$, and for any initial condition $(C,w)$ or $\beta$ in the fibre at $W_1$, we have that
\begin{align}
\eta^{\{U_1\}} (1) &= \xi \cdot \eta^{\{U_2\}} (1), \label{eqn:flat_condition_vect} \\
\lambda^{\{U_1\}} (1) &= \xi \cdot \lambda^{\{U_2\}} (1). \label{eqn:flat_condition_prin}
\end{align}
where $\xi$ is a nonzero complex number. When equation~(\ref{eqn:flat_condition_vect}) is true, we work with $\Mvecbund$, and when equation~(\ref{eqn:flat_condition_prin}) is true, we work with $\MPbund$.  When both are true, we can work with either bundle.  Then the natural pre-connection on the bundle defines a projective connection on the bundle over $\M$ and the projective connection is projectively flat.  If $\xi = 1$ always, then we have a flat connection rather than a flat projective connection.  When $\Mgraph \neq \M$, we also have a strong FPP (or FP connection when $\xi = 1$) on $\Mgraph$.
\end{Criterion}

\begin{proof}
In general, there are many unitaries $U$ such that $U(W_1) = W_2$.  Two such unitaries $U$ and $U'$ differ by their action on the subspace, $U = U' V$, with $V \in \Ub{W_1}$. However, when the hypothesis is true, we have a way to uniquely choose such a unitary $U$ (up to phase $\xi$) by focusing on subspace paths within $\N$: Suppose we have two paths $\gamma_1$ and $\gamma_2$ within $\N$ with the same endpoints, and let $U_1(t)$ and $U_2(t)$ be lifts of $\gamma_1$ and $\gamma_2$ in $\Fbigpaths$ respectively, satisfying $U_1(0) = \idmatrix = U_2(0)$.  If $U$ and $U'$ are the endpoints of the unitary paths, $U = U_1(1)$ and $U'=U_2(1)$, then by the hypothesis of the proposition, they must transform the fibre in the same way (up to a phase), and therefore $V = \xi \idmatrix$.  If we fix $W_1$ and let $W_2$ vary over $\N$, we get a canonical choice of $U$ (up to a phase) for all points in $\N$.  This canonical choice of $U$ then determines a canonical local trivialization of the fibre bundle over $\N$.  We can then use this local trivialization to define the projective connection in the obvious way within $\N$, in the sense of horizontal subspaces for the Ehresmann connection.  When $\xi = 1$, we get a regular (non-projective) connection.

This conclusion holds for any simply-connected $\N$.  If $\N_1$ and $\N_2$ are two simply-connected regions and $\N_1 \cap \N_2$ is connected, then $\N_1 \cup \N_2$ is simply connected as well, and the projective connection is well-defined on the union.\footnote{For intuition, it suffices to imagine $\N_i$ as disks.}  In particular, this means the canonical choice of projective connection is consistent between $\N_1$ and $\N_2$.  By breaking the coordinate chart of $\M$ up into simply-connected sets whose pairwise intersections are connected, we see that the projective connection is well-defined globally. (If $\M$ is not connected, we can treat each connected component separately.)

Furthermore, with this condition, the projective connection is projectively flat: Consider a contractible loop $\gamma$ based at $W_1$, and let $\N$ consist of $\gamma$ union the interior of $\gamma$ (the points swept out when we contract $\gamma$ to nothing).  We can break up $\gamma = \gamma_2 * \gamma_1^{-1} $, where $\gamma_1$ and $\gamma_2$ are two paths from $W_1$ to some other $W_2 \in \N$.  Then there exist unitary paths $U_1 (t)$ and $U_2(t)$ realizing $\gamma_1$ and $\gamma_2$.  By the hypothesis, the paths $U_1 (t)$ and $U_2(t)$ perform the same parallel transport on any fibre element $\beta$ (for $\MPbund$) up to a phase.  The loop $\gamma$ thus turns $\beta$ into $\xi \beta$.  This is true for arbitrary $\gamma$, so the projective connection is projectively flat (or flat if $\xi = 1$).  The same argument works for $\Mvecbund$ as well.
\end{proof}

\subsection{Fault-tolerant logical gates and the monodromy groups}
\label{sec:FTmonodromy}

If we think of implementing a logical gate in a real system, we are likely to do so by turning on a Hamiltonian $H(t)$ for some period of time and then turning it off.  That means the gate is not implemented all at once; instead, we perform a path
\begin{equation}
F(t) = \timeorder e^{-i \int_0^t H(t') \, dt'}.
\end{equation}
In this case, $F(t)$ is a path in $\Fpath$, and since $F(0)( C) = F(1)( C) = C$ by virtue of being a logical gate, $F(t)$ is in fact a $\Fsmall$-loop based at $C$, namely $F(t) \in \Floop$.

In cases where Conjecture~\ref{conj:flatconnection} holds, it follows that if the loop $F(t)(C)$ is homotopically trivial in $\M$, then it performs a trivial (up to a phase) fibre automorphism on $\pi^{-1}(C)$.  In other words, it performs the identity logical gate.  The non-trivial logical transversal gates correspond to some of the homotopically non-trivial loops in $\M$, so we can with justice say that all fault-tolerant gates in the set $\Floop$ are in fact topological.  Since the connections in $\Mvecbund$ and $\MPbund$ are projectively flat, $\pi_1 (\M)$, the first fundamental group of $\M$, has a monodromy action on the fibres, given by
\begin{eqnarray}
\GV &=& \{\eta^{\{F\}}(1) \;|\; F(t) \in \Floop\} \\
\GP &=& \{\lambda^{\{F\}}(1) \;|\; F(t) \in \Floop\}
\end{eqnarray}
in the two bundles respectively.

As a consequence of Conjecture \ref{conj:flatconnection}, we obtain the following corollary relating geometry/topology and fault-tolerant logical gates.

\begin{Corollary}
\begin{equation}
\GV = \{\text{Logical gates implemented via paths in $\Fsmall$}\} = \GP
\end{equation}
\label{cor:FTlogicalgates}
\end{Corollary}

Of course, it is possible that there are topologically non-trivial loops in $\M$ that nonetheless give us trivial logical gates.  See Section~\ref{sec:fivequbit} and Figure~\ref{fig:fivequbit} for an example.

\section{Example I: Transversal gates}
\label{sec:transversal}


We have seen in Section \ref{sec:squdit_errors} that minimum distance codes (Definition \ref{def:distance}) and transversal gates (Definition \ref{def:transversal}) work together in a fault-tolerant protocol to protect against $s$-qudit error models.  In this section, we investigate how they fit into our geometric description of fault tolerance.

\subsection{Single-block and multiple-block transversal gates}

In the single code block case, we work with an $n$-qudit (physical) Hilbert space of the most general form, namely $\Hilbphy=\bigotimes_{j=1}^{n} \Cvect{d_j}$ where $d_j$ is the dimension of the $\ord{j}$ \emph{transversal component}.  Therefore, the total physical dimension is $N = \prod_{j=1}^{n} d_j$.  Usually, all $d_j$ are the same and are equal to $d$, but this is not needed for our results.

Often, instead of a single block, we work with $M$ blocks of identical QECCs, each of which is an $((n,K',\delta))$ qudit code $C_1$ with physical qudit dimensions $\vec{D} = (D_1, D_2, ..., D_n)$.  It follows that the total physical Hilbert space is
\begin{equation}
\Hilbphy = (\bigotimes_{j=1}^{n} \Cvect{D_j})^{\otimes M} \cong \bigotimes_{j=1}^{n} (\Cvect{D_j})^{\otimes M},
\end{equation}
and the joint code space is $ C = C_1\,^{\otimes M} \subset \Hilbphy$, where each copy of $C_1$ sits in a different tensor factor
in the first decomposition of $\Hilbphy$.  Transversal gates interact the $\ord{j}$ registers $\Cvect{D_j}$ of each of the $M$ code blocks.

It will prove convenient to view $M$ blocks of identical QECCs as a single QECC with larger qudits, as shown in Figure~\ref{fig:transversalmultiple}.  To do this, we use the second tensor product decomposition $\Hilbphy = \bigotimes_{j=1}^{n} (\Cvect{D_j})^{\otimes M}$ and view each transversal component $(\Cvect{D_j})^{\otimes M}$ as a new qudit of dimension $d_j = D_j^M$.  This way, we obtain a single $((n, K,\delta))$ qudit QECC $C$ with physical qudit dimensions $\vec{d} = (d_1, d_2, ..., d_n)$ and $K = K'^M$.  Since the physical Hilbert space dimension for each code block $C_1$ is $N' = \prod_{j=1}^{n} D_j$, the joint physical Hilbert space dimension of the $M$ blocks is $N = N'^M$.  Note that the distance $\delta$ and the number $n$ of qudits needed remain unchanged, although the size of the qudits has changed.

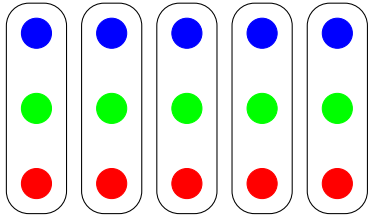
\begin{figure}
\begin{tikzpicture}

\foreach \x in {1,..., 5}
	\draw [fill, red] (\x,1) circle (2mm);
\foreach \x in {1,..., 5}
	\draw [fill, green] (\x,2) circle (2mm);
\foreach \x in {1,..., 5}
	\draw [fill, blue] (\x,3) circle (2mm);

\foreach \x in {1,...,5}
	\draw [rounded corners=0.3cm] (\x,0) ++(-0.4,0.6) rectangle ++(0.8,2.8);

\end{tikzpicture}
\caption{A transversal gate on multiple blocks of a QECC can be considered as a transversal gate on a single block of a QECC with larger physical qudits.  Each circle represents a qudit, and each row represents a block of the code.  We group together qudits in the same column to make the larger qudits.   The $\ord{j}$ qudits from the various code blocks are interacted together in a transversal gate, so a transversal gate can also be viewed as acting separately (as a tensor product) on the individual larger qudits.}
\label{fig:transversalmultiple}
\end{figure}

A transversal gate on $M$ blocks of the code $C_1$ therefore acts in the form $\bigotimes U_j$, where each tensor factor acts on the $d_j$-dimensional $\ord{j}$ qudit of the new code $V$.  Therefore, in the following sections, it is sufficient to confine our attention to single-block transversal gates.  Given instead multiple blocks of identical QECCs, we shall automatically define $C$, $\vec{d}$, $K$ and $N$ as above.

\subsection{$\F$ and $\M$ for transversal gates}
\label{sec:QECproperties.transversal}

\subsubsection{Definitions of $\F$ and $\M$}

When we consider fault-tolerant transversal gates, the natural set of allowed unitaries are simply the transversal gates:
\begin{equation}
\F = \left\{ \bigotimes_{j=1}^{n} U_j \;\bigg|\; U_j \in \Ub{\Cvect{d_j}} \right\}.
\label{eqn:transversal}
\end{equation}
The set $\F$ is a compact Lie subgroup of the full unitary group $\U{N}$.  We define the set $\Ftranspath$ of allowed/fault-tolerant unitary evolutions to be all piecewise smooth paths in $\F$.  It is easy to see that $\F$ is $\F$-path-connected to the identity $\idmatrix$.

In Conjecture \ref{conj:flatconnection}, we mention the subset $\M$ of the Grassmannian where we may hope to establish a flat projective connection for the restriction vector bundle $\Mvecbund$ and the restriction principle $\U{K}$-bundle $\MPbund$.  For transversal gates, when $C$ is the initial code space, we can define $\M$ as:
\begin{equation}
\M = \Mgraph := \F(C) = \{F(C) \;|\; F \in \F \} \subset \Grass.
\label{eqn:defM.transversal}
\end{equation}
Since $\Mgraph$ is already a nice and ``big enough'' manifold, we can let $\M = \Mgraph$ for transversal gates.

\begin{Proposition}
If $C, C' \in \M$, then the distance of $C'$ is the same as the distance of $C$.
\label{prop:singledistance}
\end{Proposition}

\begin{proof}
For a pair of codes $C, C' \in \M$, by the definition of $\M$, $\exists\, U \in \F$ such that $U(C) = C'$.
Let $\delta(C)$ be the distance of $C$ and $\delta(C')$ the distance of $C'$.  We prove that $\delta(C') \geq \delta(C)$ by showing, according to Definition \ref{def:distance} of code distance, that for any Pauli $E'$ of weight less than $\delta(C)$, $\ket{\psi_i}$ and $\ket{\psi_j}$ orthonormal codewords in the code $C'$, $\bra{\psi_i'} E' \ket{\psi_j'} = f'(E') \delta_{ij}$ for some function $f'$ independent of $i$ and $j$.

The key observation is that, for any Pauli $E'$ of weight less than $\delta(C)$ and any $U \in \F$, $U^\dagger E' U$ acts on the same physical qudits as $E'$ did, and therefore, we can write $U^{\dagger}E'U = \sum_m c_m E_m$, where each $E_m$ is a Pauli operator of weight $< \delta(C)$.  Then the conclusion follows easily from the definition of distance (Definition~\ref{def:distance}):  
\begin{align}
\bra{\psi_i'} E' \ket{\psi_j'} &= \bra{\psi_i} U^{\dagger} E' U \ket{\psi_j} \\
&= \bra{\psi_i} \sum_m c_m E_m \ket{\psi_j} = \sum_m c_m \bra{\psi_i} E_m \ket{\psi_j} \\
&= \sum_m c_m f(E_m) \delta_{ij} = f'(E') \delta_{ij}
\end{align}
if we let $f'(E') := \sum_m c_m f(E_m)$.  Here, $\ket{\psi_i}$ and $\ket{\psi_j}$ are orthonormal codewords in the code $C$, and $f(E_m)$ the functions that would appear in Definition \ref{def:distance} of code distance for the code $C$.  $\bra{\psi_i'} E' \ket{\psi_j'} = f'(E') \delta_{ij}$ for any Pauli errors of weight less than $\delta(C)$ means that $\delta(C') \geq \delta(C)$.

A similar argument, namely by switching the roles of codes $C$ and $C'$, gives us $\delta(C) \geq \delta(C')$.  Therefore, we have $\delta(C) = \delta(C')$.

\end{proof}

\begin{Remark}
The Lie group $\F$ acts on $\Grass$, and partitions it into orbits.  The orbit of $C$ under $\F$ is $\M$. 
We can get very different submanifolds $\M$ from different codes $C$.  In particular, for two codes of different distance, the submanifolds corresponding to them will always be non-intersecting.  It is also possible that two codes of the same distance are contained in disjoint $\F$ orbits.  Of course, not all the $\F$ orbits are equally interesting from the point of view of quantum error correction.  Many will contain codes of distance $1$, and as such are not even able to detect a single error.  Of the remaining $\F$ orbits, some will have larger distance than others.  Also, some families of codes are more easily ``constructible''\footnote{An example of being relatively constructible is if a code $C$ has a relatively simple encoding circuit, then the $\F$ orbit containing $C$ will contain codes which have similar constructibility.} than others.
\end{Remark}

\subsubsection{Various subgroups of $\F$ or $\U{n}$ and quotients thereof}

Here we define a few subgroups of $\F$ which we will use later.

Fix an $((n,K))$ QECC $C$, and let $N = \prod_{j=1}^{n} d_j$.  Let $P$ be the projector from $\Cvect{N}$ onto the code space $C$.  Define the set $\Logical$ of \emph{logical unitary operators} as
\begin{equation}
\Logical := \{ U \in \U{N} \;|\; (\idmatrix - P) U P = 0 \}
\end{equation}
This is the same as the condition that a unitary operator $U \in \Logical$ iff $U(C) = C$.  Lemma 1 of \cite{EastinKnill} says that $\Logical$ forms a group.

The set of $\FL$ of \emph{fault-tolerant logical unitary operators}, or in this case transversal logical unitary operators, is given by the intersection
\begin{equation}
\FL := \F \cap \Logical.
\label{eqn:defFL}
\end{equation}

There is a similar but slightly different definition of logical operators in the literature, which we shall call effective logical operators in this work.  When restricted to within the unitary group $\U{N}$, the \emph{effective logical unitary operators} are given by 
\begin{equation}
\effLogical :=  \quotient{\Logical}{(U_1 \sim U_2 \text{ if  $\exists\, \xi \in \Cstar$ s.t. } U_1(x) = \xi U_2(x)\ \forall x \in C)} = \quotient{\Logical}{\Logicalperp},
\end{equation}
where $\Logicalperp$ is defined as
\begin{equation}
\Logicalperp := \{U \in \Logical \;|\; U(e_i) = \xi \cdot e_i \text{ for all } 1 \le i \le K \text{; where } \xi \in \fieldC, |\xi|=1 \text{ independent of } i\}.
\end{equation}
Here, $\{e_i\}_{i=1}^{N}$ is an orthonormal basis for $\Hilbphy$ extending a fixed basis $\{e_i\}_{i=1}^{K}$ for $C$.

Similarly, let us introduce the physical notion of the set of \emph{effective fault-tolerant logical unitary operators},
\begin{equation}
\effFL := \quotient{\FL}{(F_1 \sim F_2 \text{ if  $\exists\, \xi \in \Cstar$ s.t. } F_1(x) = \xi F_2(x)\ \forall x \in C)} = \quotient{\FL}{\FLperp}, \label{eqn:effFL}
\end{equation}
where
\begin{equation}
\FLperp := \F \cap \Logicalperp \subseteq \FL.
\end{equation}
The sets $\FLperp$ and $\Logicalperp$ so obtained do not depend on the choice of basis.  It is an easy-to-check fact that the set $\FLperp$ is a closed subgroup of $\FL$.  Therefore, by Cartan's theorem, $\FLperp$ is a Lie subgroup and $\effFL$ is at least a coset space.  Furthermore, $\FLperp$ is normal in $\FL$: Any unitary $g \in \FL$ fixes the codespace $C$, that is, it acts independently on $C = \vspan{e_1, ..., e_K}$ and on $\C^{\bot} = \vspan{e_{K+1}, ..., e_N}$.  Therefore, if $F$ satisfies $F |_{C} = \xi \idmatrix$ and $F' = gFg^{-1}$, then $F' |_{C} = \xi \idmatrix$ as well.  ($\FLperp$ is a normal subgroup of $\FL$ but is \emph{not} a normal subgroup of $\F$.)  We deduce the following proposition.

\begin{Proposition}
$\FLperp$ is a normal Lie subgroup of $\FL$, and so the quotient $\effFL = \flatquotient{\FL}{\FLperp}$ has the structure of a Lie group.
\label{prop:FLperpNormalsubLieFL}
\end{Proposition}

Let us comment further on the physical interpretation of the quotient in equation (\ref{eqn:effFL}).  The set $\FLperp$ is analogous to the notion of stabilizer for stabilizer codes, whereas the set $\FL$ is analogous to the normalizer.  So these two sets are very naturally defined and important in the quantum information context.  Their quotient is analogous to the set of ``logical operators'' in the context of stabilizer code, which corresponds as we see to the \emph{effective} version in our language.  In particular, $\effFL$ is the set of distinct logical operations that can be implemented with transversal gates on the code $C$.

\subsubsection{Manifold structure of $\M$ for transversal gates}

The set $\M$ for transversal gates was defined in equation (\ref{eqn:defM.transversal}).  Here we show that $\M$ has a manifold structure.

From the definition of $\FL$ by equation (\ref{eqn:defFL}), we see that there is the set inclusion $\FL \subset \F$.  Let us rephrase $\M$ in terms of these two sets of operators.  Consider the action of $\F$ on the Grassmannian manifold $\Grass$.  Since $\M$ is the orbit of $C$ under this $\F$ action and $\FL$ is the stabilizer of $C$ under the same action, we have
\begin{equation}
\M \simeq \quotient{\F}{\FL}.
\label{eqn:Mquotient}
\end{equation}
%

Next, we set out to establish the fact that $\FL$ is a Lie subgroup of $\F$.  Lemma 2 of \cite{EastinKnill} reads: The logical operators contained in a Lie group of unitary operators form a Lie subgroup.  Namely, if $\A$ is any Lie subgroup of $\U{N}$, then $\AL := \A \cap \Logical$ is a Lie subgroup of $\A$.  Since $\F$ is a Lie group, it follows that $\FL = \F \cap \Logical$ is a Lie subgroup of $\F$.

Combining this fact with equation (\ref{eqn:Mquotient}), we immediately establish that $\M$ is isomorphic to a \emph{coset space} for a Lie subgroup $\FL$ inside a Lie group $\F$, and therefore by Cartan's theorem, $\M$ inherits the manifold structure of the coset space, culminating in the next theorem.

\begin{Theorem}
$\M$ is a submanifold of the Grassmannian manifold $\Grass$.
\label{thm:manifoldM.transversal}
\end{Theorem}
We write $\iota_{\M}: \M \rightarrow \Grass$.

\subsection{Establishing the flat projective connections on the restriction bundles}
\label{sec:flatconnection.transversal}

Here, we try to substantiate the claim that the pullback along $\iotaM$ of the pre-connections defined by $\F$, as proposed in Sections \ref{sec:connectionV} and \ref{sec:connectionP} for $\bigvecbund$ and $\bigPbund$ respectively, indeed become flat projective connections when we take $\F$ to be the set of transversal gates and $\M$ as defined by equation~(\ref{eqn:defM.transversal}).  We adopt some arguments from \cite{EastinKnill} to help us establish this claim.

The proof revolves around the set
\begin{equation}
\C := \text{The connected component of the identity in $\FL$},
\end{equation}
which also plays an important role in \cite{EastinKnill}.  

The following lemma is the main step of the proof of Theorem 1 in \cite{EastinKnill}, and also provides the main result we need in order to prove that the proposed pre-connection is a flat projective connection on $\M$:
\begin{Lemma}
The identity connected component $\C$ of $\FL$ acts projectively trivially, namely by $\xi \cdot \idmatrix$, on any code $C$ with distance $\ge 2$.
\label{lemma:trivialaction}
\end{Lemma}

Note that we have the following inclusions:
\begin{equation}
\C \subset \FLperp \subset \FL \subset \F.
\label{eqn:4subgroups}
\end{equation}
The second set inclusion was shown to be a Lie subgroup inclusion in Proposition \ref{prop:FLperpNormalsubLieFL}, and the third inclusion is also a Lie subgroup inclusion as argued in the discussion leading up to Theorem \ref{thm:manifoldM.transversal}.
$\C$ is a Lie subgroup of $\FL$, so this is a sequence of inclusions of Lie subgroups of $\U{N}$.  Each has an interpretation in terms of loops, as given in Table~\ref{tab:groupsloops}, which provides us with clearer intuition of the role played by these groups.

\begin{table}[h]
\begin{tabular}{| c | l |}
\hline
Restricted action of & \\
subgroups of $\U{N}$ on $C$ & Is realizable on the fibre $\pi^{-1}(C)$ in $\Mvecbund$ or $\MPbund$ as \\
\hline \hline
$\C$			& \{Automorphisms of $\pi^{-1}(C)$ induced by all $\F$-loops in the bundle lifting the \\
			&  contractible $C$-based loops in $\M$\} \\
\hline
$\FLperp$ 	& \{Automorphisms of $\pi^{-1}(C)$ induced by all $\F$-loops in the bundle lifting all \\
			& \emph{stabilizer homotopy classes} of $C$-based loops in $\M$\} \\
\hline
$\FL$		& \{Automorphisms of $\pi^{-1}(C)$ induced by all $\F$-loops in the bundle lifting\\
			& arbitrary $C$-based loops in $\M$\} \\
\hline
$\F$			& \{Fibre transformations $\pi^{-1}(\gamma(0) = C) \rightarrow \pi^{-1}(\gamma(1))$ induced by all arbitrary\\
			& $\F$-paths $\gamma(t), t \in [0,1]$ in the bundle\}\\
\hline
\end{tabular}
\caption{Correspondences between subgroups of $\U{N}$ and homotopy classes in $\M$}
\label{tab:groupsloops}
\end{table}

The main step in our proof of Theorem~\ref{thm:flatconnection.transversal} will be to establish the first row of Table \ref{tab:groupsloops} using Lemma~\ref{lemma:trivialaction}.  
Rows 3 and 4 of the table are more or less by definition.  The term \emph{stabilizer homotopy class} in Row 2 of the table makes sense once the homotopy class dependence of fibre automorphism demonstrated by Row 1 of the table has been established.  In fact, we can take Row 2 of Table \ref{tab:groupsloops} as its definition.  Namely, the \emph{stabilizer homotopy classes} of loops in $\M$ refer to those homotopy classes whose lifts to an $\F$-loop in the restriction bundle(s) always yield a projectively trivial fibre automorphism, that is, it transforms the fibre by only a complex phase.  For example, in Figure~\ref{fig:fivequbit}, the loop marked ``stabilizer generator'' would produce an automorphism in Row 2, and the loops marked ``logical $X$'' and ``logical $R_3$'' give automorphisms in Row 3.  We do not prove Rows 2, 3 and 4 in this paper.

\begin{Theorem}[Flat projective connections in $\Mvecbund$ and $\MPbund$ for transversal gates]
Let $C$ be a QECC with distance $\ge 2$, $\F$ the set of transversal gates as given by equation (\ref{eqn:transversal}), and $\M = \F(C)$ as defined by equation (\ref{eqn:defM.transversal}).  Then the pre-connections proposed in Sections \ref{sec:connectionV} and \ref{sec:connectionP}, for $\Mvecbund$ and $\MPbund$ respectively, both become well-defined projective connections in the case of transversal gates; furthermore, the projective connections are projectively flat.
\label{thm:flatconnection.transversal}
\end{Theorem}

\begin{proof}

Let
\begin{eqnarray}
\I &:=& \{\text{Automorphisms of $\pi^{-1}(C)$ induced by all $\F$-loops in the bundle lifting} \\
   &   & \text{the contractible $C$-based loops in $\M$}\}. \nonumber
\end{eqnarray}
Using Lemma~\ref{lemma:trivialaction}, we can reduce the proof of Theorem \ref{thm:flatconnection.transversal} to the following claim:
\begin{equation}
\C|_C = \I,  \label{eqn:trivialclassC}
\end{equation}
interpreting Row 1 of Table~\ref{tab:groupsloops}.  Actually, only $\C|_C \supset \I$ is needed for the proof of Theorem~\ref{thm:flatconnection.transversal}, but below we shall prove both directions of inclusion.

Here is how the reduction works: Given the claim and Lemma~\ref{lemma:trivialaction}, we can quickly prove the theorem using Criterion~\ref{prop:flat_condition}. Consider any simply-connected subset $\N$ of $\M$.  If two unitary paths $U_1(t)$ and $U_2(t)$ induce paths $\gamma^{\{U_1\}}$ and $\gamma^{\{U_2\}}$ in $\N$ with the same endpoints $W_1$ and $W_2$, then the concatenated path $\gamma = \gamma^{\{U_1\}} * (\gamma^{\{U_2\}})^{-1}$ is a loop.  Since $\N$ is simply connected, $\gamma$ is homotopically trivial.  By the claim of $\C|_C \supset \I$, it induces an automorphism of $C$ in $\C$.  But by Lemma~\ref{lemma:trivialaction}, $\C$ only contains projectively trivial maps.  Therefore, $\eta^{\{U_1\}} (1) = \xi \cdot \eta^{\{U_2\}} (1)$ and $\lambda^{\{U_1\}} (1) = \xi \cdot \lambda^{\{U_2\}} (1)$.  Thus, the conditions of Criterion~\ref{prop:flat_condition} are satisfied for either $\Mvecbund$ or $\MPbund$, and the projective connection is well-defined and projectively flat.

We thus only need to prove the claim of equation~(\ref{eqn:trivialclassC}).

First, we try to establish $\C|_{C} \subset \I$.  Suppose $F \in \C$.  Since $\C$ is the connected component of the identity in the Lie group $\FL$, there exists an $\FL$-path $F(t)$ connecting $F$ to the identity $\idmatrix$.  Here, an $\FL$-path is defined analogously to an $\F$-path, except that every point $F(t)$ in the path now belongs to the subset $\FL \subset \F$.  We have $F(0) = \idmatrix$ and $F(1) = F$.  This path $F(t)$ projects to the constant path $\gamma(t) = C$ in $\M$ because $F(t) \in \FL$ for all $t$, and the constant loop is homotopically trivial. Thus, $F|_{C}$ is realizable as a $\pi^{-1}(C)$ fibre automorphism for the homotopically trivial loop $\gamma(t) = F(t)(C)$ induced by the $\F$-loop $F(t)$, so $\C|_{C} \subset \I$.

Now we show that $\C|_{C} \supset \I$.  Suppose $F(t)$ is an $\F$-loop in $\F$ such that $\gamma(t) := F(t)(C)$ is a homotopically trivial loop in $\M$, in other words, $F(t)$ induces a $\pi^{-1}(C)$ fibre automorphism in the set $\I$.  From the fact that $F(0) = \idmatrix$ and $\gamma(t)$ is homotopic to the constant loop $C$, we want to show that $F(1)$ is in $\C$, namely that $F(1)$ can be connected to $\idmatrix$ via a path in $\FL$.  Let's denote such a homotopy by $H(t,s)$ where $s \in [0,1]$, $H(t,0) = \gamma(t)$ and $H(t,1) = C$ is the constant path.  Since $F(t)$ is a lift of $\gamma(t) = H(t,0)$ from the space $\M \subset \Grass$ to the space $\F \subset \U{N}$, by the homotopy lifting property \cite{homotopy_lifting} (which always holds for the projection from a Lie group onto a quotient thereof) for the projection $p: \F \rightarrow \quotient{\F}{\FL} \simeq \M$, we can lift the homotopy $H(t,s)$ of based loops in $\M$ to a homotopy $\tilde{H}(t,s)$ of paths in $\F$ such that $\tilde{H}(t,0) = F(t)$ and $p(\tilde{H}(t,s)) = H(t,s)$.  See Figure~\ref{fig:homotopylifting}.  In particular, $\theta_2(t) := \tilde{H}(t,1)$ is a path in $\FL$ because $H(t,1) = C$; also $\theta_1(s) := \tilde{H}(0,s)$ and $\theta_3(1-s) := \tilde{H}(1,s)$ are paths in $\FL$ because $H(1,s) = C = H(0,s)$.  Concatenating the paths $\theta_1 * \theta_2 * \theta_3$ in $\FL$, we see that $\theta_3(1) = \tilde{H}(1,0) = F(1)$ is path-connected to $\theta_1(0) = \tilde{H}(0,0) = F(0) = \idmatrix$ in $\FL$.  So, $F(1)$ is in $\C$, and we have $\C|_{C} \supset \I$.

This establishes the claim $\C|_{C} = \I$, completing the proof of the theorem.

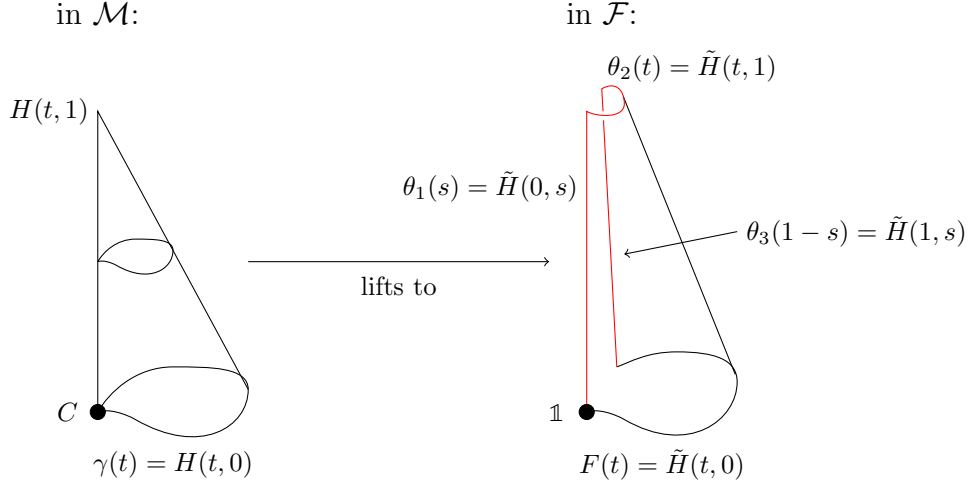
\begin{figure}
\begin{tikzpicture}
\path (-2,0) coordinate (C);
\path (C) ++(0,2) coordinate (Cmid);
\path (C) ++(0,4) coordinate (Ctop);
\path (C) ++(2,0.3) coordinate (gammamid);

\draw (C) .. controls +(20:0.15) and +(150:0.15) .. ++(0.5,-0.1) .. controls +(330:1) and +(270:0.5) .. (gammamid) .. controls +(90:0.3) and +(0:0.5) .. ++(-1,0.3) .. controls +(180:0.5) and +(60:0.3) .. (C);
\draw (C) -- (Ctop) -- (gammamid);
\draw [fill] (C) circle (0.1);
\draw (Cmid) .. controls +(20:0.075) and +(150:0.075) .. ++(0.25,-0.05) .. controls +(330:0.5) and +(270:0.25) .. ++(0.75,0.2) .. controls +(90:0.15) and +(0:0.25) .. ++(-0.5,0.15) .. controls +(180:0.25) and +(60:0.15) .. (Cmid);

\path (C) ++(6.5,0) coordinate (I);
\path (I) ++(2,0.4) coordinate (Fmid);
\path (Fmid) ++(-0.025,0.1) coordinate (Fmidup);
\path (I) ++(0.4,0.6) coordinate (Fend);
\path (I) ++(0,4) coordinate (Itop);
\path (Fmid) ++ (-1.5,3.7) coordinate (Fmidtop);
\path (Fmidtop) ++(-0.015,0.1) coordinate (Fmidtopup);
\path (Fend) ++(-0.2,3.7) coordinate (Fendtop);
\path (Fendtop) ++(0.02,-0.37) coordinate (int);
\path (int) ++(-0.003,0.07) coordinate (intup);
\path (int) ++(0.003,-0.05) coordinate (intdown);

\path (I) ++(-0.5,2) coordinate (arrowend);

\draw [->] (Cmid) ++(2,0) -- node[below=2pt] {lifts to} (arrowend);

\draw (I) .. controls +(20:0.15) and +(150:0.15) .. ++(0.5,-0.1) .. controls +(330:1) and +(270:0.5) .. (Fmid) .. controls +(90:0.3) and +(0:0.5) .. ++(-0.8,0.4) .. controls +(180:0.4) and +(20:0.2) .. (Fend);
\draw (Fmidup) -- (Fmidtopup);
\draw [red] (I) -- (Itop) .. controls +(330:0.15) and +(270:0.2) .. (Fmidtop) .. controls +(90:0.2) and +(30:0.2) .. (Fendtop) -- (intup); 
\draw [red] (intdown) -- (Fend);
\draw [fill] (I) circle (0.1);

\path (C) ++(-0.4,0) coordinate (Cleft);
\path (C) ++(1,-0.7) coordinate (gammadown);
\path (Ctop) ++(0,1.3) coordinate (Mmainlabel);

\node at (Cleft) {$C$};
\node [anchor=east] at (Ctop) {$H(t,1)$};
\node at (gammadown) {$\gamma(t) = H(t,0)$};
\node at (Mmainlabel) {{\Large in $\M$:}};

\path (I) ++(-0.4,0) coordinate (Ileft);
\path (I) ++(1,-0.7) coordinate (Fdown);
\path (I) ++(0,3) coordinate (Imidlabel);
\path (Fendtop) ++(1.2,0.3) coordinate (Ftoplabel);
\path (Fmid) ++(0,2) coordinate (Fmidlabel);
\path (Fendtop) ++(0,1) coordinate (Fmainlabel);

\node at (Ileft) {$\idmatrix$};
\node at (Fdown) {$F(t) = \tilde{H}(t,0)$};
\node [anchor=east] at (Imidlabel) {$\theta_1(s) = \tilde{H}(0,s)$};
\node at (Ftoplabel) {$\theta_2(t) = \tilde{H}(t,1)$};
\node [anchor=west] at (Fmidlabel) {$\theta_3(1-s) = \tilde{H}(1,s)$};
\draw [->] (Fmidlabel) -- ++(-1.5,-0.3);
\node at (Fmainlabel) {{\Large in $\F$:}};

\end{tikzpicture}
\caption{The homotopy $H(t,s)$ on $\M$ and the lifted homotopy $\tilde{H}(t,s)$ of paths in $\F$.}
\label{fig:homotopylifting}
\end{figure}

\end{proof}

\subsection{Implications on fault-tolerant logical gates}

\subsubsection{Logical gates and $\pi_1 (\M)$}

Since we have established a flat projective connection in the bundles $\Mvecbund$ and $\MPbund$ in the case of transversal gates, Corollary \ref{cor:FTlogicalgates} applies and says that:
\begin{equation}
\GV = \{\text{Logical gates implemented via paths in $\F$}\} = \GP
\label{eqn:transversal_logical}
\end{equation}
Equation (\ref{eqn:transversal_logical}) relates the fault-tolerant logical gates to two monodromy representations of the fundamental group of $\M$, as part of our general theory of geometric fault tolerance.  It also tells us that the logical gates are associated with the homotopy classes $\pi_1(\M,C)$.

Equation~(\ref{eqn:transversal_logical}) does not guarantee that each logical gate is associated with a \emph{unique} homotopy class.  However, we note the following:
\begin{Proposition}
In the case of transversal gates,
\begin{equation}
\pi_1 (\M, C) = \pi_0 (\FL/\U{1}).
\end{equation}
\end{Proposition}

\begin{proof}
We will define a map $\Psi: \pi_1(\M,C) \rightarrow \pi_0 (\FL/\U{1})$.  The zeroth fundamental group $\pi_0 (\FL/\U{1}) \simeq \pi_0 (\FL)$ is the set of connected components of $\FL/\U{1}$ (which is equal to the set of connected components of $\FL$).  Choose one distinguished component $\FL^{(0)}$ of $\FL$, for instance the component containing the identity.  A loop $\gamma(t)$ based at $C$ in $\M \simeq \F/\FL$ lifts to unitary paths $F(t)$ which begin and end in $\FL \subset \F$, because $\FL$ is the set of unitaries in $\F$ that fix $C$.  Given some homotopy class $[\gamma]$ of $\M$ (i.e., an element of $\pi_1(\M,C)$), choose a loop $\gamma(t)$ in $\M$ based at $C$ and lift it to a path $F(t)$ that begins in $\FL^{(0)}$.  Let $\FL^{(1)}$ be the component of $\FL$ containing $F(1)$.  Then we define $\Psi ([\gamma]):=\FL^{(1)}$.  We show below that this map is well-defined and that it is an isomorphism.

The second part of the proof of Theorem~\ref{thm:flatconnection.transversal} shows that a homotopically trivial loop in $\M$ lifts to a path $F(t)$ in $\F$ that is homotopically equivalent (in $\F$) to a path that lies completely within $\FL$.  In fact, essentially the same argument using the homotopy lifting property shows that if $\gamma$ and $\gamma'$ are \emph{any} two homotopic loops in $\M$ and $F(t)$ is a lift of $\gamma$ to $\F$, then $F(t)$ is homotopic in $\F$ to some lift $F'(t)$ of $\gamma'$, that $F(0)$ and $F'(0)$ are connected by a path within $\FL$, and that $F(1)$ and $F'(1)$ are connected by a path within $\FL$.  The only difference in the argument is that $H(t,1) = \gamma'(t)$ instead of being a constant path.  The resulting situation is shown in Figure~\ref{fig:pathsbetweenFL}.

Furthermore, any pair of lifts $F_0(t)$ and $F_1(t)$ of a loop $\gamma$ in $\M$ which start in the same component of $\FL$ also end in the same component of $\FL$:  Suppose first that
$F_0(0) = F_1(0)$.  Then the path $F_0^{-1} * F_1$ is a lift
of the homotopically trivial loop $\gamma^{-1} * \gamma$, and by
the argument in the proof of Theorem~\ref{thm:flatconnection.transversal}, the lifted path is
homotopically equivalent to a path lying completely within $\FL$.  In
particular, there is a path within $\FL$ between $F_0(1)$ and
$F_1(1)$, so they are in the same component of $\FL$.  If $F_0(0)
\neq F_1(0)$, but they are in the same component of $\FL$, then they
are connected by a path $G(t)$ and we can consider $F_0^{-1} * G * F_1$ instead of $F_0^{-1} * F_1$.

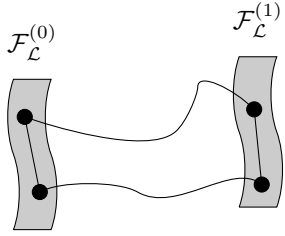
\begin{figure}
\begin{tikzpicture}

\path (0,0) coordinate (start1);
\path (start1) ++(3,0.3) coordinate (start2);
\path (start1) ++(0.35,0.5) coordinate (pathA0);
\path (start1) ++(0.15,1.5) coordinate (pathB0);
\path (start2) ++(0.3,0.3) coordinate (pathA1);
\path (start2) ++(0.2,1.3) coordinate (pathB1);
\path (start1) ++(0.25,2.5) coordinate (FL0);
\path (start2) ++(0.25,2.5) coordinate (FL1);

\draw [fill=black!20] (start1) -- ++(0.5,0) .. controls +(70:0.3) and +(290:0.3) .. ++(0,1) .. controls +(110:0.3) and +(250:0.3) .. ++(0,1) -- ++(-0.5,0) .. controls +(250:0.3) and +(110:0.3) .. ++(0,-1) .. controls +(290:0.3) and +(70:0.3) .. (start1);

\draw [fill=black!20] (start2) -- ++(0.5,0) .. controls +(70:0.3) and +(290:0.3) .. ++(0,1) .. controls +(110:0.3) and +(250:0.3) .. ++(0,1) -- ++(-0.5,0) .. controls +(250:0.3) and +(110:0.3) .. ++(0,-1) .. controls +(290:0.3) and +(70:0.3) .. (start2);

\draw [fill] (pathA0) circle (0.1);
\draw [fill] (pathA1) circle (0.1);
\draw [fill] (pathB0) circle (0.1);
\draw [fill] (pathB1) circle (0.1);

\draw (pathA0) .. controls +(30:0.3) and +(150:0.3) .. ++(1.3,0) .. controls +(330:0.6) and +(140:0.5) .. (pathA1);
\draw (pathB0) .. controls +(340:0.3) and +(220:0.5) .. ++(2,-0.2) .. controls +(40:0.2) and +(270:0.1) .. ++(0.3,0.6) .. controls +(90:0.2) and +(135:0.4) .. (pathB1);
\draw (pathA0) -- (pathB0);
\draw (pathA1) -- (pathB1);

\node at (FL0) {$\FL^{(0)}$};
\node at (FL1) {$\FL^{(1)}$};

\end{tikzpicture}
\caption{Two paths that start and end in same connected components of $\FL$ (shaded areas) can be put together with paths in $\FL$ to form a closed loop.  Since $\pi_1(\F) = 0$, the loop is contractible and therefore projects to the trivial homotopy class in $\M$.  Thus the two paths project to the same homotopy class in $\M$.}
\label{fig:pathsbetweenFL}
\end{figure}

As a consequence, given a particular homotopy class of $\M$, the lift of any path from it that starts in $\FL^{(0)}$ will end in a particular connected component $\FL^{(1)}$ of $\FL$.  The same statement holds if we work with $\F/\U{1}$ and $\FL/\U{1}$.  Therefore, $\Psi$ is well defined.

Because $\F/\U{1} = \bigotimes \SU{d_j}$ is simply connected, the converse is also true: Two paths in $\F/\U{1}$ that start and end in the same connected components of $\FL/\U{1}$ are homotopic in $\F$, so they project down to the same homotopy class in $\M$, as illustrated in Figure~\ref{fig:pathsbetweenFL}.  Therefore $\Psi$ is one-to-one.  Given any component $\FL^{(1)}$ of $\FL \subset \F$, we can choose a path $F(t)$ with $F(0) \in \FL^{(0)}$ and $F(1) \in \FL^{(1)}$, and get a path $\gamma(t)$ in $\M$ by projecting.  This shows that $\Psi$ is onto. Thus, the homotopy classes of $\M$ correspond precisely to the connected components of $\FL/\U{1}$.  That is, $\pi_1 (\M) = \pi_0 (\FL/\U{1})$.
\end{proof}

By the result of \cite{EastinKnill}, $\FL/\U{1}$ is discrete whenever the QECC $C$ has distance at least $2$, and thus $\pi_0 (\FL/\U{1}) = \FL/\U{1}$.  Therefore, the homotopy classes of loops in $\M$ exactly correspond to the fault-tolerant logical unitary operators (up to global phase).  However, the \emph{effective} fault-tolerant logical unitary operators might be a smaller set, as some non-contractible loops in $\M$ might perform different unitary operations on all of $\Cvect{N}$ but might have the same monodromy action on the fibre $C$, and hence perform the same logical operation.

\subsubsection{Transversal gates example: The five-qubit code}
\label{sec:fivequbit}

In principle, it might be possible to study the manifold $\M$ associated with a particular code $C$, determine $\pi_1 (\M, C)$, and use that to figure out the set of possible logical transversal gates for $C$.  In practice, however, this seems to be quite difficult to do.  At the moment, the only way we know to figure out the topology of $\M$ is to study the transversal gates of $C$ using other methods and then use Theorem~\ref{thm:flatconnection.transversal} to tell us about $\M$.  In this section, we will give an example of this for a single block of the five-qubit code~\cite{fivequbit,BDSW}.
 
The five-qubit code is a stabilizer code, so all unitaries in the stabilizer $S$ preserve the code space and perform the logical identity operator.  Furthermore, elements of $N(S)$, the set of Paulis that commute with the stabilizer, also preserve the code space.  The cosets $N(S)/S$ correspond to distinct effective logical operations; in fact, $N(S)/S$ in this case is isomorphic to the logical Pauli group.  The five-qubit code has two additional logical transversal gates $R_3$ and $R_3^2$~\cite{GottesmanFT}.  $R_3$ is a Clifford group operation that, under conjugation, maps $X \mapsto Y$, $Y \mapsto Z$, and $Z \mapsto X$, and on the five-qubit code, the logical $R_3$ gate can be performed transversally via $R_3^{\otimes 5}$, i.e., by performing the physical $R_3$ gate on each qubit of the code.  For the five-qubit code, there are no non-Clifford transversal gates that preserve the code space~\cite{Rains}.

The generators of $\FL$ for the five-qubit code are therefore the generators of $N(S)$ plus $R_3$.  $N(S)$ has $6$ generators --- $4$ stabilizer generators plus generators corresponding to the logical $X$ and logical $Z$.\footnote{The logical $R_3$ is not performed using Pauli operators and is therefore not in $N(S)$.} Because $\FL/\U{1}$ is a discrete group in this case (as the five-qubit code has distance $\geq 2$), each of the seven generators of $\FL/\U{1}$ corresponds to a generator of $\pi_1 (\M)$.  The monodromy representation of $\pi_1 (\M)$ on $\bigvecbund$ has an image $\effFL$ with $3$ generators $X$, $Z$, and $R_3$.  The $4$ stabilizer generators correspond to non-trivial loops in $\M$ (generators of $\pi_1(\M)$) that nonetheless have trivial holonomy on $\bigvecbund$; these loops generate $\FLperp$.

\begin{figure}
\begin{tikzpicture}

\draw (0,0) arc (210:330:2cm and 1cm) arc (-120:120:3cm and 1.5cm) arc (30:150:2cm and 1cm) arc (30:270:0.75cm and 0.375cm) arc (90:270:0.75cm and 0.375cm) arc (90:270:0.75cm and 0.375cm) arc (90:330:0.75cm and 0.375cm) -- cycle;

\path (150:0.75cm and 0.375cm) coordinate (A);

\foreach \y in {0,0.75,1.5,...,2.25}
	{
	\draw (A) ++(0,0.1) ++(0,\y) ++(330:0.4cm and 0.2cm) arc (330:210:0.4cm and 0.2cm);
	\draw (A) ++(0,0.1) ++(0,\y) ++(315:0.4cm and 0.2cm) arc (45:135:0.4cm and 0.2cm);
	}

\path (30:2cm and 1cm) coordinate (B);
\path (B) ++(0,1.5) coordinate (C);
\path  (B) ++(330:2cm and 1cm) ++(60:3cm and 1.5cm) coordinate (D);

\draw (B) ++(0,0.25) ++(330:1cm and 0.5cm) arc (330:210:1cm and 0.5cm);
\draw (B) ++(0,0.25) ++(315:1cm and 0.5cm) arc (45:135:1cm and 0.5cm);

\draw (C) ++(0,0.25) ++(330:1cm and 0.5cm) arc (330:210:1cm and 0.5cm);
\draw (C) ++(0,0.25) ++(315:1cm and 0.5cm) arc (45:135:1cm and 0.5cm);

\draw (D) ++(0,0.375) ++(330:1.5cm and 0.75cm) arc (330:210:1.5cm and 0.75cm);
\draw (D) ++(0,0.375) ++(315:1.5cm and 0.75cm) arc (45:135:1.5cm and 0.75cm);

\path (B) ++(0,1) ++(330:1cm and 0.5cm) coordinate (start);
\draw [fill] (start) circle (0.1cm);

\path (A) ++(0,0.55) coordinate (stabdown);
\draw [dotted] (start) .. controls +(210:1cm) and +(0:1cm) .. (stabdown) arc (270:90:0.5cm and 0.2cm) .. controls +(0:1cm) and +(180:1cm) .. (start);

\path (C) ++(-1.3,0) coordinate (Xleft);
\draw [dashed] (start) .. controls +(165:0.6cm) and +(270:0.3cm) .. (Xleft) arc (180:0:1.3cm and 0.4cm) .. controls +(270:0.3 cm) and +(75:0.3cm) .. (start);

\path (D) ++(210:1.8cm and 0.9cm) coordinate (Tleftish);
\draw [densely dashed] (start) .. controls +(350:0.5cm) and +(150:0.3cm and 0.15cm) .. (Tleftish) arc (-150:150:1.8cm and 0.9cm) .. controls +(210:0.3cm and 0.15cm) and +(10:0.5cm) .. (start);

\path (A) ++(-2.5,0.75) coordinate (stablabel);
\path (stablabel) ++(0,-0.35) coordinate (stablabeldown);
\node at (stablabel) {stabilizer};
\node at (stablabeldown) {generator};
\draw [->] (stablabel) ++(0.75,0) --  ++(1.2,0);

\path (C) ++(0,1.5) coordinate (Xlabel);
\node at (Xlabel) {logical $X$};
\draw [->] (Xlabel) ++(0,-0.2) -- ++(0,-0.75);

\path (D) ++(2,1.8) coordinate (R3label);
\node at (R3label) {logical $R_3$};
\draw [->] (R3label) ++(0,-0.2) -- ++(-0.8,-0.8);

\end{tikzpicture}
\caption{A cartoon of $\M$ for single-block transversal gates for the $5$-qubit code.  The actual $\M$ is $15$-dimensional, which is hard to visualize, so instead we draw a $2$-manifold with a fundamental group with the same number of generators.}
\label{fig:fivequbit}
\end{figure}
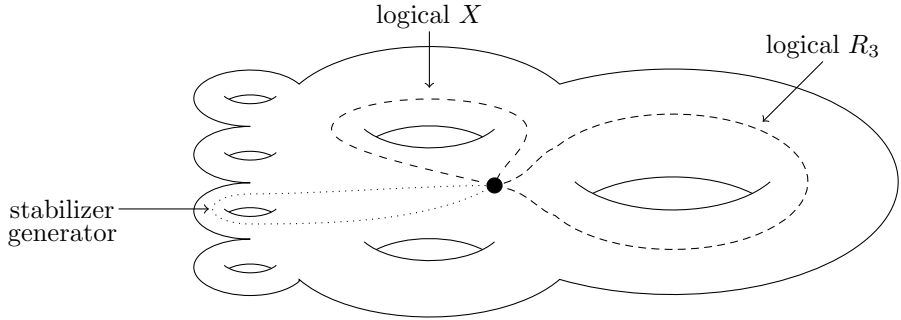

\section{Example II: Toric codes and string operators}
\label{sec:topologicalFT}

In Section \ref{sec:transversal}, we studied the restriction fibre bundles for transversal gates acting on QECCs with a fixed distance.  We claim that we can build a similar picture for string operators acting on the toric code \cite{Kitaev1} while preserving the number of each type of defects.  Some crucial features in this case will be different from those for transversal gates and are very much worth noting.  At the end of this section, we describe how to apply the same picture for Kitaev's quantum double model based on any finite gauge group $G$ or for more general Levin-Wen string net models.

Note that the toric code is a QECC of distance about $\sqrt{N}$, and the string operators we will be working with are transversal operations.  Therefore, the fibre bundle picture from Section~\ref{sec:transversal} would be a valid way of understanding these gates.  However, in this section we will develop a different picture that emphasizes the geometric structure of the toric code and can be generalized to other topological codes.

We focus on the case of the toric code with a number of defects, and fault-tolerant gates performed by moving the defects on the torus.  As we shall elaborate in this section, a defect can be moved a small distance using gates acting only locally, making this an appropriate set of fault-tolerant gates for the toric code and the 2D $(s,t)$-geometrically local error model, as discussed in Section~\ref{sec:local_errors}.  

The details of this construction are fairly complicated, but we can outline the main ideas quickly.  
The ultimate goal is to construct $\M$, a subset of $\Grass$ which is isomorphic to the configuration space of some number of points (which come in two species) on the torus, and we impose an additional ``hard-core'' constraint that the points (known as ``defects'') are not too close together.  That is, we wish to associate to each defect configuration a subspace of the Hilbert space.  The toric code gives a standard way to do this when the defects lie on any of the points of some standard lattice (for one species of point) or the dual lattice (for the other species), defining a quantum error-correcting code for each such configuration of defects.  To associate a subspace with a configuration not of this form, we interpolate between the standard quantum codes, defining a QECC whose codewords are superpositions of standard toric codes with defect configurations close to the desired one.  We also define geometrically local unitary operators that relate subspaces associated with similar defect configurations.  This produces a continuous space $\M \subseteq \Grass$ and allows us to show that $\M$ is flat.

In Sections~\ref{sec:toric_string} and \ref{sec:discrete_string}, we introduce the toric code with defects and the standard method of moving defects around via string operators.  In Section~\ref{sec:continuous_toric}, we show how to do the interpolation to move defects continuously on the edges of the graph, and in Section~\ref{sec:toricM}, we extend this to allow defects to be moved to any location on the surface of the torus.  We use this to define the desired $\M$ and show that it is isomorphic to the configuration space of the defects (subject to the hard-core condition).
Ultimately this enables us to show in Section~\ref{sec:toric_flatFPP} that we can apply the picture developed in Section~\ref{sec:geoFT}, getting a flat connection for the two usual fibre bundles over $\M$.  In Sections~\ref{sec:toricmonodromy} and \ref{sec:toricfreedom}, we comment on a few aspects of the construction.  Finally, in Section~\ref{sec:otheranyons}, we show how to apply the construction to other topological models, including many with non-Abelian anyons.

\subsection{The toric codes and the string operators}
\label{sec:toric_string}

\subsubsection{The toric codes}
\label{sec:toric}

Let us fix a square lattice $\Gamma$ on (and ``covering'') a genus-$g$ torus $T_g$.  Now imagine having one qubit sitting at each edge in the edge set $E$, with its own computational basis $\{\ket{0}, \ket{1}\}$.  As introduced by Kitaev \cite{Kitaev1}, the total physical Hilbert space of the quantum system for the toric code is
\begin{equation}
\Hilbphy = \bigotimes_{e \in E} (\Cvect{2})_e = \Cvect{N}
\end{equation}
where $N = 2^{|E|}$.  Let $\vec{\beta} = \{\beta_1, ..., \beta_{|E|}\}$ be a fixed ordered collection of qubit computational bases, one for each edge Hilbert space $\Cvect{2}$. 

Then the {\it toric code} $C_K$ refers to the ground state subspace for the following Hamiltonian:\footnote{All subscripts $K$ in this section stand for Kitaev model, of which the toric code is a special case.}
\begin{equation}
H_K = - \sum_{v \in V} A_v - \sum_{f \in F} B_f
\label{eqn:Hamiltonian.toric}
\end{equation}
where $V$ and $F$ are the sets of vertices and faces, respectively, for the graph $\Gamma$, $A_v$ acts on the four edges meeting at a vertex $v$ by
\begin{equation}
A_v = X \otimes X \otimes X \otimes X
\label{eqn:vertexop}
\end{equation}
and $B_f$ acts on the four edges around a face $f$ by
\label{eqn:faceop}
\begin{equation}
B_f = Z \otimes Z \otimes Z \otimes Z,
\end{equation}
where the Pauli operators $X$ and $Z$ are defined with respect to the fixed bases $\vec{\beta}$ for the physical qubits.  

Let $n_v$ and $n_f$ be two even non-negative integers.  We can consider a \emph{modified toric code} or \emph{toric code with defects} $\Ctoricfix$, where $S_v = (v_1, \cdots, v_{n_v})$ and $S_f = (f_1, \cdots, f_{n_f})$ are sets of vertices and faces respectively.  The code $\Ctoricfix$ is defined as the ground state subspace of the modified Hamiltonian
\begin{equation}
\Htoric = - \!\!\! \sum_{v \in V \setminus S_v} A_v - \!\!\! \sum_{f \in F \setminus S_f} B_f + \sum_{v \in S_v} A_v + \sum_{f \in S_f} B_f.
\label{eqn:Hamiltonian.defects}
\end{equation}
The Hamiltonian $\Htoric$ always has the same ground state degeneracy as the Hamiltonian $H_K$.  For example, if we work with the genus-1 torus $T_1$, the dimension of $\Ctoricfix$ is always $K=4$ for any legal $(n_v, n_f)$ or $(S_v, S_f)$ and regardless of our choice of the lattice.  For our present discussion, we work with a fixed lattice $\Gamma$.  Suppose, for a moment, we choose the lattice $\Gamma$ on the torus $T_1$ to be a square lattice of size $L\times L$; then the total Hilbert space dimension is $N = 2 L (L-1)$.  This can be derived by counting the number of distinct edges, horizontal and vertical, as in \cite{Dennis}.  Therefore, with this special choice of $\Gamma$, we are working with the Grassmannian $\Grass = \mathrm{Gr}(4, 2L(L-1))$.  In general, however, the number $N$ varies with different choices of the lattice even for the same torus, while the number $K$ is lattice-independent.  If we fix the lattice $\Gamma$ as well as the torus $T_g$, as is the assumption in this paper, then both $K$ and $N$ will be fixed for any legal choice of the pair $\locations$.

Physically, the number $n_v$ represents the number of \emph{primal defects} (which live at the vertices of $\Gamma$ specified by $S_v$) whereas $n_f$ represents the number of \emph{dual defects} (which live on the faces of $\Gamma$ specified by $S_f$).  In particular, the original code $C_K$ is just $\Ctoricfix$ for $n_v = 0 = n_f$ and it is the unique four-dimensional subspace that corresponds to having no defects of any type, given a fixed graph $\Gamma$ on the genus-$g$ torus $T_{g}$.  

It is also helpful to think about the dual lattice:
\begin{Definition}
Let $\Gamma$ be a graph embedded in a $2$-manifold.  The \emph{dual graph} $\bar{\Gamma}$ is a graph whose vertices are the faces of $\Gamma$.  $\bar{\Gamma}$ has an edge between two of its vertices if and only if the two corresponding faces of $\Gamma$ share an edge of $\Gamma$.  The faces of $\bar{\Gamma}$ are thus the vertices of $\Gamma$.
\end{Definition}
Based on this definition, we can see that the primal defects can equivalently be thought of as living on the faces of the dual graph $\bar{\Gamma}$ and dual defects living on the vertices of $\bar{\Gamma}$.

\subsubsection{The string operators}
\label{sec:strings}

A \emph{string operator} $F$ on $\Hilbphy = \bigotimes_{e \in E} (\Cvect{2})_e$ is just a tensor product of Pauli operators or the identity operator on the individual qubits, along with one of the phases $\phases$.  It can be written as
\begin{equation}
F = \xi \cdot O_1 \otimes \cdots \otimes O_{|E|}
\end{equation}
where each $O_i$ is an operator acting on the $\ord{i}$ qubit and is chosen from the set $\Paulis$, whereas $\xi$ is a phase factor from the set $\phases$.  In the case where $O_i = Y$, we can use equation (\ref{eqn:Paulis}) to rewrite $O_i = Y = iXZ$ as a product of $X$ and $Z$, hence we can product decompose the whole string operator $F$ into, for example, $\xi' \cdot F_X F_Z$ where the operator $F_X$ is a tensor product of purely $X$ and $\idmatrix$, $F_Z$ purely $Z$ and $\idmatrix$, and $\xi' \in \phases$ a new phase factor.  Alternatively, since $X$ and $Z$ anticommute, that is,
\begin{equation}
XZ = - ZX,
\end{equation}
we can equally well write a general string operator $F$ as $\xi'' \cdot F_Z' F_X'$.

Now, why do we call these simple tensor products of Pauli operators the string operators?  This is because we usually represent the physical qubits as edges on a graph, and we like to represent the action of a nontrivial Pauli operator on a qubit by ``darkening'' that edge.  When a bunch of edges, especially some geometrically connected ones, are darkened, they tend to look like a bunch of strings living on the lattice, as in Figure~\ref{fig:stringoperators}.  If we try to represent a string operator $F$ by darkening or coloring the graph $\Gamma$, in general it might appear that we will need 3 colors, one for each of $X, Y$ and $Z$ (plus 1 lighter or default shade for the identity operator).  However, after we decompose $F$ into the form $\xi \cdot F_X F_Z$, we only need 2 colors, one to color the edges upon which $F_X$ acts nontrivially (by $X$), and the other to color the edges where $F_Z$ acts nontrivially (by $Z$).  We call the string operator corresponding to $F_X$ an $X$-string operator, and the string operator corresponding to $F_Z$ a $Z$-string operator.  We typically think of the $Z$-string operators as living on the edges of $\Gamma$ and the $X$-string operators as living on the edges of the dual lattice $\bar{\Gamma}$, because $X$-string operators produce defects on faces of $\Gamma$; this also avoids the possibility of needing color a single edge with both colors.  
 
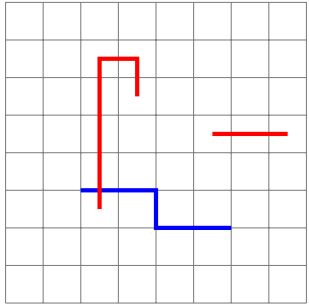
\begin{figure}
\begin{tikzpicture}
[X/.style={ultra thick,red}, Z/.style={ultra thick,blue}]
\draw [gray,step=.5cm] (-2,-2) grid (2,2);
\draw [Z] (-1,-0.5) -- (0,-0.5) -- (0,-1) -- (1,-1);
\draw [X] (-.75,-.75) -- (-.75,1.25) -- (-0.25,1.25) -- (-0.25,.75);
\draw [X] (0.75, 0.25) -- (1.75,0.25);

\end{tikzpicture}
\caption{An $X$-string operator (red, through faces) and a $Z$-string operator (blue, through vertices) for the toric code.  Note that in this example, the $X$-string operator consists of two disconnected pieces.}
\label{fig:stringoperators}
\end{figure}

\subsection{Discrete-time evolutions and the discrete configuration space $\Ctorichardcore$}
\label{sec:discrete_string}

If we want to do computation with the toric codes, we must introduce dynamics to the quantum system.  This is done, for example, by introducing a time flow to the standard string operators, representing the order in which the operations are done.  This has a natural interpretation in terms of moving defects around on the torus, or creating or annihilating pairs of defects.  We begin by defining a discrete time flow below, then we extend it to a continuous time flow in Section \ref{sec:continuous_toric}.

\subsubsection{Discrete-time string evolutions}

Let's think about \emph{how} we are going to apply a string operator $F$.  We could apply it all at once, namely instantaneously.  However, the usual strategy for gates in a topological code is to move defects around in a geometrically local fashion.  Therefore, we want to think about the toric code in the context of an $(s,t)$-geometrically local error model.  For such an error model, a large string operator could create a large cluster of errors.\footnote{In fact, for the toric code, the string operators are also transversal gates, so they do not actually produce large clusters.  However, for more general topological codes, the string operators are not tensor products.  The point really is to come up with an approach that still works for the general topological code.}  To implement a string operator fault-tolerantly, we should therefore grow the string operators slowly --- perform a small string operator, then stop and do error correction, then continue with a second small string operator, and so on, until we have built up the full string operator we wish to perform.

For the moment, let us imagine that we have a specific fixed way to implement the Pauli operators $X$ and $Z$ on every qubit $e \in E$, and furthermore let us make the temporary simplification that the implementation of a single-qubit $X$ or $Z$ is instantaneous.  (We will relax this assumption shortly.)  To introduce a discrete time flow, we are then just talking about specifying an order in which to implement the $O_i$'s for a general string operator $F = \xi \cdot O_1 \otimes \cdots \otimes O_{|E|}$.  We call a string operator with a specified discrete-time flow a \emph{discrete-time string evolution}, whose formal definition is given as follows:

\begin{Definition}
\label{dfn:discretestringev}
A \emph{discrete-time string evolution} is a sequence $(P_1, \ldots, P_l)$, where each $P_i$ is either $X$ or $Z$ acting on a single qubit.  Equivalently, we can represent a discrete-time string evolution as a sequence $(a_k)_{k=1,l}$, where each $a_k$ denotes either a \emph{primal edge} (an edge of the lattice $\Gamma$) expressed in the form of $e_k$ or a \emph{dual edge} (an edge of the dual lattice to $\Gamma$) in the form of $\bar{e}_k$.  To convert such a sequence $(a_k)_{k=1,l}$ to an ordered application of Pauli operators, we simply interpret an element $e_k$ in the sequence as ``apply the Pauli $Z$ operator to the primal edge $e_k$'', and we read an element $\bar{e}_k$ as ``apply the Pauli $X$ operator to the dual edge $\bar{e}_k$''.  Then we compose the sequence of Pauli operators in the order $P_l \circ \cdots \circ P_1$.  A discrete-time string evolution $(a_k)_{k=1,l}$ is said to realize the string operator $F$ if $P_l \circ \cdots \circ P_1 = F$.
\end{Definition}

There are a couple of points to note about this definition.  First, we do not insist that all the edges in the sequence be distinct.  It is allowed to have edge $e$ or $\bar{e}$ appear multiple times in a sequence, as in Figure~\ref{fig:tangle}.  Second, this particular evolution results in a particular choice of phase for the string operator $F$.\footnote{However, only the two phases corresponding to real string operators can be achieved in this way; i.e., phase $\pm 1$ if the number of $Y$ tensor factors in the string operator is even and phase $\pm i$ if the number of $Y$ factors is odd.}

At each step in a discrete-time string evolution, we perform a single Pauli operator.  It is standard to view string operators as a means to transport defects, as in for example \cite{Dennis}. A $Z$ string segment transports a primal defect from one end of the string to the other and an $X$ string segment (on the dual lattice) transports a dual defect from one of its ends to the other.   A discrete-time string evolution fills out this picture, giving us a history of the process in which one defect hops \emph{by one edge} at each time step.  When a primal edge $e$ is performed on a configuration with a primal defect at one end of the edge, the defect hops to the other end of $e$, and similarly a dual edge $\bar{e}$ causes a dual defect to hop from one end of the dual edge to the other.  When $e$ is  performed on a configuration where there are no defects at either end of $e$, it instead creates a pair of primal defects at its ends, and when $e$ is performed on a configuration with primal defects at both ends, it causes the two defects to annihilate.  Similarly, a dual edge $\bar{e}$ can create or annihilate pairs of dual defects.  Primal edge operators have no direct effect on dual defects and dual edge operators have no direct effect on primal defects, but in both cases, there may be some global effect on the phase of the overall string operator.\footnote{In Kitaev models with non-abelian anyons, the phase could be replaced by a unitary on the code space.}

We can interpret a discrete-time evolution acting on a code $\Ctoricfix$ as a rectilinear tangle.  An example is shown in Figure~\ref{fig:tangle}.  However, normally we define tangles only up to isotopy, and we are explicitly \emph{not} doing that here.  Instead, we insist the tangle has, at each time step, only a single thing happening: one strand moves by one step, a cup (creation of a pair of defects), or a cap (annihilation of two adjacent defects).  Isotopic tangles will realize the same string operator (or at any rate, ones that are equivalent in that they have the same action on $\Ctoricfix$), but represent different discrete-time string evolutions.  Thus, there is not a unique discrete-time string evolution realizing a given string operator.

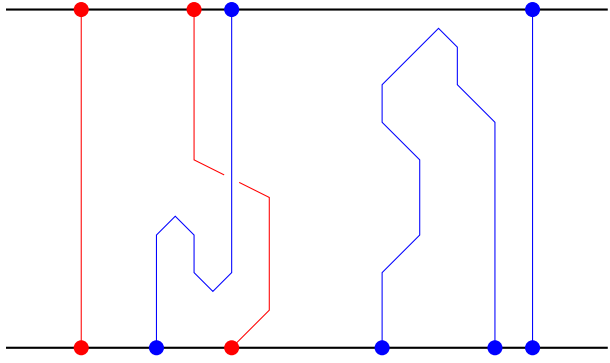
\begin{figure}
\begin{tikzpicture}
[defect/.style={circle,inner sep=0pt,minimum size=2mm}]

\draw [thick] (-1,0) -- +(8,0);
\draw [thick] (-1,4.5) -- +(8,0);

\draw [red] (0,0) -- (0,4.5);

\draw [blue] (1,0) -- (1,1.5) -- (1.25,1.75) -- (1.5,1.5) -- (1.5,1) -- (1.75,0.75) -- (2,1) -- (2,4.5);

\draw [red] (2,0) -- (2.5,0.5) -- (2.5,2) -- (2.1,2.2);
\draw [red] (1.9,2.3) -- (1.5,2.5) -- (1.5,4.5);

\draw [blue] (4,0) -- (4,1) -- (4.5,1.5) -- (4.5,2.5) -- (4,3) -- (4,3.5) -- (4.75,4.25) -- (5,4) -- (5,3.5) -- (5.5,3) -- (5.5,0);

\draw [blue] (6,0) -- (6,4.5);

\foreach \x in {0,2}
\node at (\x,0) [defect,fill=red] {};
\foreach \x in {0,1.5}
{\node at (\x,4.5) [defect,fill=red] {};}

\foreach \x in {1,4,5.5,6}
{\node at (\x,0) [defect,fill=blue] {};}
\foreach \x in {2,6}
{\node at (\x,4.5) [defect,fill=blue] {};}

\end{tikzpicture}
\caption{An example of a tangle.  This example cannot be realized as a verified discrete-time string evolution, since there are two events where defects annihilate and one where a pair of defects is created.}
\label{fig:tangle}
\end{figure}

\subsubsection{Verified string operators and evolutions}

One annoyance with discrete-time string evolutions as defined in Definition~\ref{dfn:discretestringev} is that the same evolution can correspond to very different tangles depending on the initial state.  For instance, the single Pauli $Z$ applied to primal edge $e$ can move a primal defect along the edge $e$ (in either direction), it can create a pair of defects (a cup), or it can annihilate a pair of defects (a cap).  We prefer to work with a more controlled Hilbert space, so we instead work with slightly more complicated unitary operators:
\begin{Definition}
Let $\ver{Z}_i$ be the unitary operator that performs $Z_i$ if there is exactly one primal defect on one of the vertices of the edge $e$ corresponding to qubit $i$ and performs $\idmatrix$ if there are zero or two primal defects on the vertices of $e$.  Similarly, let $\ver{X}_i$ be the unitary operator that performs $X_i$ if there is exactly one dual defect on the one of the dual vertices of the dual edge $\bar{e}$ corresponding to qubit $i$ and performs $\idmatrix$ otherwise.  We can also let $\ver{Y}_i = i \ver{X} \ver{Z}$, but we will not generally need to use this operator.  We call $\ver{X}_i$, $\ver{Y}_i$ and $\ver{Z}_i$ the \emph{verified Pauli operators} because they first check that a defect is present before moving it.

We can then define a \emph{verified discrete-time string operator} as a tensor product of verified Pauli operators with phase $\pm 1$ or $\pm i$, just replacing the usual Paulis in a string operator with verified Paulis.  Similarly, a \emph{verified discrete-time string evolution} is a sequence of verified $\ver{X}$ and $\ver{Z}$ Pauli operators.  A verified discrete-time string evolution realizes a verified discrete-time string operator in just the same way as a discrete-time string evolution realizes a regular string operator.
\end{Definition}

Note that verified string evolutions cannot create or annihilate pairs of defects, so at the end of a verified discrete-time string evolution, the numbers and types of defects are just the same as in the initial configuration of defects.  Only the positions have changed.  Thus, a tangle realized by a verified discrete-time string evolution cannot have any ``cups'' or ``caps'' in it.  For instance, the tangle in Figure~\ref{fig:tangle} cannot be realized in this way.

If the qubit $i$ is on edge $e$ with vertices $v$ and $v'$, the verified Pauli is
\begin{equation}
\ver{Z}_i = \frac{1}{4} [(\idmatrix - A_v) (\idmatrix + A_{v'}) + (\idmatrix + A_v)(\idmatrix - A_{v'})] Z_i + \frac{1}{4} [(\idmatrix + A_v) (\idmatrix + A_{v'}) + (\idmatrix - A_v)(\idmatrix - A_{v'})] \idmatrix.
\label{eq:verifiedPauli}
\end{equation}
This can be realized by acting on the qubit $i$ and the six other qubits adjacent to $v$ and $v'$.  Similarly, $\ver{X}_i$ can also be realized by acting on $7$ qubits.  Also note that $\ver{Z}_i^2 = \ver{X}_i^2 = \idmatrix$.

\subsubsection{Hard-core discrete-time string evolutions}

Our next step is to restrict attention to verified discrete-time evolutions which respect a ``hard-core'' condition at all times.  In particular, we insist that at all times, the defects in the code are separated by a minimum separation $\sep \geq 3$.  We count $\sep$ in terms of lattice distance, so lattice points $(x,y)$ and $(x',y')$ have separation $|x-x'| + |y-y'|$.  We also enforce the same minimum separation between primal and dual defects, counting the lattice distance from the vertex supporting the primal defect to the closest vertex of the face containing the dual defect.  
The hard-core condition constrains the movement of defects that are already at the minimum distance apart, limiting the possibilities for string evolutions.

For our analysis of the toric code, it is sufficient to consider $\sep = 3$, the minimum needed to avoid having two defects on the same face, but larger $\sep$ may be needed for more general anyon models.  The separation is needed for our definitions of various operators later in this section.  The topology of $\M$ does not change if we take a larger value of $\sep$ unless it is so large that the defects become unable to move while still satisfying the hard-core constraint.

The hard-core requirement more or less implements the idea that only codes of large distance should arise during the moving of defects --- when two defects get close together, even a small cluster of errors can braid them around each other.  Of course, if we were actually storing information in the defects, it would not be sufficient to keep them far apart; we would also need the defects to physically be \emph{holes} of large diameter.  Again, the point here is to work in a simplified model that contains many of the main features of a more complicated topological code.

Satisfying the hard-core condition is not an intrinsic property of a discrete-time string evolution $(a_k)$, or even of a verified discrete-time string evolution (since that only checks the locations of defects that move).  Instead, it is a property of a code $\Ctoricfix$ together with an evolution $(a_k)$.  Suppose we have performed $r$ time steps of $(a_k)$.  This has moved the defects in $\Ctoricfix$, giving us instead $\Ctoricfixend$.  Because the evolution so far has satisfied the hard-core condition, no defects have been created or annihilated, and thus $\Ctoricfixend$ still has $(|S_v|,|S_f|)$ defects.  To continue to satisfy the hard-core condition at time $r+1$, the edge $a_{r+1}$ must be adjacent to a primal defect in $\Ctoricfixend$ (if $a_{r+1}$ is a primal edge) or to a dual defect in $\Ctoricfixend$ (if $a_{r+1}$ is a dual edge).  Furthermore, the other end of the edge must not be distance $< \sep$ from any other defect.

The hard-core condition imposes substantial restrictions on the possible discrete-time realizations of a particular string operator $F$, but it's likely that there will still be multiple hard-core discrete-time string evolutions realizing $F$.  There are many ways this can occur.  For instance, given a configuration of multiple defects, we can imagine moving them in any order, including moving one defect, then moving a second one, and then moving the first one again.  Another class of examples can be derived from the observation that the string evolution $(e, e)$ realizes the identity string operator, but it is not a trivial evolution.  It or more complicated evolutions realizing the identity can be inserted into an evolution (provided they do not lead to a violation of the hard-core condition) without changing the overall string operator.  One can think of this as a non-monotonic reparameterization of the time coordinate.

\subsubsection{Discrete picture of evolution and configurations}

Based on what we have done so far, we can put together a completely discrete picture of paths and configurations starting at the toric code $\Ctoricfix$ and satisfying the hard-core condition.

Let
\begin{align}
\FdiscreteBeginend := & \{\text{Verified discrete-time string evolutions starting at defect locations $\locations$, ending} \nonumber \\
			     &   \text{at defect locations $\locationsend$, and satisfying the hard-core condition at all times}\}.
\end{align}
The set $\FdiscreteBeginend$ is non-empty only when $|S'_v| = |S_v|$, $|S'_f| = |S_f|$, and $\Ctoricfixend$ also satisfies the hard-core condition.  The set $\Fdiscretepaths$ consists of verified discrete-time string evolutions beginning at defect locations $\locations$ but which are allowed to end anywhere:
\begin{equation}
\Fdiscretepaths := \bigsqcup_{\locationsend} \FdiscreteBeginend,
\end{equation}
where the union is taken over all possible endpoints $\locationsend$.  $\Fdiscretepaths$ is a discrete-time analog of $\Fpath$ as introduced in Section~\ref{sec:geoFT}, except that $\Fpath$ allows for an arbitrary starting point of the paths it contains.  We can similarly define a discrete-time analog of $\Fsmall$:
\begin{equation}
\Fdiscrete := \{\text{verified string operators } F | F \text{ realized by } (a_1, \ldots, a_l) \in \Fdiscretepaths\}.
\end{equation}

Let $n_v = |S_v|$, $n_f = |S_f|$.  The corresponding discrete configuration space is 
\begin{equation}
\Ctorichardcore := \Fdiscrete (\Ctoricfix).
\end{equation}
Typically, $\Ctorichardcore$ is the set of all toric codes with $(n_v,n_f)$ defects such that the hard-core condition is satisfied, and is otherwise independent of $\Ctoricfix$.  However, in some unusual special cases, it may not be possible to move between otherwise valid configurations using only evolutions in $\Fdiscretepaths$.  For instance, this can happen if the lattice $\Gamma$ is completely packed with defects at the minimum distance $\sep$ from each other.

Superficially, this structure seems like precisely what we need in order to follow the program laid out in Section~\ref{sec:geoFT}.  However, with only a discrete configuration space, it is not straightforward to define a connection, let alone determine if that connection is projectively flat.  Therefore, in the following sections, we will discuss how to extend the discrete string evolutions and discrete configuration space into continuous objects for which the notion of a flat projective connection is well-defined.

\subsection{Continuous-time string evolution and the graph-like configuration space $\Mgraph$.}
\label{sec:continuous_toric}

\subsubsection{Continuous-time string evolutions}
\label{sec:contin-strings}
%

Let us denote by $\ver{\sigma}^i_j$ the verified Pauli operator $\ver{\sigma}^i$ acting on the $j^{\rm{th}}$ qubit of the system and acting as the identity everywhere else, where $i$ ranges over the set $\Paulis$.  To fit these single-qubit Pauli operators into the framework of differential geometry, we fix a continuous-time unitary evolution $U^i_j(t)$ for each single-qubit Pauli operator $\sigma^i_j$, that is, $U^i_j(0) = \idmatrix$, where $\idmatrix$ is the identity transformation, and $U^i_j(1) = \ver{\sigma}^i_j$.  In particular, we may make the following choice for $U^i_j(t)$: Let $H^i_j = (\ver{\sigma}^i_j - \idmatrix) \pi/2$.  Then
\begin{equation}
U^i_j(t) := e^{i t H^i_j} = e^{-i t \pi/2} [\cos (t \pi/2)\, \idmatrix + i \sin (t \pi/2)\, \ver{\sigma}^i_j]. \label{eq:UPauli}
\end{equation}
In order to capture the possibility of a defect moving ``backwards'' along an edge (with respect to edge orientation), we also define the opposite evolution given by
\begin{equation}
V^i_j(t) := U^i_j(-t) = e^{-i t H^i_j} = e^{i t \pi/2} [\cos (t \pi/2)\, \idmatrix - i \sin (t \pi/2)\, \ver{\sigma}^i_j]. \label{eq:VPauli}
\end{equation}
Note that $V^i_j(1) = \ver{\sigma}^i_j$ as well.  Let
\begin{align}
\alpha(t) &= e^{-i t \pi/2} \cos (t \pi/2) \\
\beta(t) &= i e^{-i t \pi/2} \sin (t \pi/2).
\end{align}
Then
\begin{align}
U^i_j (t) &= \alpha(t)\, \idmatrix + \beta(t)\, \ver{\sigma}^i_j \\
V^i_j (t) &= \alpha (1-t)\, \sigma^i_j + \beta(1-t)\, \idmatrix = U^i_j (1-t)\, \ver{\sigma}^i_j. \label{eqn:VijUij}
\end{align}
Note that $\alpha(0) = \beta(1) = 1$, $\alpha(1) = \beta(0) = 0$, and $|\alpha(t)|^2 + |\beta(t)|^2 = 1$.

\begin{Definition}
\label{defn:continuoustimestring}
For each primal edge of $\Gamma$ and each dual edge of $\bar{\Gamma}$ choose an orientation.  The orientations on $\Gamma$ and $\bar{\Gamma}$ may be chosen independently.  Assume we start with the fixed base code $\Ctoricfix$.  Let $F(t): [0,l] \rightarrow \U{N}$ be a continuous function ($l$ is a non-negative integer), and let $\Delta F (m, \delta t) = F(m+ \delta t) F(m)^{-1}$ be the change in $F(t)$ between time $t=m$ and time $t=m+\delta t$.  Given the choice of edge orientations, $F(t)$ is a \emph{continuous-time string evolution} running for $l$ time steps if it satisfies the following properties:
\begin{enumerate}
\item $F(0) = \idmatrix$.
\item The sequence $(\Delta F(0,1), \Delta F(1,1), \Delta F (2,1), \ldots,\Delta F(l-1,1))$ is a verified hard-core discrete-time string evolution (represented as a sequence of single-qubit verified Paulis). \label{item:hardcorediscrete}
\item Let $m$ be an integer in $[0,l-1]$.  Suppose $\Delta F(m,1) = \ver{\sigma}^i_j$, where $j$ is an edge oriented from vertex $a$ to vertex $b$ (if $i=Z$) or a dual edge oriented from face $a$ to face $b$ (if $i=X$).  If $F(m)(\Ctoricfix)$ has a (dual) defect at $a$, then we require that $\Delta F(m, \delta t) = U^i_j (\delta t)$ for $0 \leq \delta t < 1$.  If $F(m)(\Ctoricfix)$ has a (dual) defect at $b$, then we require that $\Delta F(m , \delta t)= V^i_j (\delta t)$.  If $F(m)(\Ctoricfix)$ has no (dual) defects at either $a$ or $b$, then we require $\Delta F(m, \delta t) = \idmatrix$.  (By property~(\ref{item:hardcorediscrete}), the evolution satisfies the hard-core condition and exactly one of these three conditions must be true.)  $j$ is the qubit on which $\Delta F(m,1)$ acts.
\end{enumerate}
A \emph{continuous string operator} is $F = F(t)$ for some $t$ and some continuous-time string evolution $F(t)$.  We can define a set $\Fgraph$:
\begin{equation}
\Fgraph := \{\text{Continuous string operators} \}.
\end{equation}
A \emph{truncated continuous-time string evolution} is a continuous-time string evolution taken on the domain $[t_1, t_2]$, with $0 \leq t_1 < t_2 \leq l$.  
\end{Definition}
Note that we use $U^i_j(t)$ from time step $m$ to $m+1$ if we are moving a defect along the orientation of an edge and use $V^i_j (t)$ if we are moving against the orientation of the edge.
The orientation of edges will also be needed in Section~\ref{sec:Mgraph} to define a continuous configuration space of defects.  At this stage, the orientations can be arbitrary, but later in Section~\ref{sec:toricM}, we shall insist on a regular pattern for the edge orientations.  Also note that Definition~\ref{defn:continuoustimestring} implies that the evolution $F(t)$ is piecewise smooth.

A continuous-time string evolution represents a string evolution where time can take continuous values.  At integer times, it agrees with the steps of a verified discrete-time string evolution, and at non-integer times, the next discrete step is in progress.  However, we have two possible ways to define the evolution between discrete times, depending on whether we are moving a defect along the orientation of an edge or against it.

Because we have fixed the way we interpolate between integer time steps, a continuous-time string evolution does not contain any more information than a discrete-time string evolution.  Indeed, given any discrete-time string evolution $(a_k)$, there is a unique continuous-time string evolution corresponding to it, given by
\begin{equation}
F(m + \delta t) = U^i_j(\delta t) \circ a_{m} \circ \cdots \circ a_1 \text{ or } V^i_j(\delta t) \circ a_{m} \circ \cdots \circ a_1
\end{equation}
for all $m \in [0,l-1]$, $0 \leq \delta < 1$, depending on the orientation of the edge $a_{m+1}$, and where $j$ is the qubit at the edge $a_{m+1}$, $i = Z$ if $a_{m+1}$ is primal, and $i=X$ if $a_{m+1}$ is dual.

We can now define the set of \emph{continuous-time fault-tolerant paths} $\Fgraphpaths$ to be the set of continuous-time string evolutions allowing for truncation, reparameterization, and reversal:
\begin{align}
\label{eqn:Fgraphpaths}
\Fgraphpaths := \{ F(t):[0,1] \rightarrow \U{N} \, | \, & F(g(t)) \text{ is a truncated piecewise smooth} \\
& \text{continuous-time string evolution for some $g(t)$} \}. \nonumber
\end{align}
The function $g(t)$ can be any possible reparameterization, including truncation and reversal, i.e. $g:[0,1] \rightarrow [0,l]$, $g$ piecewise smooth.  For the toric code, $\Fgraphpaths$ plays the role of $\Fpath$ as discussed in Section~\ref{sec:geoFT}.  Note that $\Fgraph = \{ F(1) | F(t) \in \Fgraphpaths\}$.

\subsubsection{The graph-like configuration space $\Mgraph$}
\label{sec:Mgraph}

We can define a set of codes $\Mgraph$ that can be reached from $\Ctoricfix$ via hard-core continuous-time fault-tolerant evolutions.  Since $\Fgraph$ contains the endpoints of evolutions in $\Fgraphpaths$, we have
\begin{equation}
\Mgraph = \Fgraph (\Ctoricfix) \subset \Grass.
\end{equation}

Since a continuous-time string evolution reduces to a verified discrete-time string evolution at integer times, we can certainly think of a continuous-time string evolution as moving defects around in just the same way as a discrete-time string evolution.  That is, the discrete set $\Ctorichardcore$ is contained in $\Mgraph$.  This only applies at integer times; at non-integer times, we don't have a toric code for any configuration of the defects, so $\Mgraph$ is strictly bigger than $\Ctorichardcore$.

Nevertheless, we can interpret the extra codes in $\Mgraph$ as a continuous graph-like configuration space which allows primal defects to be on either vertices or \emph{edges} of $\Gamma$ and dual defects to be on vertices or edges of the dual graph $\bar{\Gamma}$.  Specifically, we imagine a partially completed Pauli operator as moving a defect along the edge between the starting location and ending location of the defect.  

To be precise, suppose $U^Z_j(t)$ acts on a subspace $\Ctoricfix$ associated with the defect configuration $(S_v, S_f)$ such that exactly one end $a$ of the edge $j$ is in $S_v$, and let $\Ctoricfixend = \sigma^Z_j (\Ctoricfix)$ be the code reached by transporting one primal defect from $a$ to $b$ (the other end of $j$) via $\sigma^Z_j$.  Then we have $S_f' = S_f$ and $S_v' = (S_v \setminus \{a\}) \cup \{b\}$.  Then if the edge is oriented from $a$ to $b$, we define the subspace $\Ctorict{t}$ ($t \in [0,1]$) as
\begin{align}
\Ctorict{t} :=&\ U^Z_j(t) (\Ctoricfix) = (\alpha(t) \idmatrix + \beta(t) \ver{\sigma}^Z_j) (\Ctoricfix) \\
=&\  \{ \alpha(t) \ket{\psi} + \beta(t) \ver{\sigma}^Z_j \ket{\psi} \text{ s.t. } \ket{\psi} \in \Ctoricfix \}.
\end{align}
When $\ket{\psi} \in \Ctoricfix$, $\ver{\sigma}^Z_j \ket{\psi}$ is the corresponding codeword in $\Ctoricfixend$, so $\Ctorict{t}$ consists of superpositions of corresponding codewords in $\Ctoricfix$ and in $\Ctoricfixend$, and is thus ``partway in between'' these two codes.  If the edge is oriented from $b$ to $a$, we instead let
\begin{align}
\Ctorict{t} &:= V^Z_j(t) (\Ctoricfix) \\
&= U^Z_j(1-t) \ver{\sigma}^Z_j (\Ctoricfix) \\
&= U^Z_j(1-t) (\Ctoricfixend) \\
&= \{ \alpha(1-t) \ver{\sigma}^Z_j \ket{\psi} + \beta(1-t) \ket{\psi} \text{ s.t. } \ket{\psi} \in \Ctoricfix \}.
\end{align}
Here we have used the fact that $V^Z_j(t) = U^Z_j(1-t) \ver{\sigma}^Z_j$ (equation (\ref{eqn:VijUij})) and the relationship between $\Ctoricfix$ and $\Ctoricfixend$.

We interpret $\Ctorict{t}$ as associated with a configuration where there are dual defects at the locations $S_f$, primal defects at the locations $S_v \setminus \{a\}$, plus one more primal defect located on the edge between $a$ and $b$, a distance $t$ from $a$.  (Assume the length of each edge is normalized to $1$.)  Since $\Ctoricfix$ and $\Ctoricfixend$ are orthogonal subspaces and $\alpha(t)$, $\beta(t)$ are monotonic in $t$, the only intersection between the subspaces $\Ctorict{t}$ and $\Ctorict{t'}$ is the null vector whenever $t \neq t'$.  Thus, given a codeword in one of these subspaces, we can uniquely determine the value of $t$ and therefore the location of the moving defect in a corresponding configuration.  Conversely, given a location of a defect along the edge, we can immediately assign a value of $t$ and thus obtain a code $\Ctorict{t}$.  The only potential ambiguity is that the code corresponding to a particular defect location might in principle depend on whether the defect was moved to the location by starting from $a$ and moving towards $b$ or by moving from $b$ towards $a$, but by assigning orientation and using $U$ or $V$ as appropriate, we have ensured that the code $\Ctorict{t}$ is the same whether it is derived by moving a distance $t$ from $a$ or a distance $1-t$ from $b$.  Thus, we have found a smooth diffeomorphism between the subset $\{\Ctorict{t}\}_{t \in [0,1]}$ of $\Grass$ and the edge in $\Gamma$ corresponding to qubit $j$.

Similarly, if $U^X_j(t)$ maps the code $\Ctoricfix$ to $\Ctoricfixend$ by moving a single dual defect from face $a$ to face $b$ along an edge $j$ of the dual lattice $\bar{\Gamma}$, we can define
\begin{equation}
\Ctorict{t} := U^X_j(t) (\Ctoricfix) = \{ \alpha(t) \ket{\psi} + \beta(t) \ver{\sigma}^X_j \ket{\psi}  \text{ s.t. } \ket{\psi} \in \Ctoricfix \}
\end{equation}
or
\begin{equation}
\Ctorict{t} := V^X_j(t) (\Ctoricfix) =  \{ \alpha(1-t) \ver{\sigma}^X_j \ket{\psi} + \beta(1-t) \ket{\psi} \text{ s.t. }  \ket{\psi} \in \Ctoricfix \},
\end{equation}
depending on the orientation of the edge $j$ with respect to $a$ and $b$.  As with the primal defects, we may interpret $\Ctorict{t}$ as having a dual defect located on the edge of the dual lattice between $a$ and $b$, a distance $t$ from $a$.

Based on this, we can see that $\Mgraph$ is equivalent to the set of configurations such that either all primal defects are on vertices and dual defects are on faces, or all but one defect is on a vertex or face and the last defect is located on an edge of the lattice $\Gamma$ (if the defect is primal) or on an edge of the dual lattice $\bar{\Gamma}$ (if it is a dual defect).  Because we defined continuous-time string evolutions to only have a single incomplete Pauli operator, $\Mgraph$ never has more than one defect on the interior of an edge.

While $\Mgraph$ has a continuous structure (i.e., it is not discrete), it is still not a manifold.  When we imagine shifting a defect around on the interior of an edge, $\Mgraph$ looks locally like a $1$-manifold.  However, when all defects are on vertices of $\Gamma$ or $\bar{\Gamma}$, there are many choices for how to move the defects.  Any defect (subject to hard-core constraints) could move out onto any of the edges incident at its vertex, so at these points, $\Mgraph$ is a junction between many different $1$-manifolds.  In other words, $\Mgraph$ has the structure of a graph.  The maximum degree of the graph is $4(n_v + n_f)$ when $\Gamma$ is a square lattice, but for some vertices, the degree will be smaller because the hard-core condition restricts the possible moves.

We can define the usual pre-connection on the bundles over $\Mgraph$, and it turns out to be a well-defined partial connection.  We delay the proof that it is a connection until Section~\ref{sec:toric_flatFPP}, where we show that it is a strong FPP connection.

\subsection{The extended evolutions and the full configuration space $\M$}
\label{sec:toricM}

The next step is to figure out how to extend $\Mgraph$ to a larger topological space $\M$ so we can determine if the connection we get on $\M$ is flat.

$\Mgraph$ is already an extension of $\Ctorichardcore$, and was defined by allowing defects to be on the edges as well as on the vertices of $\Gamma$ and $\bar{\Gamma}$.  The next step is to relax the restrictions on the locations of defects even further, to let them be anywhere (subject to the hard-core condition) on the torus $T_g$.  Once again, we find a homeomorphism between a subset $\M$ of the Grassmannian and the hard-core configuration space of the defects.

For this purpose, we wish to assume that the lattice $\Gamma$ is a (square) rectilinear one and to choose orientations on the edges in a regular way so that all faces look the same. Figure~\ref{fig:uniformtorus} is an example of how to do this.  We believe neither of these assumptions is truly necessary, but they save us from having to consider many different cases.

\begin{figure}
\begin{tikzpicture}[thick,>=latex]

\draw [step=1cm] (-2,-2) grid (2,2);

\foreach \x in {-2,...,2}
	\foreach \y in {-2,...,1}
		\draw[->] (\x,\y) -- +(0,0.6);

\foreach \x in {-2,...,1}
	\foreach \y in {-2,...,2}
		\draw[->] (\x,\y) -- +(0.6,0);

\end{tikzpicture}
\caption{A way of choosing orientations on the torus so that all faces look the same.}
\label{fig:uniformtorus}
\end{figure}
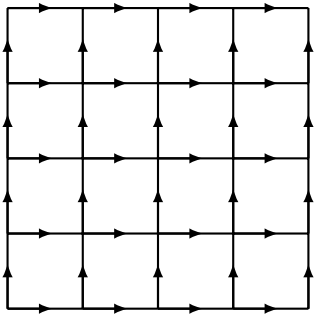

\subsubsection{Codes with one defect off a vertex}
\label{sec:codesonedefect}

Let us start by considering moving a single primal defect off of a vertex.  The case of a dual defect will be similar.  We define a map assigning a subspace in $\Grass$ to each location within a face of $\Gamma$.  Let the four vertices of the face be labelled $A$, $B$, $C$, and $D$, with the edges oriented $A \rightarrow B$, $C \rightarrow D$, $C \rightarrow A$, $D \rightarrow B$, as shown in Figure~\ref{fig:ABCD}.  Fix the locations of all but one defect, and consider four codes $C_A$, $C_B$, $C_C$, and $C_D$, which have the last defect at $A$, $B$, $C$, and $D$, respectively.  We assume that the defect configurations for all four codes satisfy the hard-core condition.  When the last defect is in the square $ABCD$, we can refer to the defect location with the coordinates $(x,y)$, with $(x,y) = (0,0)$ corresponding to $C$ and $(x,y) = (1,1)$ corresponding to $B$, etc.  We wish to define a homeomorphism between a subset of $\Grass$ and locations $(x,y) \in [0,1] \times [0,1]$ which is consistent with the one already chosen for $\Mgraph$.  For instance, a defect at position $(0,y)$ on the edge $CA$ should correspond to the code $\{\alpha(y) \ket{\psi_C} + \beta(y) \ket{\psi_A}\}$.  Here, and in the following, we let $\ket{\psi_A}$, $\ket{\psi_B}$, $\ket{\psi_C}$, and $\ket{\psi_D}$ be mutually corresponding codewords in the codes $C_A$, $C_B$, $C_C$, and $C_D$.  Namely, $\ket{\psi_B} = \ver{\sigma}^i_j \ket{\psi_A}$, where $\ver{\sigma}^i_j$ is the verified Pauli that moves a defect from $A$ to $B$, and so on.

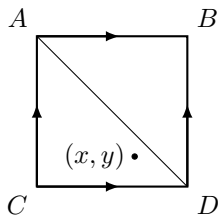
\begin{figure}
\begin{tikzpicture}[>=latex]
\draw [thick] (0,0) rectangle (2,2);
\draw (2,0) -- (0,2);

\draw[thick,->] (0,0) -- +(0,1.1);
\draw[thick,->] (2,0) -- +(0,1.1);
\draw[thick,->] (0,0) -- +(1.1,0);
\draw[thick,->] (0,2) -- +(1.1,0);

\node at (0,2) [anchor=south east] {$A$};
\node at (2,2) [anchor=south west] {$B$};
\node at (0,0) [anchor=north east] {$C$};
\node at (2,0) [anchor=north west] {$D$};

\draw (1.3,0.4) circle (1pt)[fill=black] node[anchor=east] {$(x,y)$};

\end{tikzpicture}
\caption{A single face with corners $A$, $B$, $C$, and $D$.}
\label{fig:ABCD}
\end{figure}

We can break the square face into two triangles via the line $AD$, as in figure~\ref{fig:ABCD}.  If the defect is below this line ($x + y \leq 1$), we will choose a code which is a superposition of $C_A$, $C_C$, and $C_D$.  
\begin{equation}
C_{(x,y)} = \{ a(x,y) \ket{\psi_A} + c(x,y) \ket{\psi_C} + d(x,y) \ket{\psi_D} \}.
\label{eq:ACD}
\end{equation}
If it is above the line ($x+y \geq 1$), we choose a code which is a superposition of $C_A$, $C_B$, and $C_D$.  
\begin{equation}
C_{(x,y)} = \{ a'(x,y) \ket{\psi_A} + b'(x,y) \ket{\psi_B} + d'(x,y) \ket{\psi_D} \}.
\label{eq:ABD}
\end{equation}
In order to agree with the codes along the edges, we have the following conditions:
\begin{align}
\text{Edge $CA$: } a(0,y) &= \beta(y) \\
c(0,y) &= \alpha(y) \\
d(0,y) &= 0 \\
\text{Edge $CD$: } a(x,0) &= 0 \\
c(x,0) &= \alpha(x) \\
d(x,0) &= \beta(x) \\
\text{Edge $AB$: }a'(x,1) &= \alpha(x) \\
b'(x,1) &= \beta(x) \\
d'(x,1) &= 0 \\
\text{Edge $DB$: } a'(1,y) &= 0 \\
b'(1,y) &= \beta(y) \\
d'(1,y) &= \alpha(y).
\end{align}
If the defect is on the line $AD$, we wish the rules obtained from the two triangles to agree:
\begin{align}
a(x,1-x) &= a'(x,1-x) \\
b'(x,1-x) &= 0 \\
c(x, 1-x) &= 0 \\
d(x, 1-x) &= d'(x,1-x).
\end{align}
We must also obey normalization:
\begin{align}
|a(x,y)|^2 + |c(x,y)|^2 + |d(x,y)|^2 &= 1 \\
|a'(x,y)|^2 + |b'(x,y)|^2 + |d'(x,y)|^2 &= 1.
\end{align}
Finally, we wish the mapping $(x,y) \mapsto (a,c,d; a', b', d')$ to be one-to-one so that each code corresponds to only one point in the face.

One solution is as follows:
\begin{align}
a(x,y) &=  \alpha(x) \beta(y) \\
c(x,y) &=   \alpha(x+y) \\
d(x,y) &=  \frac{\beta(x) \alpha(y)}{|\beta(x) \alpha(y)|} \sqrt{1 - |\alpha(x) \beta(y)|^2 - |\alpha(x+y)|^2} \\
d(0,y) &= 0 \\
\nonumber \\
a'(x,y) &=  \alpha(x) \beta(y) \\
b'(x,y) &=  \beta(x+y - 1) \\
d'(x,y) &=  \frac{\beta(x) \alpha(y)}{|\beta(x) \alpha(y)|} \sqrt{1 - |\alpha(x) \beta(y)|^2 - |\beta(x+y-1)|^2} \\
d'(x,1) &= 0
\end{align}
These formulas satisfy the constraints because of the properties of $\alpha$ and $\beta$ noted in Section~\ref{sec:contin-strings}.  Note that $d$ and $d'$ are continuous at $x=0$ and $y=1$, respectively.  Not all values of $(a,c,d; a',b',d')$ can be achieved for $(x,y)$ within the square.

\begin{Proposition}
There is a bijection between the points $(x,y)$ within the square $ABCD$ and the achievable sextuplets $(a,c,d; a',b',d')$.
\end{Proposition}

\begin{proof}
By the definitions above, the values of $(x,y)$ completely specify $(a,c,d; a',b',d')$, so we have a map $\iota_\Box: ABCD \rightarrow \Grass$.  We wish to show that this map is one-to-one.  

Let us consider a possible triplet $(a,c,d)$ in the triangle $ACD$ and prove that there is a unique $(x,y)$ that it corresponds to.  In fact, it is sufficient to look only at $|a|$ and $|c|$ to determine $(x,y)$.  $|\alpha|$ is monotonically decreasing and $|\beta|$ is monotonically increasing.  Therefore, the value of $|c|$ tells us the value of $x+y$ provided that $x+y \leq 1$.  Moreover, for fixed $x+y$, increasing $x$ and decreasing $y$ decreases $|a|$, so there is only a single point on the line $x+y = |\alpha|^{-1}(c)$ that can correspond to a particular value of $|a|$.

Similarly, a triplet $(a',b',d')$ in the triangle $ABD$ corresponds to a unique pair $(x,y)$ with $x+y \geq 1$, since $|b'|$ tells us the value of $x+y -1$ and $|a'|$ specifies a unique $(x,y)$ on that line.
\end{proof}

Thus, we have established an isomorphism between the possible locations of a defect within the square $ABCD$ and a subset of the Grassmannian (the image of $\iota_\Box$).  Furthermore, this isomorphism and its inverse are continuous.  Indeed, the map $\iota_\Box$ is a smooth diffeomorphism within the triangle $ACD$ and inside the triangle $ABD$.  However, there is a discontinuity in the derivatives across the line $AD$.  Therefore, the isomorphism is only \emph{piecewise smooth}.

Via the same strategy, we can define codes for which a single dual defect is located anywhere within a face of the dual graph $\bar{\Gamma}$.

\subsubsection{Unitary to move a single defect into a square}

We now wish to define unitaries that can move primal and dual defects from a vertex into the interior of a face of $\Gamma$ or $\bar{\Gamma}$ respectively.  We will describe in some detail the appropriate unitaries to move a primal defect, and the definition for a dual defect is very similar.

Let $A$ be the upper left corner of a face $ABCD$ of $\Gamma$ as in Figure~\ref{fig:ABCD}, and let $C_A$, $C_B$, $C_C$, and $C_D$, as above, be codes that have a primal defect on $A$, $B$, $C$, or $D$, respectively, and all other defects at the same locations in all four codes.  We restrict attention to situations in which all four codes satisfy the hard core condition.  We will define a unitary $U_{A,(x,y)}$ that maps $C_A$ to $C_{(x,y)}$, i.e., that moves a defect from $A$ to $(x,y)$, where $(x,y)$ is an interior point of the square $ABCD$.  

There are many possible such unitaries, but we wish to choose one that acts only on qubits in the vicinity of the face $ABCD$.  In particular, it will act non-trivially only on the $12$ qubits on edges adjacent to the four vertices $ABCD$. 

Furthermore, the unitary $U_{A,(x,y)}$ shall act non-trivially only on the subspaces $C_A \oplus C_B \oplus C_C \oplus C_D$ spanned by states with exactly one primal defect on one of the vertices $A$, $B$, $C$, or $D$.  There are many such subspaces corresponding to different configurations of defects outside the square, and each such subspace $C_A \oplus C_B \oplus C_C \oplus C_D$ will be an invariant subspace of $U_{A,(x,y)}$.   If there are no defects on these vertices or more than one (which would violate the hard-core condition), $U_{A,(x,y)}$ acts as the identity.  
Within a single subspace $C_A \oplus C_B \oplus C_C \oplus C_D$, we demand that $U_{A,(x,y)}$ maps each state $\ket{\psi} \in C_A$ to the corresponding state $\ket{\psi}\in C_{(x,y)}$, as given in Section~\ref{sec:codesonedefect}.  We also insist that if $(x,y)$ is on an edge of the square, then $U_{A,(x,y)}$ agrees with a continuous string operator that moves the defect from $A$ to that position on the edge.  (There are two such continuous string operators acting on the edges of the square, corresponding to moving around the face in opposite directions, but they agree on the subspace on which $U_{A,(x,y)}$ acts non-trivially because by the hard-core condition there is no dual defect located at that face.)  Finally, we ask that the map $(x,y) \mapsto U_{A,(x,y)}$ be a piecewise smooth diffeomorphism.\footnote{It should be a smooth diffeomorphism on each triangle, but there will be a discontinuity in the derivatives on the line $AD$.}  We have specified the action of $U_{A,(x,y)}$ on a subset of basis vectors (those from $C_A$) within each invariant subspace $C_A \oplus C_B \oplus C_C \oplus C_D$ and completely defined it when $(x,y)$ is on the edge of the square.  We don't care otherwise about the action of $U_{A,(x,y)}$, so pick some $U_{A,(x,y)}$ satisfying these constraints.

\begin{Proposition}
\label{prop:contunitary}
A family of unitaries $U_{A,(x,y)}$ with the properties listed in the previous paragraph exists, and the unitaries can be performed by acting on just the $12$ qubits at edges adjacent to the square ABCD.
\end{Proposition}

\begin{proof}
Let us focus on a single invariant subspace.  In particular, let us focus on a four-dimensional subspace $C_{\ket{\psi}}$ spanned by $\ket{\psi_A} = \ket{\psi}$, $\ket{\psi_B} = \ver{\sigma}_B \ket{\psi}$, $\ket{\psi_C} = \ver{\sigma}_C \ket{\psi}$, and $\ket{\psi_D} = \ver{\sigma}_D \ket{\psi}$, where $\ket{\psi} \in C_A$ and $\ver{\sigma}_B$, $\ver{\sigma}_C$, and $\ver{\sigma}_D$ are verified Pauli operators mapping $C_A$ to $C_B$, $C_C$, and $C_D$, respectively.  We will choose a $U_{A,(x,y)}$ such that each $C_{\ket{\psi}}$ is an invariant subspace.  

Section~\ref{sec:codesonedefect} specifies the value of $\ket{\psi_{(x,y)}} = U_{A,(x,y)} \ket{\psi_A}$ for all $(x,y)$ in the square $ABCD$.  There is no obstruction to extending $\ket{\psi_{(x,y)}}$ to a basis for $C_{\ket{\psi}}$ in a piecewise smooth way over the whole square $ABCD$.  This is equivalent to choosing a $4 \times 4$ unitary $\tilde{U}_{A,(x,y)}$ such that $\tilde{U}_{A,(x,y)} \ket{\psi_A} = \ket{\psi_{(x,y)}}$ $\forall (x,y) \in ABCD$.  In fact, because we are free to choose the overall phase, we can ensure that $\tilde{U}_{A,(x,y)} \in \SU{4}$.  However, $\tilde{U}_{A,(x,y)}$ is not necessarily the $U_{A,(x,y)}$ we desire because we also have a requirement that specifies $U_{A,(x,y)}$ completely when $(x,y)$ is on the boundary of the square $ABCD$.  

Note that choosing $U_{A,(x,y)}$ is equivalent to choosing $\hat{U}: ABCD \rightarrow \U{3}$:
\begin{equation}
\hat{U}(x,y) = \tilde{U}_{A,(x,y)}^{-1} U_{A,(x,y)} |_{\hat{C}},
\end{equation}
where $\hat{C}$ is the subspace of $C_{\ket{\psi}}$ spanned by $\ket{\psi_B}$, $\ket{\psi_C}$, and $\ket{\psi_D}$.  We chose $\tilde{U}_{A,(x,y)}$ to be piecewise smooth as a function of $(x,y)$. Therefore, $U_{A,(x,y)}$ is piecewise smooth as $(x,y)$ varies over the square iff $\hat{U}$ is piecewise smooth.  Since $U_{A,(x,y)}$ is specified on the boundary of the square, so is $\hat{U}$.  We can think of this specification as a loop $\gamma$ in $\U{3}$.  We claim that $\gamma$ is contractible.  A piecewise smooth homotopy between $\gamma$ and the trivial loop is equivalent to a choice of $\hat{U}$ that has the correct value on the boundary of ABCD, and thus determines $U_{A,(x,y)}$ on the subspace $C_{\ket{\psi}}$.

Therefore, we just need to show that $\gamma$ is contractible.  $\U{3} = \U{1} \times \SU{3}$, with $\U{1}$ the center of $\U{3}$, and the projection from $\U{3}$ to $\U{1}$ is given by the determinant.  Since $\pi_1(\SU{3})=0$, we just need to show that $\det \gamma$ is contractible in $\U{1}$.  We can check this by calculating $\det U_{A,(x,y)}$ around the loop, since $\det \tilde{U}_{A,(x,y)} = 1$.  Now, on the subspace $C_{\ket{\psi}}$, $U^i_j(t)$ and $V^i_j(t)$ act non-trivially only on the 2-dimensional subspace with a defect at one vertex of the edge for qubit $j$, so $\det U^i_j(t) = \det V^i_j (1-t) = e^{-2it\pi/2}$.  Referring to Figure~\ref{fig:ABCD}, as we move from $A$ to $C$ along the edge $CA$, we move the phase around half the unit circle, but then we move it back (since the edge orientation is opposite) as we move from $C$ to $D$ along the edge $CD$.  Then we go halfway around the unit circle again going from $D$ to $B$ and then reverse ourselves again as we return from $B$ to $A$.  We see that the phase does not loop. Therefore, $\det \gamma$ and $\gamma$ are contractible.  Note that this argument relied on our choice of edge orientation.

We thus get a valid definition of $\hat{U}$ on the square ABCD and can deduce a definition of $U_{A,(x,y)}|_{\hat{C}} := \tilde{U}_{A,(x,y)} \hat{U}(x,y)$.  We can use the same definition for $U_{A,(x,y)}$ on any subspace of the form $C_{\ket{\psi}}$.  We can let $\ket{\psi}$ run over a basis of $C_A$, giving us a definition of $U_{A,(x,y)}$ satisfying all constraints on the whole invariant subspace $C_A \oplus C_B \oplus C_C \oplus C_D$.  

To show that there exists a unitary $U_{A,(x,y)}$ which can be realized by acting only on the $12$ qubits adjacent to the square ABCD, we can apply the following lemma:
\begin{Lemma}
Let $U$ be a unitary operator with an invariant subspace $S$.  There exists $U'$ such that $U'|_S = U|_S$ such that $U'$ can be realized by acting only on the qubits in set $T$ if $\tr_T U \ket{\psi}\bra{\psi} U^\dagger = \tr_T \ket{\psi}\bra{\psi}$ for all $\ket{\psi} \in S$.
\end{Lemma}

\begin{proof}[Proof of Lemma.]

Let $R$ be reference Hilbert space of the same dimension as $S$, and let $\ket{\Phi} = \sum_{k} \ket{k}_S \ket{k}_R$ be a maximally entangled state between $S$ and $R$.  Then $\tr_{T} (U \otimes \idmatrix) \ket{\Phi} \bra{\Phi} (U^\dagger \otimes \idmatrix) = \tr_{T} \ket{\Phi} \bra{\Phi}$.  It follows (from standard arguments, see e.g. \cite{mayers}) that there exists unitary $U'$ acting only the qubits in $T$ such that
\begin{equation}
(U' \otimes \idmatrix)  \ket{\Phi} = (U \otimes \idmatrix) \ket{\Phi}.
\end{equation}
Since $U'$ and $U$ have the same Choi-Jamiolkowski isomorphism on the subspace $S$, they act the same way on $S$.
\end{proof}

We thus need only to show that all states in one invariant subspace $C_{\ket{\psi}}$ have the same density matrix on qubits not adjacent to the square $ABCD$.  Let $\tr_\#$ indicate the trace over the $12$ qubits adjacent to the square $ABCD$.  

Let $\Pi_A = \frac{1}{2} (\idmatrix - A_A)$ be the projector on the subspace with a defect at $A$ and similarly let $\Pi_B$, $\Pi_C$, and $\Pi_D$ be the projectors on the subspaces with defects at $B$, $C$, and $D$.  Note that these four projectors only act on the $12$ qubits adjacent to $ABCD$.  

Consider $\ket{\psi_i}$ and $\ket{\psi_j}$, with $i \neq j$ drawn from $\{A,B,C,D\}$.  Then
\begin{align}
\tr_{\#} \ket{\psi_i} \bra{\psi_j} &= \tr_{\#} (\Pi_i \ket{\psi_i} \bra{\psi_j}) \\
&= \tr_{\#} (\ket{\psi_i} \bra{\psi_j} \Pi_i) \\
&= 0,
\end{align}
with the second line following by cyclicity of the trace.

Furthermore, 
\begin{equation}
\tr_\# (\ket{\psi_A} \bra{\psi_A}) = \tr_\# (\ver{\sigma}_B \ket{\psi_A} \bra{\psi_A} \ver{\sigma}_B) = \tr_\# (\ver{\sigma}_C \ket{\psi_A} \bra{\psi_A} \ver{\sigma}_C) = \tr_\# (\ver{\sigma}_D \ket{\psi_A} \bra{\psi_A} \ver{\sigma}_D)
\end{equation}
by cyclicity of the trace.  Thus, given any normalized state 
\begin{equation}
\ket{\phi} = \lambda_A \ket{\psi_A} + \lambda_B \ket{\psi_B} + \lambda_C \ket{\psi_C} + \lambda_D \ket{\psi_D} \in C_{\ket{\psi}},
\end{equation}
we have that
\begin{align}
\tr_{\#} \ket{\phi} \bra{\phi} &= \sum_{i \in \{A,B,C,D\}} |\lambda_i|^2 \tr_{\#} \ket{\psi_i} \bra{\psi_i} + \sum_{i \neq j} \lambda_i \lambda_j^* \tr_{\#} \ket{\psi_i} \bra{\psi_j} \\
&= \Big( \sum_{i \in \{A,B,C,D\}} |\lambda_i|^2 \Big) \tr_{\#} \ket{\psi} \bra{\psi} + \sum_{i \neq j} 0 \\
&= \tr_{\#} \ket{\psi} \bra{\psi}.
\end{align}
(Recall that $\ket{\psi_i} = \ver{\sigma}_i \ket{\psi}$.)  Since the density matrix of all states in $C_{\ket{\psi}}$ is the same outside of the $12$ qubits adjacent to $ABCD$, any unitary on $C_{\ket{\psi}}$ can be implemented by acting on just the qubits adjacent to $ABCD$.  It follows that the unitaries $U_{A,(x,y)}$ can be implemented by acting on just the qubits adjacent to $ABCD$.
\end{proof}

Superficially, the definition only says what $U_{A,(x,y)}$ does to the specific code $C_A$, and not to codes that have a defect at $A$ but with a different configuration of the far away defects.  However, because of the requirement that $U_{A,(x,y)}$ acts only on the qubits adjacent to $ABCD$, $U_{A,(x,y)}$ has the same effect (i.e., it moves a defect from $A$ to $(x,y)$) for all codes satisfying the hard core condition and with a defect at $A$.  To see this, consider such a code $C'_A$.  Because $C'_A$ and $C_A$ only differ in the location (and possibly number) of defects away from $ABCD$, we can reach $C'_A$ from $C_A$ by applying some string operator $F$ that does not involve any of the qubits adjacent to $ABCD$.  $F$ and $U_{A,(x,y)}$ act on disjoint sets of qubits, and therefore they commute.

Let $C'_B$, $C'_C$, $C'_D$ be the codes that have all defects in the same locations as $C'_A$ except for having one at $B$, $C$, or $D$ instead of $A$, and $C'_{(x,y)}$ be the code derived from them which corresponds to having a defect at location $(x,y)$ in the square $ABCD$.  Assume all of the codes $C'_A$, $C'_B$, $C'_C$, and $C'_D$ satisfy the hard-core condition.  Just as $C'_A = F C_A$, it is also true that $C'_B = F C_B$, $C'_C = F C_C$, and $C'_D = F C_D$.  Therefore, $C'_{(x,y)} = F C_{(x,y)}$ as well, since these two codes are defined in equations (\ref{eq:ACD}) and (\ref{eq:ABD})  via linear superpositions of the other codes.  Thus,
\begin{equation}
C'_{(x,y)} = F \left( C_{(x,y)} \right) = F U_{A,(x,y)} \left( C_A \right) = U_{A,(x,y)} F \left( C_A \right) = U_{A,(x,y)} \left( C'_A \right).
\end{equation}
That is, $U_{A,(x,y)}$ moves a defect from $A$ to $(x,y)$ in the square $ABCD$ no matter what the locations are of other defects (primal or dual) far away.

In the same way, we can define a unitary $V_{\bar{A},(x,y)}$ that moves a dual defect from a vertex of the dual graph $\bar{\Gamma}$ to anywhere on the dual face for which $\bar{A}$ is in the upper left corner.  Once again, we insist that $V_{\bar{A},(x,y)}$ acts only on the qubits adjacent to the dual face, and therefore $V_{\bar{A},(x,y)}$ has the same effect on all codes satisfying the hard-core condition that have a dual defect at $\bar{A}$.

Note that it is possible that a code $C_A$ satisfies the hard-core condition when a defect is at $A$, but it would violate the hard-core condition if the defect is moved from $A$ to one of the other vertices $B$, $C$, or $D$.  In this case, we will not use the unitary $U_{A,(x,y)}$ since it might have an unexpected effect and we will not allow a defect to be moved into the interior of a face for which one of the vertices would violate the hard-core condition.  It is still possible to move it along an edge of that face using a continuous string operator, provided both vertices adjacent to the edge satisfy the hard-core condition.  Similarly, we do not allow dual defects to be moved into the interior of a dual face if having a dual defect at one of the dual vertices of the dual face would violate the hard-core condition.

We have only defined the unitaries $U_{A,(x,y)}$ and $V_{\bar{A},(x,y)}$ to move defects from the upper left corner of a face or dual face into the interior.  However, it is straightforward to move a defect from any other corner of a face or dual face into the interior by first moving the defect to the upper left corner via a verified discrete string operator and then using $U_{A,(x,y)}$ or $V_{\bar{A},(x,y)}$ to move it into the interior.

\subsubsection{The set $\Fext$ of unitaries and the set $\Fextpaths$ of unitary evolutions}

The next step is to define a set of unitaries and unitary evolutions that allows us to move \emph{multiple} defects off of the vertices of the graph.  To do so, we will need to be careful that moving multiple defects off of the graph does not involve codes that violate the hard-core condition.

\begin{Definition}
Let $S \subseteq \Ctoricall$ be a set of toric codes, let $S_A \subseteq S$ be the subset of codes in $S$ that have a primal defect at the vertex $A$, and let $S_{\bar{A}} \subseteq S$ be the subset of codes that have a dual defect at the dual vertex (i.e., primal face) $\bar{A}$.  Let $f$ be the face for which $A$ is the upper left corner.  Then define 
\begin{equation}
S_f^A = S \cup \{C | C = F C_A, C_A \in S_A, F \text{ a verified discrete $Z$ string operator acting on edges of $f$.} \}.  
\end{equation}
Similarly, let $\bar{f}$ be the dual face for which $\bar{A}$ is the upper left corner.
\begin{equation}
S_{\bar{f}}^{\bar{A}}  = S \cup \{C | C = F C_{\bar{A}}, C_{\bar{A}} \in S_{\bar{A}}, F \text{ a verified discrete $X$ string operator acting on edges of $\bar{f}$.} \}.  
\end{equation}
We can also define $S_e^A$ and $S_{\bar{e}}^{\bar{A}}$ analogously with $F$ a verified discrete string operator (i.e. a verified Pauli for a single qubit) and $e$ or $\bar{e}$ is an edge or dual edge for which $A$ is the starting point.
\end{Definition}

The set $S_f^A$ could add up to $3$ new codes to $S$ for each code $C_A \in S_A$ and similarly for $S_{\bar{f}}^{\bar{A}}$.  $S_e^A$ could add $0$ or $1$ new code to $S$ for each code $C_A \in S_A$.
$S_f^A$ is the set of codes that we can get by starting with a code in $S$ and moving a defect from $A$ to one of the other corners of $f$.  The definition makes the most sense for codes which satisfy the hard-core condition, since then we don't have to worry about annihilating two defects, but for the moment, we let the definition apply to all sets of codes.  Similarly, $S_e^A$ includes the codes we get from $S$ by moving the defect along the edge $e$.  The sets $S_{\bar{f}}^{\bar{A}}$ and $S_{\bar{e}}^{\bar{A}}$ are the versions of these sets for dual defects.

\begin{Definition}
\label{def:WFdecomp}
Suppose we have a unitary $W$ of the form $W = \left(\prod_{i=1}^r W_i \right) F$, where $F$ is a verified discrete string operator, and each $W_i$ is $U_{A,(x,y)}$, $V_{\bar{A},(x,y)}$, or $U^P_a(t)$ (with $a$ a qubit representing either an edge or dual edge, depending on the Pauli $P$; see equation (\ref{eq:UPauli}) for the definition of $U^P_a$, $0 \leq t \leq 1$).  We can then define a sequence of sets of codes $S_i$ for $(W,\Ctoricfix)$ via the following process:
\begin{enumerate}
\item Let $\Ctoricfixend = F \Ctoricfix$, and let $S_0 = \{ \Ctoricfixend \}$.
\item Let $S_i$ be defined recursively from $S=S_{i-1}$ as one of the following:
\begin{itemize}
\item If $W_i = U_{A,(x,y)}$, then $S_i = (S_{i-1})_f^A$, where $f$ is the face with upper left corner $A$.
\item If $W_i = V_{\bar{A},(x,y)}$, then $S_i = (S_{i-1})_{\bar{f}}^{\bar{A}}$, where $\bar{f}$ is the dual face with upper left corner $\bar{A}$.
\item If $W_i = U^Z_a(t)$ and $A$ is the start of the edge $a$, then $S_i = (S_{i-1})_a^A$.
\item If $W_i = U^X_a(t)$ and $\bar{A}$ is the start of the dual edge $a$, then $S_i = (S_{i-1})_a^{\bar{A}}$.
\end{itemize}
\end{enumerate}
If a vertex $A$ or dual vertex $\bar{A}$ appears while defining $S_i$, we say it is the \emph{starting point} of $W_i$.
\end{Definition}
Because we always use $W_i$ to move a defect from the starting point of an edge, we do not need to use $V^i_a(t)$ in the definition.  Note that given $W$ of this form, the decomposition into $W_i$ and $F$ satisfying the above definition is unique up to permutation of the $W_i$'s; we will show this as part of Lemma~\ref{lemma:Wunique}.

Recall that when one defect is off of a vertex, we get a state which is a superposition of states from multiple toric codes with slightly different defect locations.  The same will hold true when we move multiple defects off of the vertices of $\Gamma$.  The set $S_r$ (with $r$ the number of $W_i$'s in the decomposition of $W$) is the set of codes that appear in the superposition, and $S_i$ is the set of codes that would appear if we only used $W_1, \ldots, W_i$ instead of $W_1, \ldots, W_r$.  Note that $S_i \subseteq S_{i+1}$.

The extended set $\Fext$ of allowed unitaries acting on $\Ctoricfix$ then consists of unitaries of the form $\prod_{i=1}^r W_i F$, as above, with the additional constraints that 
\begin{itemize}
\item $F$ can be realized via a hard-core verified discrete-time string evolution from $\Ctoricfix$.
\item All of the vertices $A$ or dual vertices $\bar{A}$ that appear as starting points of a $W_i$ are locations of defects in $\Ctoricfixend = F \Ctoricfix$.
\item Each vertex $A$ or dual vertex $\bar{A}$ appears only once, i.e., as a starting point of only one $W_i$.
\item The set $S_r$ contains only codes satisfying the hard-core condition.
\end{itemize}

This definition may seem complicated, but it is not actually that bad.  All we are saying is that $\Fext$ consists of unitaries that can be created via a verified discrete-time string evolution followed by moving some or all of the defects onto edges or into the interior of faces.  We can then create any allowed configuration by starting with an approximation to it where all defects are on integer lattice sites, the upper left corners of the destination faces, or the starting vertices of the destination edges (this can be created from $\Ctoricfix$ using $F$), and then applying an appropriate set of $W_i$'s.

We can also define a set of allowed evolutions
\begin{equation}
\Fextpaths := \{ F(t) | \text{$F(t)$ piecewise smooth}, F(t) \in \Fext \ \forall t \in [0,1] \}.
\label{eqn:Fextpaths}
\end{equation}

The set of unitaries $\Fext$ is the version of $\Fbig$ from Section~\ref{sec:geoFT} for this example, and $\Fextpaths$ is the analog of $\Fbigpaths$.

\subsubsection{The full configuration space $\M$ of codes}

We now can define the set of codes which corresponds to allowing multiple defects to leave the vertices of the graph:
\begin{equation}
\M = \Fext (\Ctoricfix).
\label{eqn:toricMdef}
\end{equation}
The various sets of codes and paths are summarized in Figure~\ref{fig:setsummary}.

\begin{figure}
\begin{tikzpicture}[text height=1.5ex, text depth=.25ex]
\tikzset{
actson/.style = {->, to path={++(0,0.75) ++(225:0.5cm) arc (225:-45:0.5cm)}}
}

\node (C) at (0,0) {$\Ctorichardcore$};
\node at (1.25,0) {$\subset$};
\node (Mgraph) at (2,0) {$\Mgraph$};
\node at (3,0) {$\subset$};
\node (M) at (4,0) {$\M$};
\node at (4.8,0) {$\subset$};
\node (Grass) at (6,0) {$\Grass$};

\draw [->] (C) ++(0,0.75) ++(225:0.5cm) arc (225:-45:0.5cm);
\node at (0,1.6) {$\Fdiscretepaths$};
\node at (0,2.3) {$\Fdiscrete$};

\draw [->] (Mgraph) ++(0,0.75) ++(225:0.5cm) arc (225:-45:0.5cm);
\node at (2,1.6) {$\Fgraphpaths$};
\node at (2,2.3) {$\Fgraph$};

\draw [->] (M) ++(0,0.75) ++(225:0.5cm) arc (225:-45:0.5cm);
\node at (4,1.6) {$\Fextpaths$};
\node at (4,2.3) {$\Fext$};

\node at (-3,0) {Subsets of $\Grass$:};
\node at (-3,1.6) {Paths in $\U{N}$:};
\node at (-3,2.3) {Subsets of $\U{N}$:};

\end{tikzpicture}
\caption{We have three configuration spaces satisfying the hard core condition: $\Ctorichardcore$ (which corresponds to having defects only on vertices of the graph or dual graph), $\Mgraph$ (which corresponds to having defects only on edges), and $\M$ (allowing defects also in the interior of faces).  All three correspond to subsets of $\Grass$.  They have sets of operator evolutions acting on them: $\Fdiscretepaths$, $\Fgraphpaths$, and $\Fextpaths$, respectively, which are different versions of the sets of allowed fault-tolerant operations.  In the notation of Section~\ref{sec:geoFT}, $\Fgraph = \Fsmall$ and $\Fext = \Fbig$.}
\label{fig:setsummary}
\end{figure}
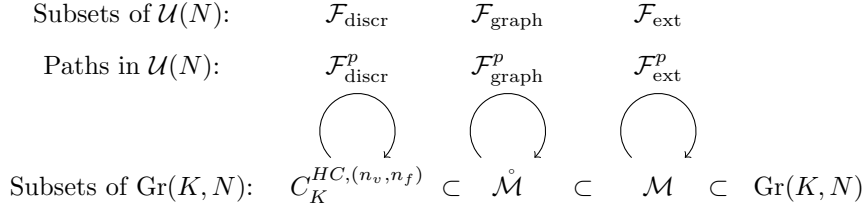

We claim that $\M$ is isomorphic to the configuration space of the defects (subject to the hard-core condition) on $T_g$.

First, let us be explicit about what we mean by the hard-core condition when the defects can be anywhere on the surface:
\begin{Definition}
Let $\Gamma$ be a lattice on $T_g$ and $\bar{\Gamma}$ be its dual lattice, also interpreted as a lattice on $T_g$.  Let $x$ be a point on $T_g$.  We say that a vertex $v$ of $\Gamma$ is \emph{adjacent} to $x$ if one of the following holds:
\begin{itemize}
\item If $x$ is in the interior of a face $f$ of $\Gamma$ and $v$ is a vertex of the face $f$,
\item If $x$ is on an edge $e$ of $\Gamma$ and $v$ is a vertex on the edge $e$, or
\item If $x$ is located at the vertex $v$.
\end{itemize}
Similarly, we say that a dual vertex $\bar{v}$ of $\bar{\Gamma}$ is \emph{adjacent} to $x$ if $\bar{v}$ is on dual face $\bar{f}$ when $x$ in the interior of $\bar{f}$, $\bar{v}$ is an endpoint of a dual edge $\bar{e}$ when $x$ is on $\bar{e}$, or if $x$ is on $\bar{v}$ itself.

A pair of defects at points $x$ and $y$ on $T_g$ satisfies the \emph{hard-core condition} with respect to the lattice $\Gamma$ and the dual lattice $\bar{\Gamma}$ if every vertex (if the defect at $x$ is primal) or dual vertex (if the defect at $x$ is dual) adjacent to $x$ has lattice distance at least $s$ from every vertex or dual vertex adjacent to $y$.  (Recall that we count distance between a vertex $v$ and a dual vertex $\bar{v}$ by counting the minimal distance between $v$ and any vertex on the face of $\Gamma$ given by $\bar{v}$.)  A configuration of many defects satisfies the hard-core condition if all pairs of defects satisfy the hard-core condition.
\end{Definition}

\begin{Theorem}
$\M$ is isomorphic (via piecewise smooth diffeomorphisms) to the configuration space of defects subject to the hard-core condition and labelled by defect type (primal or dual defect).
\end{Theorem}

In case there are so many defects that not all defect configurations can be moved to each other without violating the hard-core condition at some point (i.e., the hard-core configuration space is not connected), we focus attention on the component of the configuration space path-connected to the code $\Ctoricfix$.

\begin{proof}
We will demonstrate a continuous piecewise smooth bijection between $\M$ and the defect configuration space.  Fix a reference code $\Ctoricfix$ as before.

There is a straightforward map $\Phi$ from the defect hard-core configuration space to $\M$ based our definitions so far: Let us describe the configuration by assigning to the $i$th defect a vertex $A_i$ of $\Gamma$ (if the defect is a primal defect) or a vertex $\bar{A}_i$ of $\bar{\Gamma}$ (if the defect is a dual defect), such that the vertex (or dual vertex) is the upper left corner of the face (or dual face) containing the defect, the start of the oriented edge containing the defect, or the location of the defect, depending on if the defect is in the interior of a face $a_i$, on an edge $e_i$, or on a vertex of $\Gamma$ (or $\bar{\Gamma}$ if it is a dual defect).  Then we assign coordinates $(x_i,y_i)$ within the face if the defect is in the interior of a face, or we assign $t_i$ if the defect is along an edge.  (We do not need any additional coordinate if the defect is on a vertex.)  Begin with the code $\Ctoricfix$ and apply a verified string operator $F$ to bring it to $\Ctoricfixend$, the code with defects in the locations $\{A_i\} \cup \{\bar{A}_i\}$.  Then, for each defect $i$ not on a vertex, apply to the code $U_{A_i,(x_i,y_i)}$, $V_{A_i,(x_i,y_i)}$, $U^Z_{e_i} (t_i)$, or $U^X_{e_i} (t_i)$, as appropriate (for primal defect in the interior of a face, dual defect in the interior of a dual face, primal defect on an edge, and dual defect on an edge, respectively).  The result is a code reached from $\Ctoricfix$ via an operator in $\Fext$.  Therefore the resulting code is in $\M$ by equation~(\ref{eqn:toricMdef}).  Because the four maps used to move defects from vertices onto faces or edges are continuous as a function of the coordinates $(x_i, y_i)$ or $t_i$, and because they are chosen to be consistent with each other, $\Phi$ is continuous.  
Further, note that as defects move within individual triangles, $\Phi$ is smooth.  Whenever a defect crosses an edge, there is a discontinuity in derivatives, so the map as a whole is only piecewise smooth.

We can also define a map $\Xi$ from $\M$ to the defect hard-core configuration space.  Note that for any element $C$ of $\M$, we can pick an element $W$ of $\Fext$, decomposed as $\prod_{i=1}^r W_i F$, such that $(\prod_{i=1}^r W_i F) (\Ctoricfix) = C$.  From the additional constraints in the definition of $\Fext$, we can assign to each defect in $\Ctoricfixend = F \Ctoricfix$ an ``end point'' in the following ways:  If the vertex or dual vertex containing the $i$th defect is the starting point for some $W_i$, then the end point of that defect is the location specified by the coordinates used in $W_i$.  If the vertex or dual vertex does not appear as the starting point of a $W_i$, then the end point of the defect is the (dual) vertex itself.  This assigns to each code $C \in \M$ a configuration of defects.  We need to verify that the configuration is unique (to ensure that the map $\Xi$ is well-defined) and does not depend on the choice of $W$ or its decomposition; this is a consequence of the following lemma:

\begin{Lemma}
\label{lemma:Wunique}
If $W, W' \in \Fext$ and $W (\Ctoricfix) = W' (\Ctoricfix)$, then $W = W' G$, with $G$ a verified discrete-time string operator.  Furthermore, $G$ is a closed loop.
\end{Lemma}

\begin{proof}
We write $W = (\prod_{i=1}^r W_i) F$ and $W' = (\prod_{j=1}^{r'} W_j') F'$, with $W_i$ and $F$ as in Definition~\ref{def:WFdecomp}.  We can assume without loss of generality that $r=r'$ by adding additional trivial $W_i$ or $W_j'$ operators to the smaller decomposition.  Unless the set of starting points for the set $\{W_i\}$ match the starting points for $\{W_j'\}$, the defects in $W (\Ctoricfix)$ will be in different locations than those for $W' (\Ctoricfix)$, and the codes will be different.  Furthermore, for a given starting point, there is a unique $W_i$ that moves a defect from that starting point to a particular location in the interior of an adjacent edge or face.  Thus, the set $\{W_i\}$ must be the same as $\{W_j'\}$.  The order may be different, but due to the hard-core condition, the $W$'s commute, so $\prod_{i=1}^r W_i = \prod_{j=1}^r W_j'$.  Thus, $W = W' (F')^{-1} F$.  $F$ and $F'$ are both verified discrete-time string operators, so $G := (F')^{-1} F$ is a verified discrete-time string operator.

Since all the starting points of the $\{W_i\}$ match the starting points of the $\{W_i'\}$, the locations of all defects in $F(\Ctoricfix)$ and $F'(\Ctoricfix)$ must be the same.  Thus, the end points of the string $F$ must be the same as the end points of $F'$.  Therefore the product $(F')^{-1} F$ is a loop.
\end{proof}

That is, given $C \in \M$, the choice of $W \in \Fext$ is unique up to a verified discrete-time string operator $G$ which is a closed loop.  A closed loop acting on a code in $\Ctorichardcore$ does not move any defects; therefore, the map $\Xi$ we have given above is well-defined.  By comparing the resulting map from $\M$ to the defect configuration space with the definition of $\Phi$ above, we can see that the map $\Xi$ is equal to $\Phi^{-1}$.  As with $\Phi$, $\Xi$ is piecewise smooth.

\end{proof}

There are two awkward things about the topological space $\M$.  First, it is not generally a manifold, or even a manifold with boundary.  Second, it is not smooth, but only piecewise smooth.

It is not a manifold because of the hard-core condition.  When there are $n_v$ primal defects and $n_f$ dual defects, the set of defect configurations satisfying the hard-core condition is, in most places, locally diffeomorphic to a $2(n_v + n_f)$-dimensional Euclidean space: When all defects are far from each other, each is unconstrained and has $2$ directions it can move.  When a pair of defects gets close to each other, there is a boundary that prevents defects moving closer, and a cluster of close-together defects can constrain movement even more.  Because we measure distance between defects in terms of lattice distance, which only takes integer values, a cluster of three or more defects may be constrained to a space of even lower dimension than $2(n_v + n_f)-1$.

$\M$ is not smooth because it is defined by being pieced together from configurations where each defect is constrained to be inside a triangle.  We have been careful to make everything continuous when a defect crosses into a different triangle, but the structure is not smooth at such points.  If viewed using the usual differentiable structure of $\Grass$, there are kinks in $\M$ when a defect crosses between triangles, either across a diagonal or across an edge of $\Gamma$ or $\bar{\Gamma}$.

The failure of $\M$ to be a manifold is not that serious for us, although it does make it unclear precisely what structure we should demand for other examples of fault-tolerant gate sets.  The failure to be smooth is more troublesome, since in order to define a connection, we need a well-defined tangent space for $\M$ and the bundles $\toricvecbundext$ and $\toricPbundext$. 

Luckily, $\M$ is sufficiently well-behaved that we \emph{can} define a tangent space over $\M$ and the two bundles of interest.  All of these spaces are piecewise smooth, and in fact can be decomposed as a union of smooth manifolds with boundary (or corners), and the only intersection between different pieces is at the boundary of the pieces.  If we consider a point in the interior of one of the pieces, the tangent space is defined normally.  If we are at the boundary of just a single piece, we consider the tangent space to just consist of those tangent vectors which point inwards, in directions which are still within the piece. 

If we are at a point which is the intersection of two or more pieces, we can define the tangent space as the disjoint union of the tangent spaces of the pieces meeting at the point (truncated as above due to boundaries).  Effectively we are gluing together the pieces of the tangent spaces in the same way that the pieces of $\M$ are attached together.  In this case, every direction in $\M$ (or in a bundle over $\M$) has tangent vectors associated to it, as usual.  An equivalent definition would be to let the tangent space at any point consist of sets of curves, where two curves are equivalent if they have the same directional derivative \emph{in the positive direction} at the point.  This is very close to a standard definition of the tangent space for a differentiable manifold; the difference being that in this case, the derivative in the negative direction might not agree with the derivative in the positive direction, or the negative direction might not exist at all if the curve begins at the point.  The tangent space defined in this way is sufficient for our purposes (defining a connection), but it lacks some standard properties of tangent spaces, so one should be cautious when applying this definition in other contexts.  For instance, the tangent space defined this way does not have a natural linear structure.

\subsection{Strong flat partial connections over $\Mgraph \subset \M$}
\label{sec:toric_flatFPP}

As in Section \ref{sec:proposal}, we can define a natural pre-connection $\conngraphtoric$ on $\toricPbundgraph$ for paths in $\Mgraph$.\footnote{The ``K'' in $\conngraphtoric$ stands for ``Kitaev code''.} Let $U(t) \in \Fgraphpaths$ be a path of continuous string operators inducing the path $\gamma(t)$ in $\Mgraph$.  Then pre-parallel transport along $\gamma$ in the bundle $\toricPbundgraph$ at time $t$ is given by $\beta \mapsto h_{U(t)} (\beta)$, with $h_{U(t)}$ defined by equation (\ref{eq:h_U}).  Recall that different paths $U(t)$ corresponding to the same path $\gamma(t)$ could in principle give different results from the pre-parallel transport.  We claim that instead the pre-connection defines a flat partial connection (not just projectively flat) on $\Mgraph$.  We can show this by extending it to a flat connection on all of $\M$, making $\conngraphtoric$ a \emph{strong} flat partial connection.  (See Definition~\ref{def:FPP}.)  The same argument and same conclusion applies for $\toricvecbundgraph$.

\subsubsection{Extending $\conngraphtoric$ to $\connexttoric$ on $\M$}

To extend the pre-connection $\conngraphtoric$ to a pre-connection $\connexttoric$ on $\M$, all we have to do is to replace $\Fgraphpaths$ with the set $\Fextpaths$ of allowed unitary paths for the toric code.  We prove in this section the following theorem:
\begin{Theorem}
$\connexttoric$ gives a well-defined notion of parallel transport on $\M$.
\end{Theorem}

\begin{proof}
Suppose we have a path $\gamma(t)$ in $\M$.  By the definition of $\M$, $\gamma(t) = W(t) (\Ctoricfix)$ for some $W(t) \in \Fextpaths$; in particular, for each $t$, $W(t) \in \Fext$.  We wish to choose $W(t) \in \Fextpaths$.  Lemma~\ref{lemma:Wunique} tells us that the only ambiguity in choosing $W(t)$ for any given $t$ is up to verified discrete-time string loops.  Therefore, for any path $\gamma(t)$, there is a natural choice for a \emph{continuous} unitary path $W(t)$: Start with some value of $W(0)$ as the initial value of the path, to be determined later.  As long as defects move within a face or edge, we simply vary the appropriate $W_i$ in the decomposition of $W(t)$ and leave the discrete string operator unchanged; the resulting segment of the unitary path is continuous.  

When a defect moves from one face $f$ to an edge $e$, we need to be careful, as the starting point for the edge might be different from the starting point for the face.  Let $A$ be the upper left corner of the face $f$, and let $H$ be a continuous string operator that moves a defect from $A$ along the edges of $f$ to the new location of the defect on $e$.  The starting point $A'$ of the edge $e$ might be different than $A$.  Write $H = U^P_a(s) E$ for an appropriate time $s$, Pauli $P = X$ or $Z$, and verified discrete-time string operator $E$ from $A$ and $A'$ along the edges of the face.  ($E$ is trivial if $A=A'$.)  We must update the decomposition of $W(t) = \prod W_i F$ when the defect reaches the edge.  We replace $W_i$ with $U^P_a(s)$ and replace $F$ by $EF$.  Because $U_{A,(x,y)}$ is required to agree with a continuous string operator that moves a defect along the edge(s) from $A$ to $(x,y)$ when $(x,y)$ lies on an edge of the face, $U_{A,(x,y)} \rightarrow H = U^P_a(s) E$ as $(x,y)$ reaches the edge.  Therefore, the change in the decomposition of $W(t)$ as the defect reaches the edge still leads to a continuous unitary path $W(t)$.

We change the decomposition similarly if we move from an edge into the interior of a face, or if we move on to or off of a vertex, or on to or off of dual edges or vertices.  This process tells us uniquely what the verified discrete-time string operators for $W(t)$ must be at all times in order for the unitary path to be continuous and how to represent $W(t)$ in the decomposition of Definition~\ref{def:WFdecomp}.  

The only remaining ambiguity is a choice of the verified discrete-time string operator for the starting point $W(0)$ of the unitary path.  If we fix $\gamma(0) = \Ctoricfix$, then it is natural to choose $W(0) = \idmatrix$.  Then $W(t)$ is unique, and therefore parallel transport along $\gamma(t)$ is well-defined.  If we use a path with a different starting point $\gamma(0)$, we must choose some verified discrete-time string operator $F$ for $W(0)$ such that $W(0)(\Ctoricfix) = \gamma(0)$.  (Lemma~\ref{lemma:Wunique} tells us this choice is the only ambiguity given $\gamma(0)$.)  Now imagine we choose a different path $W'(t) \in \Fextpaths$ for $\gamma(t)$ with verified discrete-time string operator $F'$ at the start.  $W'(t)$ is the unique path with $F'$, and $W'(t) = W(t) F^{-1} F'$.  Parallel transport along $\gamma(t)$ determined by $W(t)$ is given by $W(t) W(0)^{-1}$, but $W'(t) W'(0)^{-1} = W(t) W(0)^{-1}$, so the two unitary paths give the same parallel transport.  Thus, $\connexttoric$ gives a well-defined notion of parallel transport everywhere on $\M$, independent of even the choice of initial value $W(0)$.
\end{proof}

The reason we find a connection in this case and not a projective connection is that while we allow a global phase factor in discrete string operators, it is a discrete phase $\pm 1$ or $\pm i$.  Thus, there is no continuous unitary path that changes the global phase while not braiding the defects around each other.

\subsubsection{Flatness}

\begin{Theorem}
\label{thm:toricflat}
$\connexttoric$ is a flat connection, so $\conngraphtoric$ is a strong flat partial connection.
\end{Theorem}

\begin{proof}
Consider any loop $\gamma(t)$ in $\M$.  Let us for the moment restrict attention to loops $\gamma(t)$ that begin and end on points in $\Ctorichardcore$ (with all defects on vertices).  Because $\M$ is homeomorphic to the configuration space of the defects, we can think of a path $\gamma(t)$ in $\M$ as a colored braid over $T_g$ (with two colors representing primal and dual defects).  We can distort $\gamma(t)$ to a loop $\gamma'(t)$ that lies on $\Mgraph$, as in Figure~\ref{fig:distortpath}, by taking each strand of the braid and pushing (and possibly stretching or compressing) it via isotopy to travel along the edges of whichever face of $\Gamma$ or $\bar{\Gamma}$ within which the defect is currently moving.  When we do so, parallel transport along $\gamma'(t)$ is the same as parallel transport along $\gamma(t)$, broken up into segments moving within a single face, and with the segments interspersed by loops $\eta_i(t)$ that move a single defect within a single face (including its edges and vertices) of $\Gamma$ or $\bar{\Gamma}$.

\begin{figure}
\begin{tikzpicture}[scale=3]

\draw [step=1cm,very thin,gray] (0,1) grid (5,5);

\coordinate (start) at (1,2);

\draw [red,ultra thick] (start) to[out=45,in=180] ++(0.7,0.4) -- ++(1.6,0) to[out=0,in=240] ++(0.7,0.6) to[out=60,in=270] ++(0.2,0.4) to[out=90,in=350] ++(-0.2,0.3) to[out=170,in=330] ++(-1,0.3);
\node at (3.55,2.7) {$\gamma(t)$};

\draw [blue,ultra thick] (start) -- ++(3,0) -- ++(0,2) -- ++(-1,0);
\node at (4.25,2.3) {$\gamma'(t)$};

\draw [green,ultra thick] (start) ++(0.15,0.05) to[out=45,in=180] ++(0.65,0.3) -- ++(0.15,0) -- ++(0,-0.3) -- cycle;
\draw [green,ultra thick] (start) ++(1.05,0.05) -- ++(0,0.3) -- ++(0.9,0) -- ++(0,-0.3) -- cycle;
\draw [green,ultra thick] (start) ++(2.05,0.05) -- ++(0,0.3) -- ++(0.3,0) to[out=0,in=240] ++(0.60,0.45) -- ++(0,-0.75) -- cycle;
\draw [green,ultra thick] (start) ++(3,1) ++(0.05,0.15) to[out=60,in=270] ++(0.1,0.25) to[out=90,in=350] ++(-0.1,0.23) -- cycle;
\draw [green,ultra thick] (start) ++(3,1) ++(-0.05,0.75) to[out=170,in=330] ++(-0.73,0.2) -- ++(0.73,0) -- cycle;

\draw (start) ++(0.28,0.12) to[out=45,in=180] ++(0.49,0.19)  -- ++(0.14,0) -- ++(0,-0.23) -- ++(-0.79,0) -- ++(0,-0.14) -- ++(0.97,0) -- ++(0,0.37) -- ++(0.82,0) -- ++(0,-0.22) -- ++(-0.78,0) -- ++(0,-0.15) -- ++(0.96,0) -- ++(0,0.37) -- ++(0.3,0) to[out=0,in=240] ++(0.52,0.37) -- ++(0,-0.59) -- ++(-0.78,0) -- ++(0,-0.15) -- ++(0.91,0) -- ++(0,1.07) -- ++(0,0.21) to[out=60,in=270] ++(0.07,0.19) to[out=90,in=300] ++(-0.03,0.16) -- ++(0,-0.43) -- ++(-0.17,0) -- ++(0,0.66) to[out=170,in=340] ++(-0.58,0.12) -- ++(0.58,0) -- ++(0,-0.08) -- ++(0.13,0) -- ++(0,0.20) -- ++(-1,0);
\draw [->] (start) ++(0.91,0.31) -- ++(0,-0.115);
\draw [->] (start) ++(0.91,0.08) -- ++(-0.395,0);
\draw [->] (start) ++(0.12,-0.06) -- ++(0.485,0);
\draw [->] (start) ++(1.09,0.31) -- ++(0.41,0);
\draw [->] (start) ++(1.91,0.09) -- ++(-0.41,0);
\draw [->] (start) ++(1.13,-0.06) -- ++(0.37,0);
\draw [->] (start) ++(2.09,0.31) -- ++(0.3,0);
\draw [->] (start) ++(2.91,0.09) -- ++(-0.3,0);
\draw [->] (start) ++(2.13,-0.06) -- ++(0.26,0);
\draw [->] (start) ++(3.04,-0.06) -- ++(0,0.7);
\draw [->] (start) ++(3.04,1.21) to[out=60,in=270] ++(0.07,0.19);
\draw [->] (start) ++(3.08,1.56) -- ++(0,-0.12);
\draw [->] (start) ++(2.91,1.14) -- ++(0,0.43);
\draw [->] (start) ++(2.33,1.92) -- ++(0.29,0);
\draw [->] (start) ++(3.04,2.04) -- ++(-0.5,0);

\node at (1.7,2.2) {$\eta_1(t)$};
\node at (2.5,2.2) {$\eta_2(t)$};
\node at (3.5,2.2) {$\eta_3(t)$};

\end{tikzpicture}
\caption{Distorting the path $\gamma(t)$ of a single particle (red) to a path $\gamma'(t)$ (blue) along edges of $\Gamma$. The difference between the two is a product of loops $\eta_i(t)$ (green), each contained in a single face.  A path (in black) consisting of alternating green loops and blue segments is equivalent (because of cancelling segments) to the original red path $\gamma(t)$.}
\label{fig:distortpath}
\end{figure}
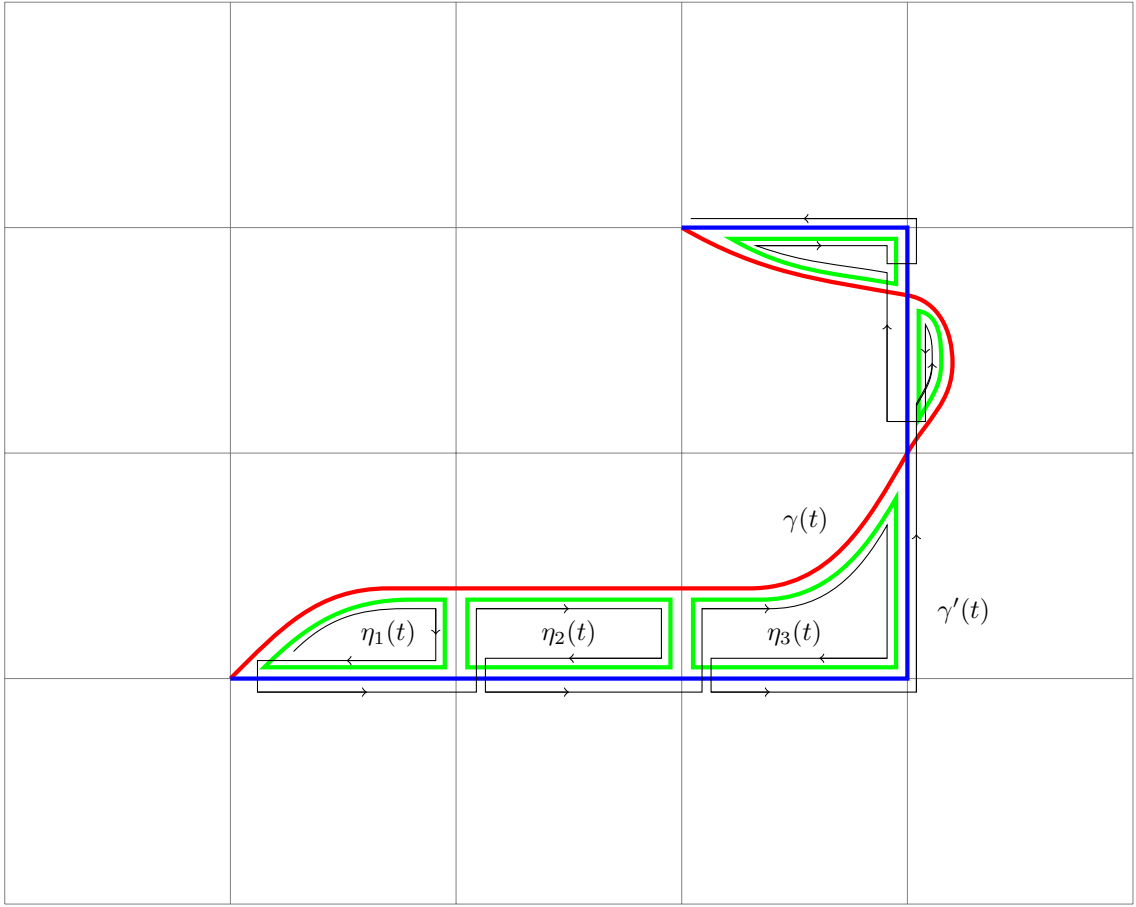

Now, we claim that the parallel transport along a loop $\eta_i(t)$ that lies within a single face of $\Gamma$ or $\bar{\Gamma}$ is trivial.  This is true because the unitary path $G_i(t)$ corresponding to $\eta_i(t)$ simply consists of varying $U_{A,(x,y)}$ over points $(x,y)$.  There is a unique $U_{A,(x,y)}$ for each point $(x,y)$ and $\eta_i(t)$ is a loop, so $G_i(0) = G_i(1)$ and parallel transport along $\eta_i(t)$ is the identity. Consequently, parallel transport along $\gamma(t)$ is the same as parallel transport just along the distorted path $\gamma'(t)$.

If $\gamma(t)$ is a homotopically trivial loop in $\M$, then the corresponding braid is trivial as well.  That is, homotopy in $\M$ is equivalent to isotopy in $\mathbb{R}^3$.  In a trivial braid, no defects circle any others and no defects travel around a non-trivial loop on $T_g$.  The $\gamma'(t)$ that we get by distorting such a $\gamma(t)$ will therefore be a product of disjoint loops on $\Mgraph$, and the unitary path corresponding to $\gamma'$ will be a product of disjoint discrete string loop operators that do not circle any defects or loop around $T_g$.  Therefore, due to standard properties of the toric code, parallel transport along $\gamma'(t)$, and thus along $\gamma(t)$, is trivial.

Now we need to consider loops based at arbitrary points in $\M$.  When we change the base point of a loop, we conjugate the holonomy with the parallel transport along a path between the base points.  However, conjugating the identity still gives us the identity.  Thus, the holonomy of any homotopically trivial (contractible) loop in $\M$ is the identity.  By Criterion~\ref{prop:flat_condition}, $\conngraphtoric$ is a well-defined connection and is flat on all of $\M$.
\end{proof}

\subsection{Fault-tolerant logical gates from string operators}
\label{sec:toricmonodromy}

By Theorem~\ref{thm:toricflat} and the observation of Section~\ref{sec:FTmonodromy}, the logical gates for the toric code are given by the monodromy representation of the fundamental group of  $\M$ induced by the bundle $\toricPbundgraph$ or $\toricvecbundgraph$.  The fundamental group of $\M$ is the braid group on $T_g$ with two colors of strands.  We have three types of generators of the colored braid group on $T_g$: those consisting of moving one defect of one color around another of a different color, those consisting of swapping two defects of the same type, and those consisting of moving a defect around a homotopically non-trivial loop of $T_g$.  

We can determine what logical gates the colored braid group generators correspond to using the standard approach to analyzing the toric code (as per \cite{Kitaev1} or \cite{Dennis}):  Swapping two defects of the same type (``half braiding'') does not change the encoded state unless the combination of the two paths forms a loop enclosing another defect.  Moving one defect around another (``full braiding'') produces a global phase of $-1$ iff the two defects are of different types (one primal and one dual).  Therefore full-braiding generators do not produce any non-trivial logical gate.  On the other hand, moving a defect around a homotopically non-trivial loop of $T_g$ performs a logical Pauli operator on the appropriate logical qubit of the code.  If the code $\Ctoricfix$ contains just one type of defect (e.g., primal), we can only perform logical Pauli operators of one type (e.g., only logical $X$'s).  If $\Ctoricfix$ contains both primal and dual defects, the monodromy group is isomorphic to the full Pauli group on the logical qubits.  Adding extra defects does not help us do more logical gates.

While in the case of the toric code, the monodromy group is quite simple, recall that the toric code is just an example and a model for more complicated topological codes.  In Section~\ref{sec:otheranyons}, we show how to apply the same construction to models with non-Abelian anyons, in which moving defects around each other will perform non-trivial logical gates, even potentially a universal set of gates.

\subsection{Freedom in choosing $\Mgraph$ and $\M$}
\label{sec:toricfreedom}

We made a number of arbitrary choices in the definitions of $\Mgraph$ and $\M$.  To what extent do our conclusions depend on these choices?  The set of logical gates cannot depend on these choices, so this is a question of how the topology of different possible choices of the topological space $\M$ relate to each other.

The freedom in choosing $\Mgraph$ is somewhat limited.  We need a way of performing continuous-time string operators with local gates that generate the usual discrete string operators at integer times.  There is certainly some freedom in precisely how to do that, for instance by reparameterizing time, but all of them are basically quite similar.

$\M$ is much less constrained.  The main thing we need is simply to find some way to extend $\Mgraph$ to a topological space $\M$ for which the extended connection $\connexttoric$ on $\M$ is projectively flat.  No doubt there are many different ways of doing this.  We have chosen one in this paper for which $\M$ is isomorphic to the configuration space of the defects, so that loops in $\M$ correspond to braids.  This allows us to make close contact with the usual topological picture of the toric code, so it is a natural choice.  Even then, there are a number of arbitrary details in our construction, but they don't seem to play any fundamentally important role.  We simply made a set of choices that works for the construction, but surely other choices are possible.

It is also not necessary to define an $\M$ that is isomorphic to the configuration space provided we get a flat projective connection.  However, not all choices are equally enlightening as to the structure of the fault-tolerant protocol.  For instance, suppose instead of the actual $\M$ we defined above, we were instead to make similar definitions, but to restrict defects to be within a small distance $\epsilon$ of an edge or a dual edge rather than allowing them to be anywhere in a face.  Topologically, this would correspond to a ``thickening'' of $\Mgraph$.  This alternate $\M'$ is a subset of the $\M$ we defined, so the connection on it is still well-defined and flat.  However, in $\M'$ there are many more homotopically non-trivial loops, since moving a defect around a face is now non-contractible.  These additional non-contractible loops always produce trivial logical gates, so are not very interesting, but they complicate matters.  In particular, they make the fundamental group of $\M'$ larger and therefore obscure the ``meaningful'' non-contractible loops that correspond to non-trivial logical gates.  The moral here is that we want to pick the ``largest'' $\M$ that gives us a flat projective connection.  Beyond that, it is unclear if there are additional criteria that $\M$ should satisfy.

\subsection{Other topological codes}
\label{sec:otheranyons}

The construction given for toric codes can be immediately applied to give a similar result for any anyon model with similar properties.  In particular, we assume that the anyon model has the following properties:
\begin{enumerate}
\item The model works within a larger Hilbert space which has $N$ qudits located on the edges of a square lattice $\Gamma$ on the torus.%
\footnote{The use of a square lattice is not essential, but it is convenient to use the same lattice as previously.}
The edges of $\Gamma$ are oriented as in Figure~\ref{fig:uniformtorus}.
\item The anyon model contains multiple species of quasiparticles (also called defects).  Some species (the ``primal species'') can be located on vertices of $\Gamma$, others (the ``dual species'') on the faces of $\Gamma$ (or equivalently, on the vertices of $\bar{\Gamma}$).  There is a vacuum species $0$, which is also the absence of a quasiparticle at a location.  Each configuration of defects corresponds to a subspace (the ``code subspaces'') whose dimension $K$ depends only on the number of each species of quasiparticles and not on the locations of the quasiparticles.
\item There exist projectors $\Pi^i_a$, with $i$ labelling a quasiparticle species and $a$ either a vertex (for a primal species) or dual vertex (for a dual species).  The projectors satisfy the following properties:
	\begin{enumerate}
	\item $\Pi^i_a$ acts only on qudits with lattice separation at most $\lfloor (\sep-1)/2 \rfloor$ from $a$.  (Recall that $\sep$ is the minimum separation allowed by the hard-core condition.)
	\item $\Pi^i_a \ket{\psi} = \ket{\psi}$ if $\ket{\psi}$ has a defect of type $i$ at $a$; otherwise $\Pi^i_a \ket{\psi} = 0$.
	\item $[\Pi^i_a, \Pi^i_b] = 0$.
	\end{enumerate}
\item There exist unitary operators $\omega^i_j$, with $j$ labelling a qudit and $i$ a quasiparticle species, which have the following properties:
	\begin{enumerate}
	\item $\omega^i_j$ acts only on qudits with lattice separation at most $\lfloor (\sep-1)/2 \rfloor$ from $j$.
	\item $(\omega^i_j)^2 = \idmatrix$
	\item $\omega^i_j$ moves a defect of species $i$ along the edge $j$: That is, if $C$ is a code subspace satisfying the hard-core condition with a defect of species $i$ at a vertex $v$ (or dual vertex if $i$ is a dual species) adjacent to qubit $j$, then $\omega^i_j (C)$ is a code subspace with the location of all defects the same as $C$ except that instead of a defect of type $i$ at $v$, there is a defect of type $i$ at the vertex (or dual vertex) $v'$, which is at the other end of the edge (or dual edge) associated with qubit $j$.
	\item $\omega^i_j$ acts as the identity if the two vertices (or dual vertices) adjacent to $j$ have no defects of type $i$ or if there are defects (of any species) at both vertices.
	\end{enumerate}
\item Let ABCD be a square (or dual square) in $\Gamma$.  Then $\omega^i_{CA} \omega^i_{CD} \omega^i_{BD} \omega^i_{AB} |_{C_A} = \idmatrix |_{C_A}$ when $C_A$ is a code subspace satisfying the hard-core condition with a defect of species $i$ at vertex (or dual vertex) $A$.
\end{enumerate}

The projectors $\Pi^i_a$ can detect the presence of a defect at $i$, so condition (iii) is essentially a refinement of condition (ii), and ensures that the quasiparticles are well-defined localized objects.  Note that if $C$ satisfies the hard-core condition, when there is a defect of species $i$ at $v$, there cannot be any defect at the adjacent vertex $v'$.  The interpretation of $\omega^i_j$ is that it moves a defect of type $i$ from $v$ to $v'$ or from $v'$ to $v$ (since $\omega^i_j$ squares to the identity).  We do not put a constraint on the behavior of $\omega^i_j$ when the defect configuration does not satisfy the hard-core condition or if there is not a particle of species $i$ on a vertex (or dual vertex) adjacent to qubit $j$.  Condition (v) basically encapsulates the topological order property of an anyon model: When the hard-core condition is satisfied, there are no other defects on or within the square ABCD, so moving an anyon around the square should be a trivial operation.

We believe that Kitaev's quantum double models \cite{Kitaev1} and the more general Levin-Wen string net models \cite{LW} can be formulated in this way.  The projectors $\Pi^i_a$ are the usual projectors in these models giving the location and type of a quasiparticle.  However, the standard string operators $F^i_j$ for these models are non-unitary, so they cannot serve as the operators $\omega^i_j$.  Instead, we use the combination
\begin{equation}
\omega^i_j = \Pi^0_{v} \Pi^i_{v'} F^i_j \Pi^i_v \Pi^0_{v'} + \Pi^i_{v} \Pi^0_{v'} F^i_{-j} \Pi^0_v \Pi^i_{v'}  + [(\idmatrix - \Pi^i_v) (\idmatrix - \Pi^i_{v'}) + (\idmatrix - \Pi^0_v)  \Pi^i_{v'} + \Pi^i_{v}(\idmatrix - \Pi^0_{v'}) - \Pi^i_v \Pi^i_{v'}] \idmatrix,
\end{equation}
where $v$ and $v'$ are the vertices (or dual vertices) at the ends of the edge associated with qubit $j$, and $F^i_{-j}$ is the string operator that moves a defect against the orientation of $j$.
This is analogous to equation~(\ref{eq:verifiedPauli}).  We believe that the properties of the string operators ensure that the $\omega^i_j$ defined this way satisfy condition (iv), but we have not been able to prove this rigorously.

Another minor complication is that these models are not usually presented on a square lattice.  We can always imbed the usual structure into a square lattice, which may require increasing $\sep$ somewhat since the distance between vertices of the imbedded graph could be greater than $1$.

Given a model with the desired properties, we can just follow the full construction for the toric codes presented in this section, substituting $\omega^i_j$ for $\ver{\sigma}^i_j$ throughout.  There are only a few small differences in the proofs.  First, the various configuration spaces require a specification of the number and locations of all types of defects, and there are generally more than two.  This produces no substantive difference in the construction.  A slightly larger change is that $\omega^i_j$ acts on qudits within a distance $\lfloor (\sep-1)/2 \rfloor$ of the edge instead of within distance $1$.  This means that the $U^i_j (t)$, $V^i_j(t)$ will also, and that the $U_{A,(x,y)}$ will act on qudits within a distance $\lfloor (\sep-1)/2 \rfloor$ of the square.  However, since the hard-core condition requires defects to be a distance at least $\sep$ apart, the operators associated with two different defects will not overlap.

The only place we used a property of a Pauli operator not part of condition (iv) is in the proof of  Proposition~\ref{prop:contunitary}.  Specifically, we need to calculate the determinant of $U^i_j(t) = e^{-i t \pi/2} [\cos (t \pi/2)\, \idmatrix + i \sin (t \pi/2)\, \omega^i_j]$ on the $4$-dimensional subspace $C_{\ket{\psi}}$.  $\omega^i_j$ acts as the identity on the two basis states which have no defects at the ends of $j$ and swaps the other two basis states.  Thus, on $C_{\ket{\psi}}$, $\omega^i_j = X \oplus \idmatrix_2$, and $\det U^i_j(t) = e^{-2it\pi/2}$ as for the toric code.

Thus, the same construction we used for the toric code also shows that a wide variety of models with non-Abelian anyons can also be understood in terms of our framework.  In these models, the monodromy group includes the full set of gates that can be produced via braiding.

\section{Towards full fault tolerance}
\label{sec:fullFT}

Sections \ref{sec:transversal} and \ref{sec:topologicalFT} illustrated how our picture works for two of the main classes of fault-tolerant gates, namely transversal gates and topological gates.  Here, we focus mostly on extending the transversal gate set to enable universal quantum computation; in particular, we comment on some aspects of the bundle construction involved.  Since transversal gates are not universal, fault-tolerant protocols based on them usually add in gates from a third major class of fault-tolerant gates, magic state constructions or equivalently magic state injection.  Topological constructions sometimes also make use of magic state injection.  This third class of fault-tolerant gates is not unitary.

Generally speaking, if we can perform the logical Clifford gates and one logical gate outside of the Clifford group, then the system is capable of universal computation. Magic state injection enables us to simplify the requirement of being able to perform one logical gate outside of the Clifford group to a requirement of being able to prepare a particular non-stabilizer state known as a ``magic state''. For instance, to perform the gate $R_{\pi/8} = \mathrm{diag}(e^{-i\pi/8}, e^{+i\pi/8})$, we just need the magic state $R_{\pi/8} \ket{+}$.  The magic state injection procedure, depicted in Figure~\ref{fig:magicstate} for $R_{\pi/8}$, typically uses some logical gates from the Clifford group, measurement, and classical control, plus some copies of the magic state.  In order to fulfill these requirements, the idea is that we can provide fault-tolerant measurement and logical Clifford group gates using transversal gates or another fault-tolerant construction, and somehow separately create encoded magic states, perhaps using state distillation techniques~\cite{BK}.  Magic state injection then allows us to complete the universal gate set.  See \cite{GC} for further discussion of how magic state constructions work.

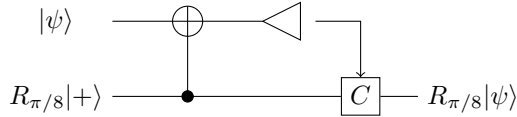
\begin{figure}
\begin{tikzpicture}[text height=1.5ex, text depth=.25ex]

\node (input) at (0,1) {$\ket{\psi}$};
\node (ancilla) at (0,0) {$R_{\pi/8} \ket{+}$};
\path (ancilla.east) coordinate (start);
\path (start) ++(0,1) coordinate (startup);
\path (start) ++(3.3,0) coordinate (conditional);

\draw (startup) -- ++(2,0) -- ++(0.5,0.25) -- ++(0,-0.5) -- ++(-0.5,0.25);
\draw [fill] (start) ++(1,0) circle (0.075cm);
\draw (start) ++(1,0) -- ++(0,1.2);
\draw (startup) ++(1,0) circle (0.2cm);

\node (gate) [shape=rectangle,draw] at (conditional) {$C$};
\draw (start) -- (gate.west);
\draw (gate.east) -- ++(0.5,0) node[anchor=west] {$R_{\pi/8} \ket{\psi}$};
\draw [->] (startup) ++(2.7,0) -| (gate.north);

\end{tikzpicture}
\caption{A magic state injection circuit to perform the $\pi/8$ phase rotation $R_{\pi/8} = \mathrm{diag}(e^{-i\pi/8}, e^{+i\pi/8})$ on an arbitrary state $\ket{\psi}$.  It involves an ancilla ``magic state'' $R_{\pi/8} \ket{+}$, a CNOT gate, a measurement, and a Clifford group gate $C = R_{\pi/8} X R_{\pi/8}^{-1}$, conditioned on the measurement outcome.  While $C$ is in the Clifford group, the gate $R_{\pi/8}$ itself is non-Clifford.  In a fault-tolerant scenario, all quantum states are encoded in a QECC and operations are performed with fault-tolerant gadgets.}
\label{fig:magicstate}
\end{figure}

We have not worked out our picture in detail for magic state constructions, but in this section, we will quickly sketch how we expect it to work.  We will focus on the case where the basic set of logical gates available is performed through transversal gates, extending Section~\ref{sec:transversal}.

\subsection{Ancillas}

First let us consider logical unitary gates on the encoded data which are performed via a physical unitary operation on the joint system of the data and ancillas encoded in the same QECC.  The ancilla serves as a scratch space to assist in the circuit and is discarded after the circuit is complete.  The resulting operation should still be a unitary operation on the logical state.  We do not have a concrete fault-tolerant protocol in mind for this section --- instead, the discussion here serves as a warm-up for the scenario we shall consider in Section \ref{sec:measurements}, where measurements are also involved.

A unitary construction using an ancilla will generally introduce the ancilla in a known state, perform gates interacting the computational data blocks with the ancilla blocks, and then discard the ancilla.  Since we are assuming for the moment that the final logical operation on the system is unitary and preserves the logical subspace, the ancilla will end up in a tensor product with the data blocks.  The final state of the ancilla cannot depend on the state of the data blocks, because then the data and ancilla blocks would in some cases be entangled and the logical operation would be non-unitary.
Therefore the final state of the ancilla depends \emph{only} on the gate being performed (which determines the initial state of the ancilla and the physical unitary done to it).

We have already noted in Section~\ref{sec:transversal} (see Figure~\ref{fig:transversalmultiple}) that transversal gates interacting multiple blocks of the QECC can be rephrased in the picture of a single data block by simply considering all the corresponding registers of the different blocks together as one qudit of a larger size.  For instance, for a transversal gate interacting two blocks of a qubit code, we can combine the first qubits of the two  blocks into a single dimension-$4$ qudit, and the second qubits of the two blocks into a second $4$-state qudit, and so on.  The transversal gate on the two blocks then becomes a single-block transversal gate, a tensor product of single-qudit operations acting on the individual $4$-state systems.  Therefore, we continue our discussion below assuming without loss of generality that we are dealing with a single code block of data and a single ancilla block and doing transversal gates between them.

Recall that in the version of our picture without ancillas, we describe a QECC, which is a subspace $C$ in the physical Hilbert space $\Hilbphy$, as an element of the Grassmannian $\Grass$.  When we add an ancilla into the picture, the Grassmannian changes from $\Grass$ to $\subGrass{K}{\hat{N}}$.  In particular, $C \in \Grass$ is replaced by $\hat{C} = C \otimes \ket{B} \in \subGrass{K}{\hat{N}}$.  Since $\ket{B}$ is a single state independent of the logical data, the subspace $\hat{C}$ has the same dimension $K$ as $C$.  However, we have increased the dimension of the physical Hilbert space from $N$ to $\hat{N} = NA$, where $A$ is the dimension of the Hilbert space containing the ancilla.

Including ancillas also requires us to modify $\M \subset \Grass$, the manifold of subspaces related to the QECC $C$ by transversal gates, to $\hat{\M} \subset \subGrass{K}{\hat{N}}$, the manifold related to $\hat{C}$ by transversal gates on larger qudits, including the ancilla system.  Let $\Hilb{N}$ be the Hilbert space of the data blocks and $\Hilb{A}$ be the Hilbert space of the ancilla.  Let $\hat{\F}$ be the set of transversal gates acting on data plus ancilla blocks, as discussed above, and let $\hat{C} = C \otimes \ket{A}$, where $C$ is the starting QECC code block and $\ket{A}$ is the starting state of the ancilla.  Then we define
\begin{equation}
\hat{\M} := \hat{\F}(\hat{C}) = \{F(\hat{C}) \;|\; F \in \hat{\F} \} \subset \subGrass{K}{\hat{N}}.
\label{eqn:M'}
\end{equation}

Note that $\M$ can be considered as a submanifold of $\hat{\M}$; in particular, $\M \otimes \ket{A} \subset \hat{\M}$ provides one such embedding.  Indeed, we can embed $\M$ into $\hat{\M}$ in many different ways, e.g., as $\M \otimes \ket{A'}$. It is possible that  $\M \otimes \ket{A'}$ is not a subset of $\hat{\M}$ for all $\ket{A'}$, but it will be if $C \otimes \ket{A'} = U(C \otimes \ket{A})$ for some $U \in \hat{\F}$.
We will consider $\M \otimes \ket{A}$ as the ``standard'' embedding of $\M$ in $\hat{\M}$.  See Figure \ref{fig:largermanifold}.

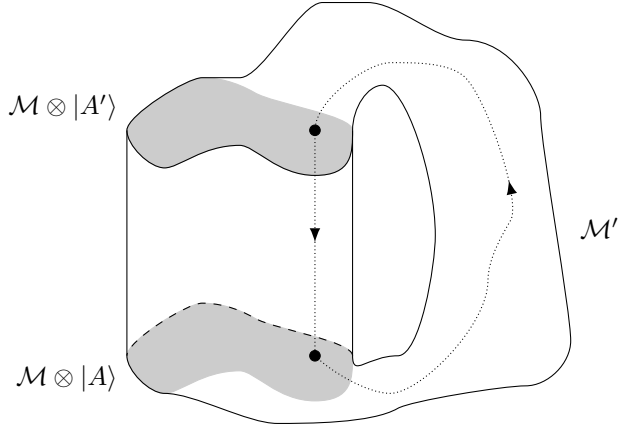
\begin{figure}
\begin{tikzpicture}[decoration={
	markings,
	mark=at position 0.4 with {\arrow[line width=0.3mm]{latex};},
	mark=at position 0.87 with {\arrow[line width=0.3mm]{latex};}}
	]

\path (2,-1) coordinate (Mleft);
\path (Mleft) ++(3,0) coordinate (Mright);
\path (Mleft) ++(0.5,-0.5) coordinate (Mbranch);
\path (Mleft) ++(0,3) coordinate (topleft);
\path (Mright) ++(0,3) coordinate (topright);
\path (topleft) ++(1.5,0.7) coordinate (topmid);
\path (topmid) ++(-0.5,0) coordinate (topmidleft);
\path (Mright) ++(2.9,1.7) coordinate (rightside);
\path (Mright) ++(-0.5,0) coordinate (pathstart);
\path (pathstart) ++(0,3) coordinate (pathend);

\draw [color=black!20, fill] (Mleft) .. controls +(0,-0.2) and +(-0.2,0) .. (Mbranch) .. controls +(0.2,0) and +(-0.5,0) .. ++(1,0.3) .. controls +(0.2,0) and +(-0.5,0) .. ++(1,-0.4) .. controls +(0.5,0) and +(0,-0.3) .. (Mright) .. controls +(0,0.2) and +(330:0.3) .. ++(-1.3,0.5) .. controls +(150:0.3) and +(0.2,0) .. ++(-0.7,0.2) .. controls +(-0.2,0) and +(0,0.2) .. (Mleft);

\draw [dashed] (Mright) .. controls +(0,0.2) and +(330:0.3) .. ++(-1.3,0.5) .. controls +(150:0.3) and +(0.2,0) .. ++(-0.7,0.2) .. controls +(-0.2,0) and +(0,0.2) .. (Mleft);

\draw [color=black!20, fill] (topleft) .. controls +(0,-0.2) and +(-0.2,0) .. ++(0.5,-0.5) .. controls +(0.2,0) and +(-0.5,0) .. ++(1,0.3) .. controls +(0.2,0) and +(-0.5,0) .. ++(1,-0.4) .. controls +(0.5,0) and +(0,-0.3) .. (topright) .. controls +(0,0.2) and +(330:0.3) .. ++(-1.3,0.5) .. controls +(150:0.3) and +(0.2,0) .. (topmidleft) .. controls +(-0.2,0) and +(0,0.2) .. (topleft);

\draw (topleft) .. controls +(0,-0.2) and +(-0.2,0) .. ++(0.5,-0.5) .. controls +(0.2,0) and +(-0.5,0) .. ++(1,0.3) .. controls +(0.2,0) and +(-0.5,0) .. ++(1,-0.4) .. controls +(0.5,0) and +(0,-0.3) .. (topright);

\draw (Mleft) -- (topleft);
\draw (Mright) -- (topright);

\draw (Mleft) .. controls +(0,-0.2) and +(-0.2,0) .. (Mbranch) .. controls +(0.4,0) and +(-0.4,0) .. ++(1.5,-0.4) .. controls +(0.4,0) and +(210:0.4) .. ++(1.6,0.2) .. controls +(30:0.2) and +(180:0.3) .. ++(1.7,0.3) .. controls +(0:0.3) and +(270:0.3) .. ++(0.6,0.6) .. controls +(90:0.4) and +(270:0.4) .. ++(-0.4,+3) .. controls +(90:0.3) and +(0:0.6) .. ++(-1,1) .. controls +(180:0.4) and +(0:0.4) .. ++(-1.3,0.5) -- ++(-0.6,0) .. controls +(180:0.3) and +(0:0.3) .. (topmid) -- (topmidleft) .. controls +(-0.2,0) and +(0,0.2) .. (topleft);

\draw (Mright) .. controls +(270:0.3) and +(180:0.4) .. ++(0.6,0) .. controls +(0:0.3) and +(270:0.4) .. ++(0.5,1.6) .. controls +(90:0.8) and +(0:0.3) .. ++ (-0.7,2) .. controls +(180:0.2) and +(90:0.3) .. (topright);

\draw [fill] (pathstart) circle (0.07cm);
\draw  [densely dotted, postaction={decorate}] (pathstart) .. controls +(315:0.3) and +(180:0.3) .. ++(1,-0.5) .. controls +(0:0.4) and +(270:0.6) .. ++(1.3,1.5) .. controls +(90:0.3) and +(240:0.2) .. ++(0.3,0.8) .. controls +(60:0.3) and +(315:0.4) .. ++(-0.6,1.9) .. controls +(135:0.3) and +(0:0.3) .. ++(-1,0.2) .. controls +(180:0.3) and +(90:0.4) .. (pathend) -- cycle;
\draw [fill] (pathend) circle (0.07cm);

\node [anchor=north east] at (Mleft) {$\M \otimes \ket{A}$};
\node [anchor=south east] at (topleft) {$\M \otimes \ket{A'}$};
\node [anchor=west] at (rightside) {$\M'$};

\end{tikzpicture}
\caption{Adding ancilla blocks embeds $\M$ into a larger manifold $\hat{\M}$, and gates that interact the computational blocks with the ancilla blocks produce a path (dotted line) in $\hat{\M}$ that leaves $\M$ and eventually returns to $\M$ when the ancilla blocks are reset to their initial value $\ket{A}$.  $\hat{\M}$ is defined via equation (\ref{eqn:M'}).}
\label{fig:largermanifold}
\end{figure}

We begin with the subspace $\hat{C} = C \otimes \ket{A} \in \M$, where $\M$ in this context means its standard embedding into $\hat{\M}$.  As we perform a unitary gate interacting $\Hilb{N}$ and $\Hilb{A}$, we follow a path which leaves $\M$ and eventually ends up at the subspace $C \otimes \ket{A'}$.  (Recall that, as discussed above, we cannot end up with any entanglement between the ancilla and code because we are doing a logical unitary operation.)  If $\ket{A'} = \ket{A}$, we have returned to $\M$ and performed a loop in $\hat{\M}$.  In this case, we can identify the logical gates implementable via transversal gates plus ancillas with the monodromy representation of $\pi_1(\hat{\M})$ as before.  However, note that the monodromy representation of $\pi_1 (\hat{\M})$ might be different than the monodromy representation of $\pi_1(\M)$, potentially giving us access to additional fault-tolerant logical gates.

If $\ket{A'} \neq \ket{A}$, however, we have not performed a loop even though there is a sensible way to equate the final subspace $C \otimes \ket{A'}$ with the initial subspace $C \otimes \ket{A}$.  If it is possible to reset the ancilla from $\ket{A'}$ to $\ket{A}$ using transversal gates acting only on the ancilla block, then we can use these gates to complete the loop.  However, it might not be possible to do so --- it may be that the only way to get from $\ket{A'}$ to $\ket{A}$ using transversal gates is to interact with the data block.  One potential approach to circumvent this problem is to work with an appropriate quotient space of $\hat{\M}$ which identifies subspaces $D \otimes \ket{A} \in \hat{\M}$ and $D \otimes \ket{A'} \in \hat{\M}$ for arbitrary $D \in \M$.  

Another straightforward way is to allow additional paths in $\hat{\M}$ to include non-transversal unitaries which change the ancilla block without touching the data block.  That is, we would add to $\hat{\M}$ all points of the form $D \otimes \ket{B}$, for any $D \in \M$, $\ket{B} \in \Hilb{A}$, and allow any unitary paths that map subspaces of the form $D \otimes \ket{B}$ only into other subspaces of the same form with the same $D$.  This procedure means extending the manifold $\hat{\M}$ and group $\hat{\F}$ further to $\hat{\M}'$ and $\hat{\F}'$.  Because travelling along one of the new paths does not alter the $D$ subspace, there is a natural way to trivialize the fibre along the path.
We can then create a loop by performing the transversal gate on $\Hilb{N} \otimes \Hilb{A}$ and then following one of the new paths from $C \otimes \ket{A}$ to $C \otimes \ket{A'}$, as in Figure~\ref{fig:largermanifold}.  The same argument as in Section~\ref{sec:transversal} should show that the connection is projectively flat and that the gate is non-trivial only if the loop is homotopically non-trivial in $\hat{\M}'$.

The additional paths represent ``resetting'' the ancilla register.  To follow such a path using physical operations might require many-qudit gates and would not necessarily be at all fault-tolerant, but we can still use them as a conceptual tool.

\subsection{Measurements}
\label{sec:measurements}

The strategy for gadgets which involve measurements is similar but more complex.  First, note that we can always replace a gadget involving measurement and classical control by one using unitary operations and ancillas: We purify the measurements by replacing each with a CNOT to an ancilla qudit that starts in the state $\ket{0}$.  We can maintain fault-tolerance while doing so~\cite{AB} by repeating the measurement multiple times and storing the results as a classical repetition code.  Since the classical repetition code allows transversal implementation of a universal set of \emph{classical} gates, we can perform the classical control of the quantum computer fault-tolerantly.  Classical error correction consists of non-transversal gates, but implemented in a way that limits the propagation of bit flip errors.  Note that the classical code does not protect against phase errors on the qudits which we have added to replace the classical measurement results, but because the original circuit assumes those bits are classical anyway, phase errors on them have no effect on the main computation qudits.

The complication for our picture arises because the set of allowed unitaries $\hat{\F}$ for data block plus ancillas consists of transversal quantum gates on the ``quantum'' registers in the system and a different set of possible gates on the ``classical'' registers in the computer.  A fault-tolerant classical computation by itself probably would not result in a flat projective connection, because generating classical gates continuously via a Hamiltonian requires leaving the set of allowed classical states and generally involves intermediate states which are sensitive to phase errors.  However, since we have restricted attention to magic state constructions where phase errors in the classical registers have no effect on the actual encoded states, we believe that applying our construction to magic state injection gadgets will lead to a flat projective connection.

Summarizing the points in this section, we believe the following procedure allows us to construct a bundle representing a magic state construction, and gives a flat projective connection on that bundle.

\begin{enumerate}
\item Purify measurements by adding additional ancillas.  Replace each measurement by a CNOT.
\item Store formerly classical bits in a classical repetition code and implement classical gates transversally on this code.
\item Define an enlarged set of allowed unitaries $\hat{\F}$ by transversal gates on quantum registers, continuous generation of classical gates on classical registers, and reset operations on ancillas that do not touch the original computational qubits.  
\item Enlarge the subset $\M \subset \Grass$ to $\hat{\M} \subset \subGrass{K}{N'}$ based on the ancilla Hilbert space and allowed unitaries $\hat{\F}$.
\item Consider loops that start at the original code in $\M$, leave $\M$ to explore $\hat{\M}$ via transversal gates, purified measurements, and classical gates, and then return to $\M$ by resetting ancillas.
\end{enumerate}


\section{Conclusion}
\label{sec:conclusion}

We have shown that fault-tolerance via transversal gates and from braiding defects in the toric code or similar topological codes can both be phrased in the same general mathematical framework.  In both cases, the fault tolerance of the protocol reveals itself when the connection on $\M \subset \Grass$ turns out to be projectively flat.  Intuitively, it makes sense that this would be a general feature of fault-tolerant protocols: When $U(t)$ represents a possible unitary evolution in the protocol, small errors in the gate implementation lead to $\tilde{U}(t)$, a path close to $U(t)$.  The protocol is fault-tolerant precisely when $U(t)$ performs the same logical gate as $\tilde{U}(t)$ --- i.e., both paths induce the same parallel transport in bundles over $\M$.  This, in turn, suggests that the connections on $\bigvecbund$ and $\bigPbund$ should be projectively flat.  This is the motivation for Conjecture~\ref{conj:flatconnection}, which states that projective flatness is a necessary condition for fault tolerance.

Indeed, this criterion seems to capture the heart of fault tolerance: faulty implementations of gate gadgets should have the same effect on encoded qudits as an ideal implementation would.
It is therefore tempting to go further and guess that projective flatness is a \emph{sufficient} condition as well:
\begin{Conjecture}
A protocol is fault-tolerant if there is a set of paths $\Fpath$ and a topological space $\Mgraph$ such that the natural pre-connection defines a flat projective partial connection on $\Mgraph$.
\label{conj:FTprotocol}
\end{Conjecture}

To prove either Conjecture~\ref{conj:flatconnection} or Conjecture~\ref{conj:FTprotocol} would require a rigorous and general definition of what it means for a protocol to be fault-tolerant.  Right now, we can state rigorous conditions for fault tolerance for specific error models, such as the $s$-qudit error model, and we have a general intuitive sense of what fault tolerance means for arbitrary error models, but we do not have a rigorous formulation of fault tolerance for general error models.  Having one would be valuable, not only for the program begun in this paper, but also for other potential methods of proving general theorems about fault tolerance.

Once we have a general definition of fault tolerance, it may be possible to complete our program and prove Conjectures~\ref{conj:flatconnection} and \ref{conj:FTprotocol}.  To do so, we will need a canonical way to choose $\M$ given $\Mgraph$.  The main element that is missing from our picture currently is the error model, so we suspect that the error model will go into the definition of $\M$.

Even without such a general framework, there are a variety of more specific open questions that are suggested by our work.  Can the ideas sketched out in Section~\ref{sec:fullFT} be fleshed out, or is there some subtle issue that arises?  What, precisely, is the relationship between the manifold $\M$ for single-block transversal gates and the manifold for multiple-block transversal gates --- is there a direct relationship between their fundamental groups?  Is there a way to calculate the monodromy representation of $\pi_1(\M)$ directly, and therefore learn something about fault-tolerance from the topology of $\M$?

%
%
%
%
%
%

\subsection{Acknowledgements}

We would like to thank various people for helpful conversations and feedback, particularly John Cortese, Steve Flammia, Mike Freedman, Joel Kamnitzer, Alexei Kitaev, John Preskill, and Spiros Michalakis.  Research at Perimeter Institute is supported by the Government of Canada through Industry Canada and by the Province of Ontario through the Ministry of Economic Development \& Innovation. D.~G.~is supported by the CIFAR Quantum Information Processing program.  L.~Z.~is supported by Joel Kamnitzer's ERA (Early Researcher Award) grant.  This paper was written in part while D.~G. was a guest at the Newton Institute for Mathematical Sciences.

\end{document}